\documentclass[letterpaper,twocolumn,10pt]{article}
\usepackage{usenix2019_v3}

\usepackage{amsthm}
\usepackage[english]{babel}
\usepackage{ifthen}
\usepackage{xcolor}
\usepackage{blindtext}
\usepackage{algorithm}
\usepackage{graphicx}
\usepackage{amsmath}
\usepackage[noend]{algpseudocode}
\usepackage{subfig}
\usepackage{xspace}
\usepackage{amssymb} 
\usepackage{balance}

\usepackage{graphicx}
\usepackage{footnote}
\usepackage{comment}
\usepackage{array}
\usepackage{enumitem,kantlipsum}
\makesavenoteenv{tabular}
\makesavenoteenv{table}
\usepackage{booktabs} 
\usepackage{multicol}
\usepackage{multirow}
\usepackage{color, colortbl}
\usepackage{framed, newfloat}

\DeclareCaptionType{List}

\definecolor{shadecolor}{rgb}{.9, .9, .9} 
  \newenvironment{myframe}{%
  \MakeFramed {\hsize\linewidth\advance\hsize-\width \FrameRestore}}%
  {\endMakeFramed}

\usepackage{url}

\newcommand{\exclude}[1]{}
\newcommand{\showComments}{yes}
\newcommand{\note}[2]{
    \ifthenelse{\equal{\showComments}{yes}}{\textcolor{#1}{#2}}{}
}

\makeatletter
\let\OldStatex\Statex
\renewcommand{\Statex}[1][3]{%
  \setlength\@tempdima{\algorithmicindent}%
  \OldStatex\hskip\dimexpr#1\@tempdima\relax}
\makeatother
\makeatletter
\newcommand*{\inlineequation}[2][]{%
  \begingroup
    \refstepcounter{equation}%
    \ifx\\#1\\%
    \else
      \label{#1}%
    \fi
    \relpenalty=10000 %
    \binoppenalty=10000 %
    \ensuremath{%
      #2%
    }%
    ~\@eqnnum
  \endgroup
}
\makeatother
\algnewcommand{\LeftComment}[1]{\OldStatex \(\triangleright\) #1}
\algnewcommand{\LineComment}[1]{\OldStatex \(\triangleright\) #1}
\newcommand{\psass}{\ensuremath{\mathbin{{=}}\ }}
\newcommand{\addeq}{\ensuremath{\mathbin{{+}{=}}\ }}
\newcommand{\subeq}{\ensuremath{\mathbin{{-}{=}}\ }}
\newcommand{\muleq}{\ensuremath{\mathbin{{\times}{=}}\ }}

\newcommand{\eqeq}{\ensuremath{\mathbin{{=}{=}}\ }}

\newcommand{\name}{{\sc{TopoOpt}}\xspace}

\newcommand{\algo}{{\sc{TopologyFinder}}\xspace}
\newcommand{\algobf}{\bfseries{\scshape{TopologyFinder}}\xspace}
\newcommand{\fattree}{{Fat-tree}\xspace}
\newcommand{\fattrees}{{Fat-trees}\xspace}
\newcommand{\LBE}{{\fattree}\xspace}
\newcommand{\SBE}{{Ideal Switch}\xspace}
\newcommand{\OBE}{{Oversub. \fattree}\xspace }
\newcommand{\namebf}{\bfseries{\scshape{TopoOpt}}\xspace}
\newcommand{\para}[1]{{\textbf{{#1}}}}
\newcommand{\net}{{Meta}\xspace}
\newcommand{\MP}{{MP}\xspace}
\newcommand{\reduce}{{AllReduce}\xspace}
\newcommand{\fancyperms}{{TotientPerms}\xspace}
\newcommand{\fancypermsnospace}{{TotientPerms}}
\newcommand{\franksim}{\textit{FlexNetPacket}\xspace}
\newcommand{\fnet}{\textit{FlexNet}\xspace}

\newcommand{\captionvspace}{0em}
\algnewcommand{\IOComment}[1]{\OldStatex \(\triangleright\) #1}
\algnewcommand{\firstLeftComment}[1]{\OldStatex \(\indent\triangleright\) #1}
\algnewcommand{\secondLeftComment}[1]{\OldStatex \(\indent\indent\triangleright\) #1}
\algnewcommand{\thirdLeftComment}[1]{\OldStatex \(\indent\indent\indent\triangleright\) #1}
\definecolor{LightCyan}{rgb}{0.88,1,1}
\definecolor{celadon}{rgb}{0.67, 0.88, 0.69}

\date{}
\usepackage{lipsum} 
\usepackage{enumitem}
\setlist{nosep} 

\begin{document}

\title{\Large \bf \namebf: Co-optimizing Network Topology and \\ Parallelization Strategy for Distributed Training Jobs}

\author{
Weiyang Wang$^*$ \quad Moein Khazraee$^*$ \quad Zhizhen Zhong$^*$ \qquad Manya Ghobadi$^*$  \\
Zhihao Jia$^{\dagger,\ddagger}$ \quad Dheevatsa Mudigere$^\dagger$ \quad Ying Zhang$^\dagger$ \quad  Anthony Kewitsch$^\S$\\\\
 $^*$Massachusetts Institute of Technology\quad
$^\dagger$Meta \quad
$^\ddagger$CMU \quad
$^\S$Telescent \\
 }

\setcounter{page}{1}
\maketitle
\begin{abstract}
We propose \name, a novel direct-connect fabric for deep neural network (DNN) training workloads. \name co-optimizes the distributed training process across three dimensions: computation, communication, and network topology. We demonstrate the mutability of \reduce traffic, and leverage this property to construct efficient network topologies for DNN training jobs. \name then uses an alternating optimization technique and a group theory-inspired algorithm called \fancyperms to find the best network topology and routing plan, together with a parallelization strategy. 
We build a fully functional 12-node direct-connect prototype with remote direct memory access (RDMA) forwarding at 100~Gbps. Large-scale simulations on real distributed training models show that, compared to similar-cost \fattree interconnects, \name reduces DNN training time by up to 3.4$\times$. 
\end{abstract}

\section{Introduction}
\label{sec:intro}

Our society is rapidly becoming reliant on deep neural networks (DNNs). New datasets and models are invented frequently, increasing the memory and computational requirements for training. This explosive growth has created an urgent demand for efficient distributed DNN training systems. 

Today's DNN training systems are built on top of traditional datacenter clusters, with electrical packet switches arranged in a multi-tier \fattree topology~\cite{fat-tree}. \fattree topologies are traffic oblivious fabrics, allowing uniform bandwidth and latency between server pairs. They are ideal when the workload is unpredictable and consists mostly of short transfers--two inherent properties of legacy datacenter workloads~\cite{pfabric, dc_benson, dctcp, resource_disaggregation, projector}. But \fattree networks are becoming a bottleneck for distributed DNN training workloads~\cite{pipedream, flex_flow, gpipe, blueconnect, goyal, wang2020blink, mudigere2021highperformance}.

Previous work has addressed this challenge by reducing the size of parameters to transmit through the network~\cite{horovod, ako, qsgd, lin2017deep, imagenet_in_four, blueconnect, tictac, firecaffe, goyal, pipedream, dcscale_ml, parameter_propagation} and developing techniques to discover faster parallelization strategies while considering the available network bandwidth~\cite{flex_flow, pipedream, placeto, efficient_comp_comm_neurips, qsgd}. These proposals co-optimize computation and communication as two important dimensions of distributed DNN training, but they do not consider the \textit{physical layer topology} as an optimization dimension.

In this paper, we analyze DNN training jobs from production clusters of Meta. We demonstrate that training workloads do not satisfy common assumptions about datacenter traffic that underlie the design of \fattree interconnects. Specifically, we show that ($i$) the communication overhead of large DNN training jobs increases dramatically as we increase the number of workers; and ($ii$) the traffic pattern of a DNN training job depends on its parallelization strategies.

Motivated by these observations, we propose \name, a direct-connect DNN training system that co-optimizes network topology and parallelization strategy. \name creates dedicated partitions for each training job using reconfigurable optical switches and patch panels, `and jointly optimizes the topology and parallelization strategy within each partition. To achieve our goal, we grapple with the \emph{algorithmic} challenges of finding the best topology, such as how to navigate the large search space across computation, communication, and topology dimensions, and also with various \emph{operational} challenges, such as which optical switching technologies match well with the traffic patterns of DNN models.

We cast the topology and parallelization strategy co-optimization problem as an off-line alternating optimization framework. Our optimization technique \textit{alternates} between optimizing the parallelization strategy and optimizing the network topology. It searches over the parallelization strategy space assuming a fixed topology, and feeds the traffic demand to a \algo algorithm. The updated topology is then fed back into the parallelization strategy search algorithm. This alternating process repeats until the system converges to an optimized parallelization strategy and topology. 

We demonstrate that finding an optimized network topology for DNNs is challenging because the ideal network topology needs to meet two goals simultaneously: ($i$) to complete large \reduce transfers efficiently, and ($ii$) to ensure a small hop-count for Model Parallel transfers. To meet these goals, we propose a novel \textit{group theory-based technique}, called \fancyperms, that exploits the \textit{mutability} of \reduce transfers. Our \fancyperms approach builds a series of \textit{\reduce permutations} that  not only carry \reduce transfers efficiently, but are also well-positioned to carry Model Parallel transfers and, hence, improve the overall training performance. 

Optical circuit-switched networks traditionally support point-to-point traffic across hosts with direct circuits between them. As a result, for a given set of circuits, only directly connected hosts can communicate leaving the rest of the hosts wait for new circuits to be established. To support arbitrary communication across all hosts participating in a job, we enable \name's hosts to act as relays and forward the traffic that does not belong to them. Host-based forwarding introduces a new challenge for RDMA flows since RDMA NICs drop packets that do not belong to them. To enable host-based RDMA forwarding, we exploit the network partition (NPAR) function of modern NICs, creating an efficient logical overlay network for RDMA packet forwarding (\S\ref{sec:prototype}).

To evaluate \name, we build a 12-server prototype with NVIDIA A100 GPUs~\cite{a100}, 100~Gbps NICs and a Telescent reconfigurable optical patch panel~\cite{telescent}. Moreover, we integrate our \fancyperms\ \reduce permutations into NCCL and enable it to load-balance parameter synchronization across multiple ring-\reduce sub-topologies. Our evaluations with six representative DNN models (DLRM~\cite{dlrm}, CANDLE~\cite{candle_uno}, BERT~\cite{transformer}, NCF~\cite{ncf}, ResNet50~\cite{resnet}, and VGG~\cite{vgg}) show that \name reduces the training iteration time by up to 3.4$\times$ compared to a similar-cost \fattree. Moreover, we demonstrate that \name is, on average, 3.2$\times$ cheaper than an ideal full bisection bandwidth \fattree. \name is the first system that co-optimizes topology and parallelization strategy for ML workloads and is currently being evaluated for deployment at  Meta. The source code and scripts of \name are available at \url{https://topoopt.csail.mit.edu}.

\section{Motivation}
\label{sec:measurement}

Prior research has illustrated that \textit{demand-aware} network fabrics are flexible and cost-efficient solutions for building efficient datacenter-scale networks~\cite{helios,mordia, projector}. However, predicting the upcoming traffic distribution is challenging in a traditional datacenter setting. This section demonstrates that DNN training workloads present a unique opportunity for demand-aware networks, as the jobs are long-lasting, and the traffic distribution can be \textit{pre-computed} before the jobs start to run. First, we provide the necessary background to understand distributed DNN training and introduce three types of data dependencies between accelerator nodes in training jobs (\S\ref{sec:types_of_data_transfers}). Then, we present measurements from production clusters in \net (\S\ref{sec:production_measurement}), and discuss the important properties of DNN training traffic. 

\subsection{Background on Distributed DNN training}
\label{sec:types_of_data_transfers}

\para{Training iteration.} A common approach to training DNNs is stochastic gradient descent (SGD)~\cite{sgd}. Each SGD \textit{iteration} involves selecting a random batch of labeled training data, computing the error of the model with respect to the labeled data, and calculating gradients for the model's weights through backpropagation. The SGD algorithm seeks to update the model weights so that the next evaluation reduces the error~\cite{brownlee2018better}. Training iterations are repeated with new batch of data until the model converges to the target accuracy. 

\para{Data parallelism.} Data parallelism is a popular parallelization strategy, whereby a batch of training samples is distributed across training accelerators. Each accelerator holds a replica of the DNN model and executes the forward and backpropagation steps locally. 
In data parallelism, all accelerators synchronize their model weights during each training iteration. This step is commonly referred to as \emph{\reduce} and can be performed using various techniques, such as broadcasting~\cite{sufficient_broadcast}, parameter servers~\cite{osdi_parameter_server}, ring-\reduce~\cite{thakur, baidu,imagenet_in_four}, tree-\reduce~\cite{decision_tree}, or hierarchical ring-\reduce~\cite{10.1177/1094342005051521, 8752949}.

\para{Hybrid data and model parallelism.} 
Large DNN models cannot fit in the  memory of a single accelerator or even a single server with multiple accelerators. As a result, the model needs to be divided across multiple accelerators using {\em model parallelism}~\cite{NIPS2014, pmlr-v80-jia18a}. 
Moreover, pure data parallelism is becoming suboptimal for large training jobs because of the increasing cost of synchronizing model parameters across accelerators~\cite{flex_flow, naumov2020deep, bert_blog, shoeybi2020megatronlm, dlrm, huang2019gpipe}. As a result, large DNNs are distributed using a hybrid of data and model parallelism, where different parts of a DNN and its dataset are processed on different accelerators in parallel. 

\para{Types of data dependencies in DNN training.} Each training iteration includes two major types of \textit{data dependencies}. Type (1) refers to \textit{activations} and \textit{gradients} computed during the Forward and Backpropagation steps. This data dependency is required for each input sample. Type (2) refers to synchronizing the {\em model weights} across accelerators through the \reduce step once a batch of samples is processed. Depending on the parallelization strategy, these data dependencies may result in local memory accesses or cross-accelerator traffic. For instance, in a hybrid data and model parallelization strategy, type (1) and (2) both result in cross-accelerator traffic, depending on how the model is distributed across accelerators. Given that type (1) is related to model parallelism, we refer to the network traffic created by type (1) as \textit{\MP transfers}. Similarly, we refer to the network traffic created by type (2) as \textit{\reduce transfers}. Note that \reduce transfers do not strictly mean data parallelism traffic, as model parallelism can also create \reduce transfers across a subset of nodes.

\begin{figure}[t]
\centering
\includegraphics[width=0.45\textwidth]{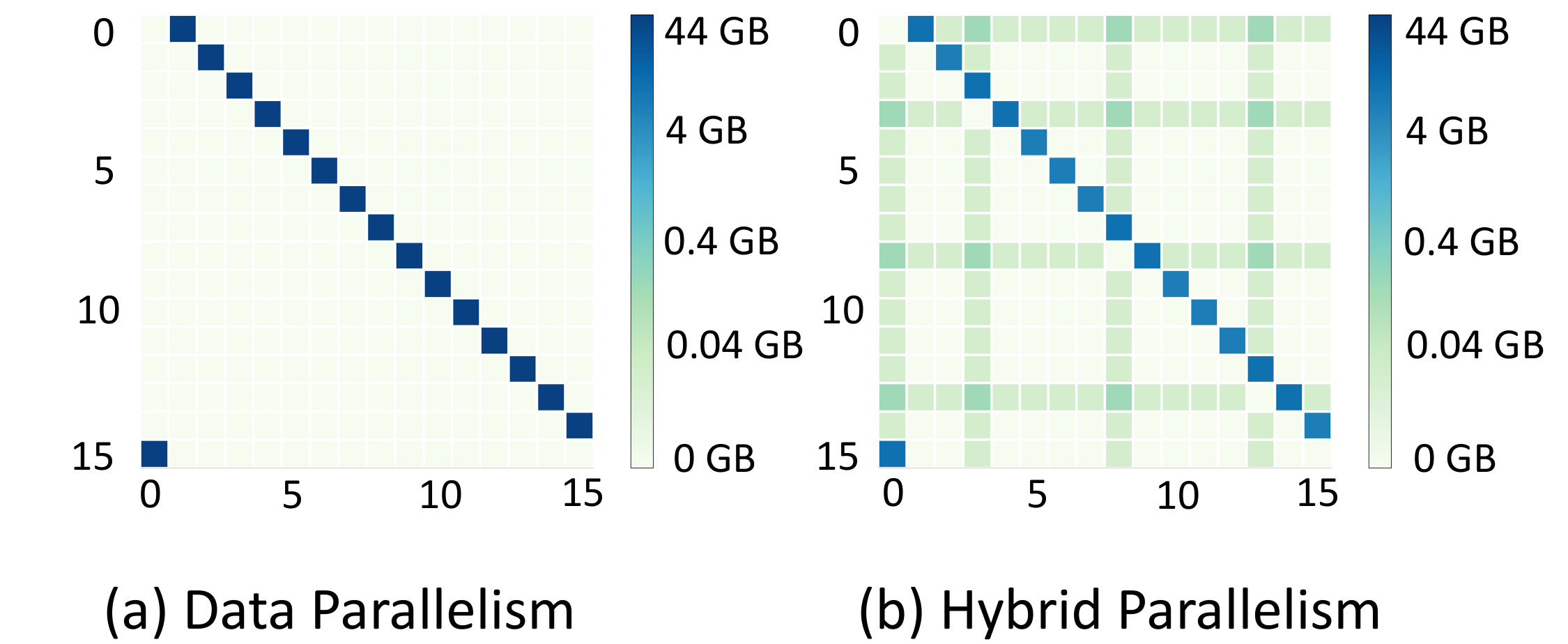}
\caption{DLRM traffic heatmaps for different parallelization strategies.}
\label{fig:dlrm_traffic_pattern_dpmp}
\end{figure}

\para{Example: DLRM traffic pattern.} 
Deep Learning Recommendation Models (DLRMs) are a family of personalization and recommendation models based on embedding table lookups that capitalize on categorical user data~\cite{naumov2019deep}. DLRMs are large, typically with 100s of billions of parameters, primarily because of their large embedding tables. Using pure data parallelism to distribute a DLRM results in massive \reduce transfers. For instance, consider a DLRM architecture with four embedding tables $E_0, \cdots, E_3$, each with embedding dimensions of 512 columns and $10^7$ rows (total size 22~GB for the model) distributed across 16 servers $S_0, \cdots, S_{15}$ with data parallelism. We compute the resulting traffic distribution, and Figure~\ref{fig:dlrm_traffic_pattern_dpmp}a illustrates the traffic pattern for a single training iteration. The rows and columns indicate source and destination servers, while the color encodes the amount of traffic between server pairs. The heatmap shows that using ring-\reduce for synchronization, a pure data parallelism strategy results in 44~GB of \reduce transfers. 

Hence, a common parallelization strategy for DLRMs is to use a hybrid of data and model parallelism where the embedding tables are divided across nodes, while the rest of the model is replicated on all nodes~\cite{mudigere2021highperformance}. Following the parallelization strategy used at \net, we place $E_0$ on $S_0$, $E_1$ on $S_3$, $E_2$ on $S_8$, and $E_3$ on $S_{13}$, and replicate the rest of the model on all servers. This parallelization strategy creates a mix of \MP and \reduce traffic, shown in Figure~\ref{fig:dlrm_traffic_pattern_dpmp}b. It reduces the maximum transfer size from 44~GB to 4~GB. 

Note that \MP transfers in DLRM form one-to-many broadcast and many-to-one incast patterns to transfer the activation and gradients to all nodes because the servers handling embedding tables must communicate with \textit{all other servers}. In this example, the size of each \reduce transfer is 4~GB, whereas the size of \MP transfers is 32~MB, as shown by darker green elements in the heatmap.

\subsection{Production Measurements}
\label{sec:production_measurement}

We study traffic traces from hundreds of production DNN training jobs running on multiple clusters at \net. We instrument each job to log its training duration, number of workers, and the total amount of data transferred across its workers during training.

\para{Number of workers and job duration.} Figure~\ref{fig:fblearner_job_profile}a shows the cumulative distribution function (CDF) of the number of workers for different models in \net's clusters. Most jobs are distributed across 32 to 700 workers, agreeing with recent announcements by other major players in the industry~\cite{acun2020understanding, bert_blog}, where each worker is a single GPU. Figure~\ref{fig:fblearner_job_profile}b demonstrates the CDF of total training job duration; as the figure shows, most jobs last over 10 hours. In fact, the top 10\% of jobs take more than 96 hours (four days) to finish. This measurement shows production DNN jobs at Meta are long-lasting, and take up to weeks to finish.

\para{Network overhead.} Figure~\ref{fig:network_bottleneck} illustrates the percentage of network overhead as the number of GPUs is increased from 8 to 128 for six DNN jobs in production. We use RDMA to transmit packets between servers and measure the percentage of time consumed by communication during training as network overhead. The figure shows that as the number of GPUs increases, the network quickly takes up a significant portion of training iteration time. In fact, the network overhead accounts for up to 60\% of a DNN training iteration time in Meta's production environment. Similar observations have been made in prior work~\cite{pipedream, sip-ml, gpipe, horovod, dcscale_ml, commlatencyimc20}. Such bottleneck suggests the existing datacenter networks are insufficient for the emerging DNN training workloads.

\begin{figure}[t]
\vspace{\captionvspace}
\captionsetup[subfigure]{justification=centering}
\centering
\subfloat[Number of workers] {
\includegraphics[width=0.49\columnwidth]{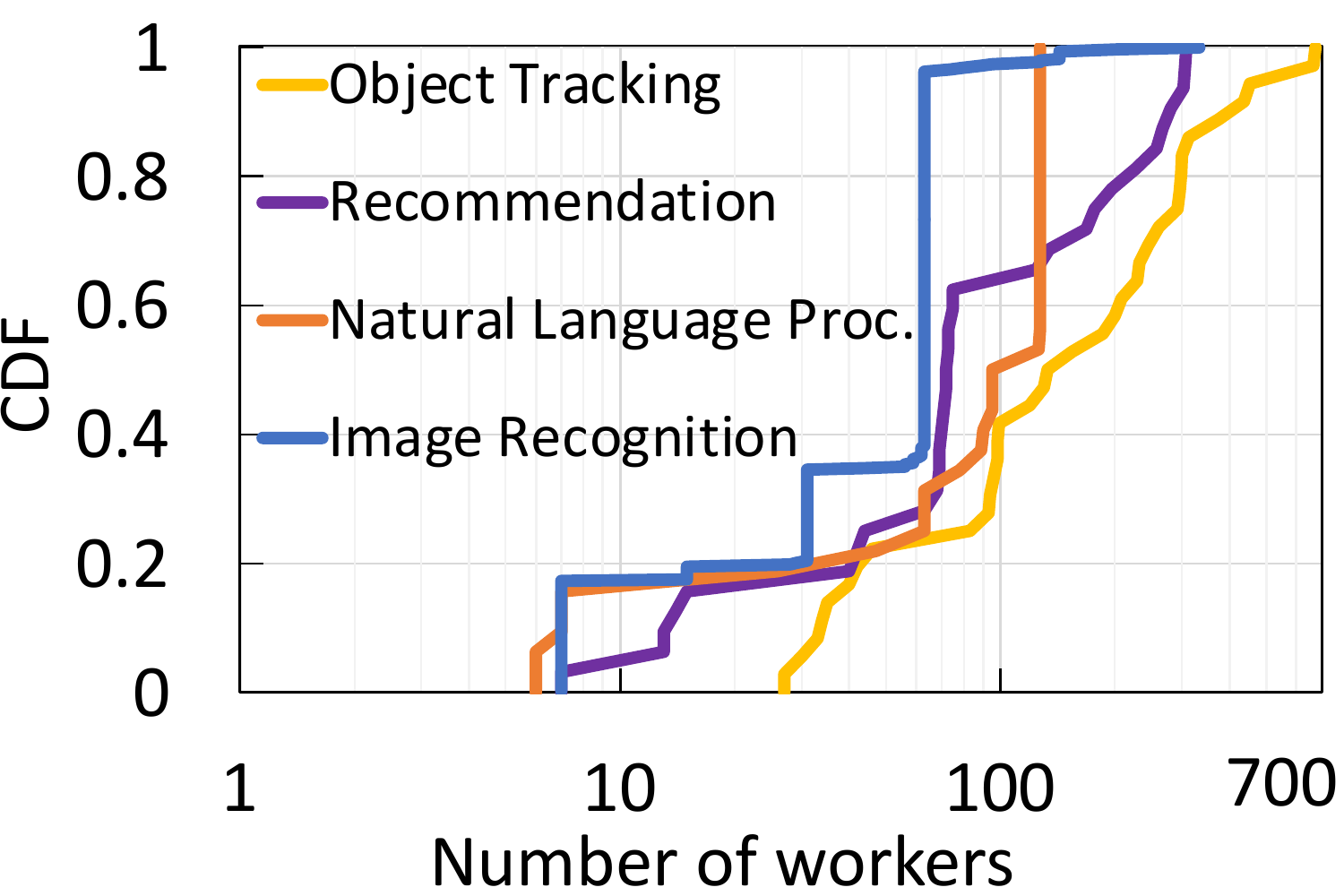}
\label{fig:fblearner_numgpus}
}
\subfloat[Training job duration] {
\includegraphics[width=0.49\columnwidth]{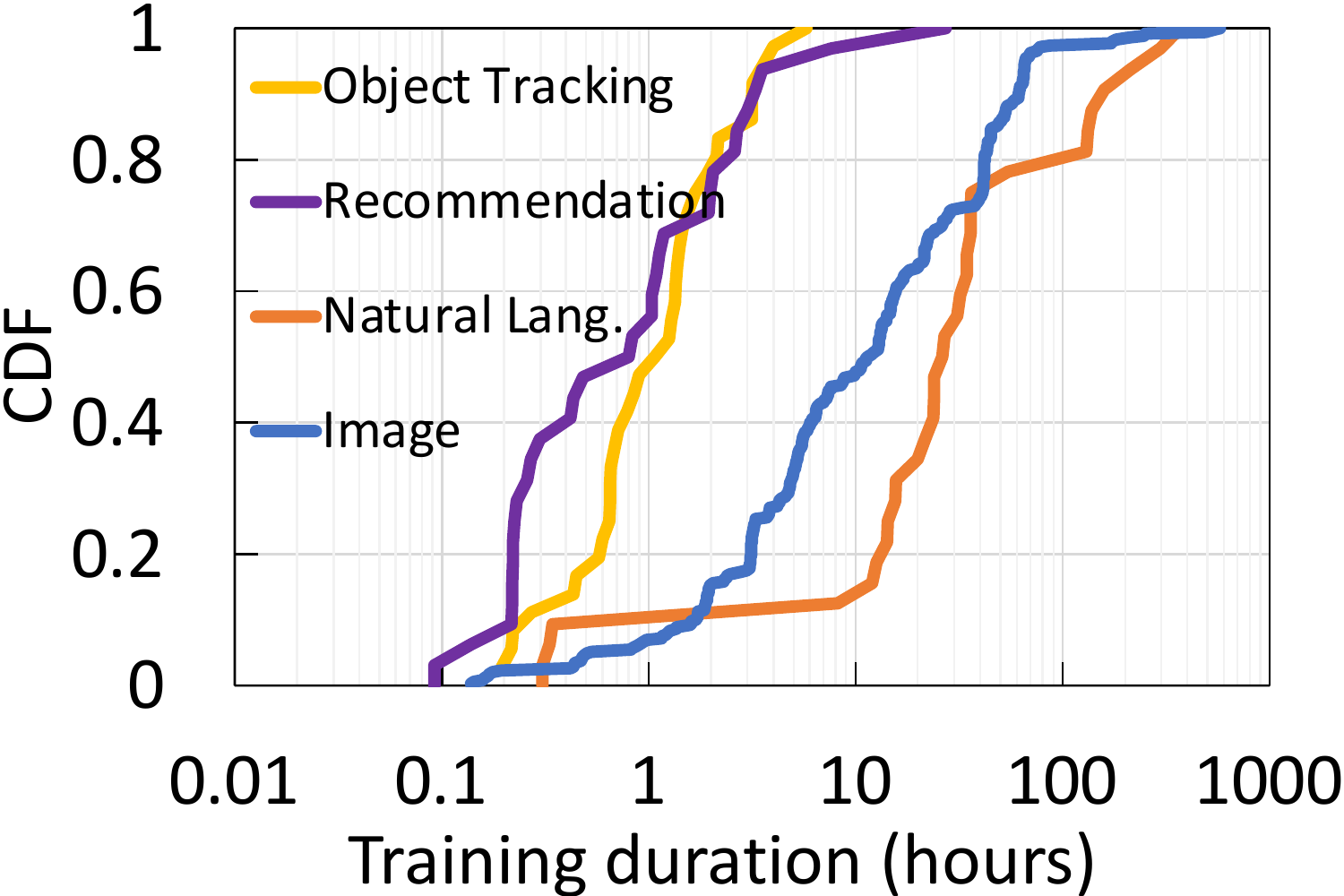}
\label{fig:fblearner_duration}
}
\caption{Profiling distributed DNN training jobs in \net.}
\label{fig:fblearner_job_profile}
\end{figure}

\begin{figure*}[t]
\centering
\begin{minipage}{0.24\textwidth}
\includegraphics[width=\textwidth]{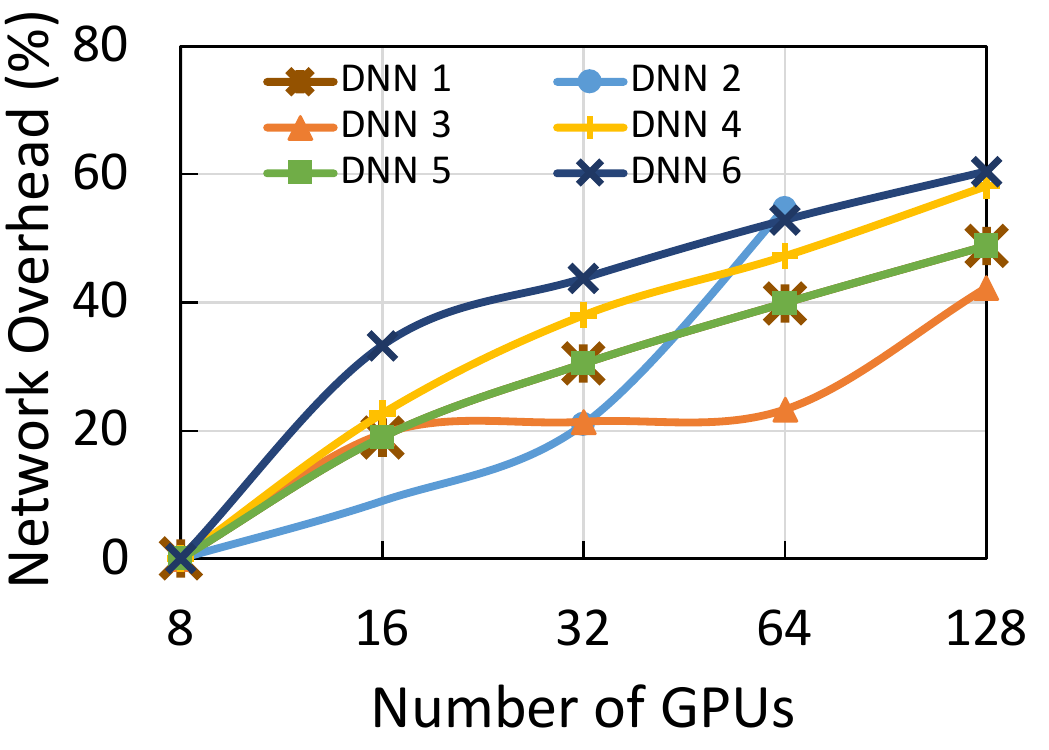}
 \vspace{\captionvspace}
\caption{Network overhead measurements in \net.} 
 \vspace{\captionvspace}
\label{fig:network_bottleneck}
\end{minipage}
\hspace{0.7cm}
\begin{minipage}{0.6\textwidth}
\includegraphics[width=\textwidth]{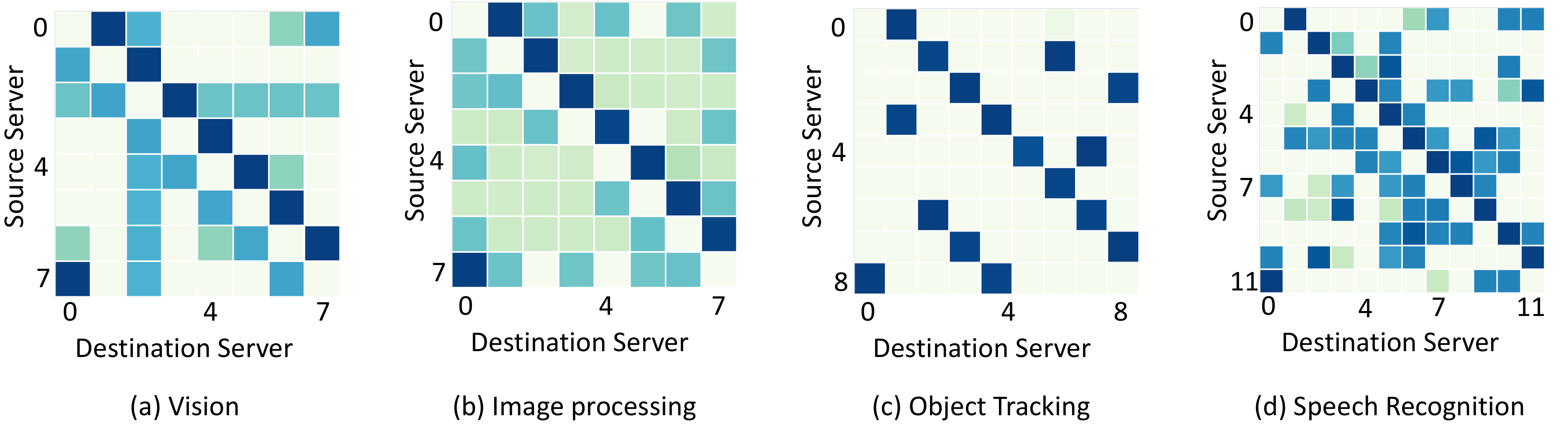}
\vspace{\captionvspace}
\caption{Traffic heatmaps of production jobs in \net.}
 \vspace{\captionvspace}
\label{fig:production_heatmap}
\end{minipage}
\end{figure*}

\para{Traffic heatmaps.} Figure~\ref{fig:production_heatmap} shows the heatmap of server-to-server traffic for four training jobs running in Meta's production GPU clusters. The values on the colormap and the exact names of DNN models are not shown for confidentiality reasons. All heatmaps in the figure contain diagonal squares (in dark blue), indicating a ring communication pattern between servers. This is expected, as ring-\reduce is the common \reduce communication collective at \net. But the \MP transfers (light blue and green squares) are \textit{model-dependent} because \MP transfers depend on the parallelization strategy and device placement of a training job. Moreover, we find that the traffic patterns of training jobs do not change between iterations \textit{for the entire training duration}, resulting in the same per-iteration heatmap throughout the training. Once a training job starts, the same parallelization strategy and synchronization method are used across training iterations, resulting in a periodic and predictable traffic pattern. Similar observations have been made in previous work~\cite{gandiva}. In particular, the traffic heatmap is identical \textit{across} training iterations. Note that the traffic pattern changes \textit{within} a training iteration during forward, backward, and  \reduce phases.  
\section{\namebf System Design}
\label{sec:design}

The observations in the previous section suggest that demand-aware fabrics are excellent candidates for a DNN training cluster. In this section, we seek to answer the following question: \textit{``Can we build a demand-aware network to best support distributed training?"} To answer this question, we propose \name, a novel system based on optical devices that jointly optimizes DNN parallelization strategy and topology to accelerate today's training jobs.

\para{\namebf interconnect.} A \name cluster is a \textit{shardable} direct-connect fabric where each server has $d$ interfaces connected to a core layer of $d$ optical switches, as shown in Figure~\ref{fig:flat_conn}. The optical switches enable \name to shard the cluster into dedicated partitions for each training job. The size of each shard depends on the number of servers the job requests. Given a DNN training job and a set of servers, \name first finds the best parallelization strategy and topology between the servers off-line (\S\ref{sec:alt_opt}). Then, it reconfigures the optical switches to realize the target topology for the job. Appendix~\ref{app:sharding} provides details on how \name achieves sharding and dynamic job arrivals in shared clusters.

\para{Degree of each server.} We denote the number of interfaces on each server (i.e., the degree of the server) by $d$. Typically, $d$ is the same as the number of NICs installed on the server. In cases where the number of NICs is limited, the degree can be increased using NICs that support break-out cables or the next generation of co-packaged optical NVLinks~\cite{optical_nvlink}. In our testbed, we use one 100~Gbps HPE NIC~\cite{hpe-nic} with 4$\times$25~Gbps interfaces to build a system with degree four ($d=4$). 

\para{Direct-connect topology.} In \name, optical switches connect the servers directly, forming a \textit{direct-connect topology}. To further scale a \name cluster, we create a hierarchical interconnect by placing the servers under Top-of-Rack (ToR) switches and connecting ToR switches to the optical layer, creating a direct-connect topology at the ToR or spine layers, similar to previous work~\cite{sirius, rotornet, firefly, bcube, jupiter_evo}.

\para{Host-based forwarding.} In DNN training workloads, the degree of each server is typically smaller than the total number of neighbors with whom the server communicates during training. To ensure traffic is not blocked when there is no direct link between two servers, we use a technique called \textit{host-based forwarding}, where hosts act as switches and forward incoming traffic toward the destination. Previous work used similar technique at the ToR switch level~\cite{rotornet, opera, sirius}.

\begin{figure}[t]
\centering
\includegraphics[width=0.8\columnwidth]{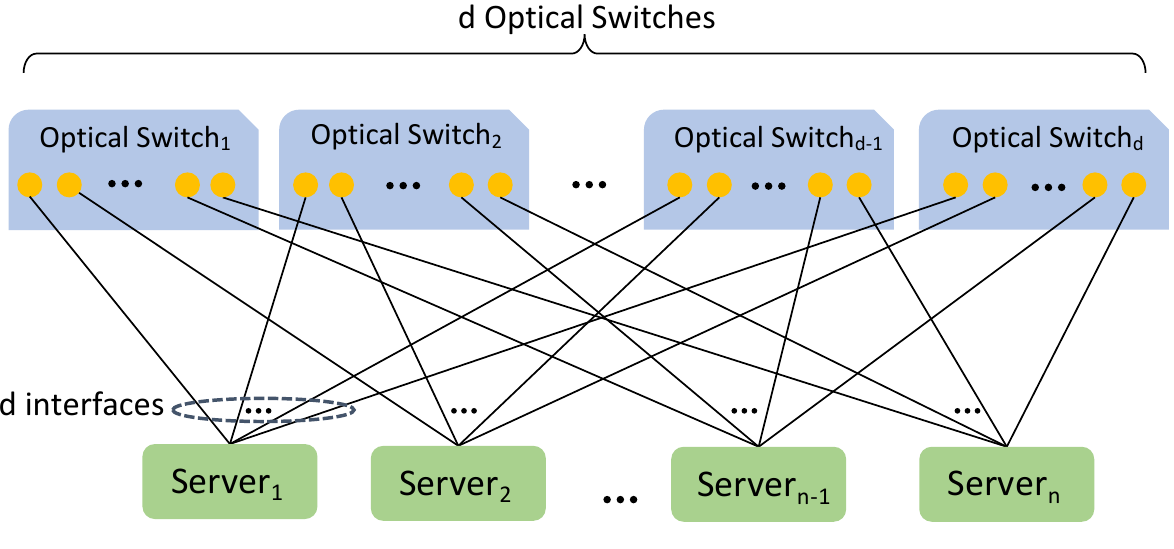}
\caption{Illustration of \name's interconnect.}
\label{fig:flat_conn}
\end{figure}

\para{Optical switching technologies.} A wide range of optical switches is suitable for \name, including commodity available optical patch panels~\cite{telescent} and 3D-MEMS~\cite{calient, polatis}, as well as futuristic designs such as Mordia~\cite{mordia}, MegaSwitch~\cite{megaswitch}, and Sirius~\cite{sirius_ecoc, sirius}. Table~\ref{tab:optical_technologies} lists the characteristics of these devices. \name is compatible with any of these technologies. Appendix~\ref{sec:patch_panel_details} provides details about these devices.

\para{One-shot reconfiguration.} Patch panel and OCS are both applicable for an immediate deployment of \name, as shown in Table~\ref{tab:optical_technologies}. The choice of which technology to use depends on several factors, including scale of the cluster, iteration time of jobs, and frequency of job arrivals. For instance, OCSs can potentially be used to reconfigure the topology of a job \textit{within} training iterations, whereas patch panels are only suitable when the topology remains intact throughout the entire training of a particular job. Our evaluations demonstrate that the reconfiguration latency of today's OCSs is too high for today's DNNs, leading to sub-optimal performance when the topology is reconfigured within iterations (\S\ref{sec:sims}). As a result, given that faster technologies are not yet available, \name uses a one-shot reconfiguration technique based on an offline co-optimization framework (\S\ref{sec:topoopt_algorithms}) that jointly optimizes the parallelization strategy and topology. \name then reconfigures the interconnection between training servers of each job before the job starts and keeps the topology intact until the training is complete (or to recover from failures).
 \section{Co-optimizing Parallelization Strategy and Network Topology}
 \label{sec:topoopt_algorithms}
 
 \begin{table}[tb!]
   \begin{center}
   \scriptsize
   \centering
   \resizebox{\linewidth}{!}{%
   \setlength{\tabcolsep}{4pt} 
     \begin{tabular}{|p{2.7cm}|p{0.5cm}|p{0.9cm}|p{1cm}|p{1.6cm}|}
     \hline
     \textbf{Technology} & \textbf{Port-count} & \textbf{Reconfig. latency} & \textbf{Insertion Loss (dB)}& \textbf{Cost /port}  \\ \hline
      \rowcolor{celadon}
      Optical Patch Panels~\cite{telescent} & 1008 & minutes & 0.5& \$100 \\ \hline
      \rowcolor{celadon}
      3D MEMS~\cite{calient, polatis} & 384 & 10~ms & 1.5--2.7& \$520 \\ \hline
      2D MEMS~\cite{mordia, megaswitch} & 300 & 11.5~$\mu$s & 10--20& Not commercial \\ \hline
      Silicon Photonics \cite{sip-ml, wu_switch} & 256 & 900~ns & 3.7& Not commercial \\ \hline
      Tunable Lasers~\cite{sirius_ecoc, sirius} & 128 & 3.8~ns & 7-13& Not commercial \\ \hline
      RotorNet~\cite{rotornet, opera} & 64 & 10~$\mu$s & 2& Not commercial \\ \hline
     \end{tabular}
   }
   \end{center}
     \vspace{\captionvspace}
     \caption{Comparison of optical switching technologies.}
     \vspace{\captionvspace}
  \label{tab:optical_technologies}
 \end{table}
 
 This section describes \name's co-optimization framework for finding a network topology and parallelization strategy for a given DNN training job. 
 
 \subsection{Alternating Optimization}
 \label{sec:alt_opt}
 
 \para{The search space is too large.} Finding the optimal parallelization strategy alone is an NP-complete problem~\cite{flex_flow}, and adding network topology and routing makes the problem even harder. An extreme solution is to jointly optimize compute, communication, and topology dimensions using a \textit{cross-layer optimization formulation}. Theoretically, this approach finds the optimal solution, but the search space quickly explodes, even at modest scales (e.g., six nodes~\cite{efficient_comp_comm_neurips}).

 \para{Naive approach.} The other extreme is to optimize the network topology \textit{sequentially after} the parallelization strategy has been found. While this approach is able to reconfigure the network to better match its traffic demand, the eventual combination of topology and parallelization strategy is likely to be sub-optimal in the global configuration space. 
 
 \para{\namebf's approach: alternating optimization.} In \name, we seek to combine the best of both worlds. To make the problem tractable, we divide the search space into two planes: $Comp.\times Comm.$ and $Comm.\times Topo.$ We use an alternating optimization technique to iteratively search in one plane while keeping the result of the other plane constant. Figure~\ref{fig:joint_optimization} illustrates our alternating optimization framework. We use FlexFlow's MCMC (Markov Chain Monte Carlo) search algorithm~\cite{flex_flow} to find the best parallelization strategy for a given network topology while considering the communication cost. If the parallelization strategy improves the training iteration time, we feed it to the $Comm. \times Topo.$ plane to find the efficient network topology and routing using our \algo algorithm. 
 The discovered topology is then fed back into the $Comp. \times Comm.$ plane, which further optimizes the parallelization strategy and device placement based on the new topology.
 This optimization loop repeats until convergence or after $k$ iterations, where $k$ is a configurable hyper-parameter.

 \begin{figure}[t]
 \centering
 \includegraphics[width=0.9\columnwidth]{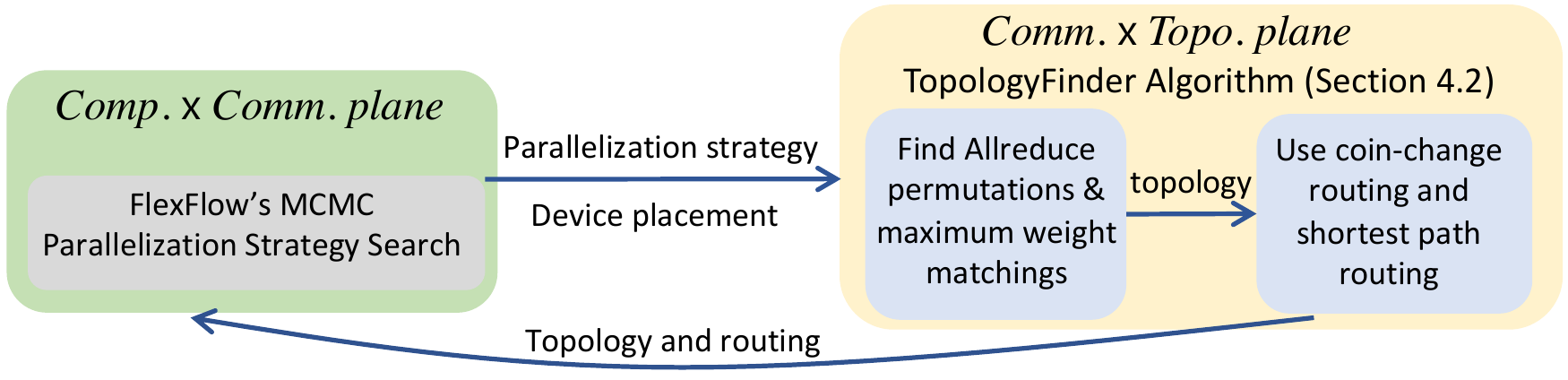}
 \caption{\name searches for the best parallelization strategy, jointly with routing, and topology.}
 \label{fig:joint_optimization}
 \vspace{\captionvspace}
 \end{figure}

 \subsection{\algobf Algorithm}
 
 \para{\algobf steps.} Algorithm \ref{alg:topo_finder} presents the pseudocode of our \algo procedure. The algorithm takes the following inputs: $n$ dedicated servers for the training job, each with degree $d$, as well a list of \reduce and \MP transfers ($T_{\reduce}$ and $T_{\MP}$) based on  the parallelization strategy and device placement obtained from the $Comp. \times Comm.$ plane. The algorithm then finds the best topology ($G$) and routing rules ($R$) and returns them  to the $Comp. \times Comm.$ plane for the next round of alternating optimization. Our algorithm consists of the following four steps.

 \para{Step 1: Distribute the degree.} This step distributes the degree $d$ between \reduce and \MP sub-topologies proportionally, based on their share of total traffic. We specifically start with \reduce transfers and allocate at least one degree to the \reduce sub-topology  to ensure the network remains connected (line~\ref{algo:d_allreduce}). The remaining degrees, if any, are allocated to the \MP sub-topology (line~\ref{algo:d_MP}). 
 
 \begin{algorithm}[t]
 \scriptsize
 \begin{algorithmic}[1]
     \Procedure{\algo}{$n$, $d$, $T_{\reduce}$, $T_{\MP}$}
         \IOComment{\textbf{Input} $n$: Number of dedicated training servers for the job.}
         \IOComment{\textbf{Input} $d$: Degree of each server.}
         \IOComment{\textbf{Input} $T_{\reduce}$: \reduce transfers.}
         \IOComment{\textbf{Input} $T_{\MP}$: \MP transfers.} 
         \IOComment{\textbf{Output} $G$: Topology to give back to the $Comp. \times Comm.$ plane.}
         \IOComment{\textbf{Output} $R$: Routing rules to give back to the $Comp. \times Comm.$ plane.}
         \textcolor{cyan}{\OldStatex \ \ \ \ \ \ \(\triangleright\) \textit{Distribute degree $d$ between \reduce and \MP sub-topologies}}
         \State $d_{A}$ \psass max(1, $\lceil d \times \frac{sum(T_{reduce})}{sum(T_{reduce})+sum(T_{\MP})} \rceil$) \label{algo:d_allreduce}
         \State $d_{\MP} = d - d_{A}$ \label{algo:d_MP}
         \textcolor{cyan}{\OldStatex \ \ \ \ \ \ \(\triangleright\) \textit{Construct the \reduce sub-topology $G_{\reduce}$}}
          \State $G_{\reduce}$ \psass $\{\}$ \label{algo:g_allreduce}
         \For {each \reduce group $k$ with set of transfers $T_k$} \label{algo:g_all_reduce_for}
                 \textcolor{cyan}{\OldStatex \hspace{0.3cm}\ \ \ \ \ \ \ \ \(\triangleright\) \textit{Assign degree $d_k$ to group $k$ according to its total traffic}}
             \State $d_k$ \psass $\lceil d_{A} \times \frac{sum(T_{k})}{sum(T_{reduce})} \rceil$ \label{algo:g_all_reduce_degree}
             \State $d_{A}$ \psass $d_{A} - d_k$
                 \textcolor{cyan}{\OldStatex \hspace{0.3cm}\ \ \ \ \ \ \ \ \(\triangleright\) \textit{Find all the permutations between servers in group $k$}}
             \State P$_k$ =  \texttt{\fancypermsnospace($n$, $k$)} 
             \textcolor{gray}{\hspace{0.05cm} \(\triangleright\) \textit{(Details in \S\ref{sec:mutability_reducetopo})}}
             \label{algo:find_permutations}
                 \textcolor{cyan}{\OldStatex \hspace{0.3cm}\ \ \ \ \ \ \ \ \(\triangleright\) \textit{Select $d_k$ permutations from $P_k$}}
             \State $G_{\reduce} = G_{\reduce} \cup$ \texttt{SelectPermutations($n$, $d_k$, $P_k$)}  	        \textcolor{gray}{\hspace{0.05cm} \(\triangleright\) \textit{(\S\ref{sec:mutability_reducetopo})}}\label{algo:select_top_permutations}
             \If {$d_\reduce \eqeq 0$}
             \State break 
             \EndIf
     \EndFor
         \textcolor{cyan}{\OldStatex \ \ \ \ \ \ \(\triangleright\) \textit{Construct the \MP sub-topology $G_{\MP}$}}  
         \State $G_{\MP}$ \psass $\{\}$ \label{algo:g_MP}
         \For {$i: i < d_{\MP}$}  \label{algo:g_MP_loop}
                 \textcolor{cyan}{\OldStatex \hspace{0.3cm}\ \ \ \ \ \ \ \ \(\triangleright\) \textit{Find a maximum weight matching according to $T_{\MP}$}}
             \State $g$ \psass \texttt{BlossomMaximumWeightMatching($T_{\MP}$)} \label{algo:maximum_matching}
             \State $G_{\MP} = G_{\MP} \cup g$ \label{algo:union_g_matching}
                 \textcolor{cyan}{\OldStatex \hspace{0.3cm}\ \ \ \ \ \ \ \ \(\triangleright\) \textit{Reduce the amount of demand for each link $l$ in graph $g$}}
             \For {$l \in g$} \label{algo:link_loop}
                 \State $T_{\MP}[l]$ \psass $T_{\MP}[l] /$ 2  \label{algo:diminishing_return}
             \EndFor
         \EndFor
         \textcolor{cyan}{\OldStatex \ \ \ \ \ \ \(\triangleright\) \textit{Combine the \reduce and \MP topologies}}        
         \State $G$ \psass $G_{\reduce} \cup G_{\MP}$ \label{algo:g_mp_reduce_union}
         \textcolor{cyan}{\OldStatex \ \ \ \ \ \ \(\triangleright\) \textit{Compute routes on $G_{\reduce}$ using the coin change algorithm~\cite{coin_change}}}        
         \State $R$ \psass \texttt{CoinChangeMod($n$, $G_{\reduce}$)} \textcolor{gray}{\hspace{0.05cm} \(\triangleright\) \textit{(Appendix~\S\ref{app:algo_details})}} \label{algo:coin_change}\label{algo:g_reduce_end}
         \textcolor{cyan}{\OldStatex \ \ \ \ \ \ \(\triangleright\) \textit{Compute routes on $G_{\MP}$ with shortest path}}        
         \State $R$ \addeq \texttt{ShortestPath($G$, $T_{\MP}$)} \label{algo:shortest_path}
         \State \Return $G, R$
     \EndProcedure
 \end{algorithmic}
 \caption{\algo pseudocode \label{alg:topo_finder}}
 \end{algorithm}

 \para{Step 2: Construct the \reduce sub-topology.} To find the \reduce sub-topology, the algorithm iterates over every \reduce group $k$ and allocates degree $d_k$ to each group proportionally based on the amount of traffic (line~\ref{algo:g_all_reduce_degree}). Note that in hybrid data and model parallelism strategies, the \reduce step can be performed across a subset of servers when a DNN layer is replicated across a few servers instead of all servers. To efficiency serve both \reduce and \MP transfers, \name constructs the \reduce sub-topology such that the diameter of the cluster is minimized. Section~\ref{sec:mutability_reducetopo} explains two algorithms, called \texttt{\fancyperms} (line~\ref{algo:find_permutations}) and \texttt{SelectPermutations} (line~\ref{algo:select_top_permutations}) to construct the \reduce sub-topology.

 \para{Step 3: Construct the \MP sub-topology.} 
 We use the Blossom maximum weight matching algorithm~\cite{edmonds_1965} to find the best connectivity between servers with \MP transfers (line~\ref{algo:maximum_matching}). We repeat the matching algorithm until we run out of degrees. To increase the likelihood of more diverse connectivity across server pairs, we divide the magnitude of $T_{\MP}$ for pairs that already have an edge between them by two (line~\ref{algo:diminishing_return}). In general, division by two can be replaced by a more sophisticated function with a diminishing return.

 \para{Step 4: Final topology and routing.} Finally, we combine the \MP and \reduce sub-topologies to obtain the final topology (line~\ref{algo:g_mp_reduce_union}). We then use a modified version of the coin-change algorithm~\cite{coin_change} (details in Appendix \ref{app:algo_details}) to route \reduce on the \reduce sub-topology (line~\ref{algo:coin_change}). Further, we use k-shortest path routing for the \MP transfers to take advantage of the final combined topology  (line~\ref{algo:shortest_path}).

 \subsection{Traffic Mutability and \reduce Topology}
 \label{sec:mutability_reducetopo}

 \para{Finding an efficient \reduce sub-topology.} At first blush, finding an \reduce sub-topology for a given DNN seems straightforward: we just need to translate the parallelization strategy and device placement from the $Comp. \times Comm.$ plane into a traffic matrix and map the traffic matrix into circuit schedules. Several papers have used this technique for datacenter networks~\cite{solstice, helios, mordia, reactor, cthrough, firefly, projector, sip-ml, megaswitch, quartz}. However, the conventional wisdom in prior work is to allocate as many direct parallel links as possible to \textit{elephant flows} and leave \textit{mice flows} to take multiple hops across the network. In principle, this approach works well for datacenters but it leads to sub-optimal topologies for distributed DNN training. While the size of \reduce transfers is larger than \MP transfers, \MP transfers have a higher communication degree than \reduce (Appendix~\ref{sec:dnn_details}). Hence, the conventional approach creates parallel direct links for carrying \reduce traffic and forces \MP flows to have a large hop-count, thereby degrading the training performance. 
 
 \para{\namebf's novel technique.} In \name, we seek to meet two goals simultaneously: ($i$) allocate ample bandwidth for \reduce transfers, as the bulk of the traffic belongs to them, but ($ii$) ensure a small hop-count for \MP transfers. We meet both goals by demonstrating a unique property of DNN training traffic -- the \reduce traffic is \textit{mutable}.
 
 \begin{figure}[t]
 \centering
 \includegraphics[width=\columnwidth]{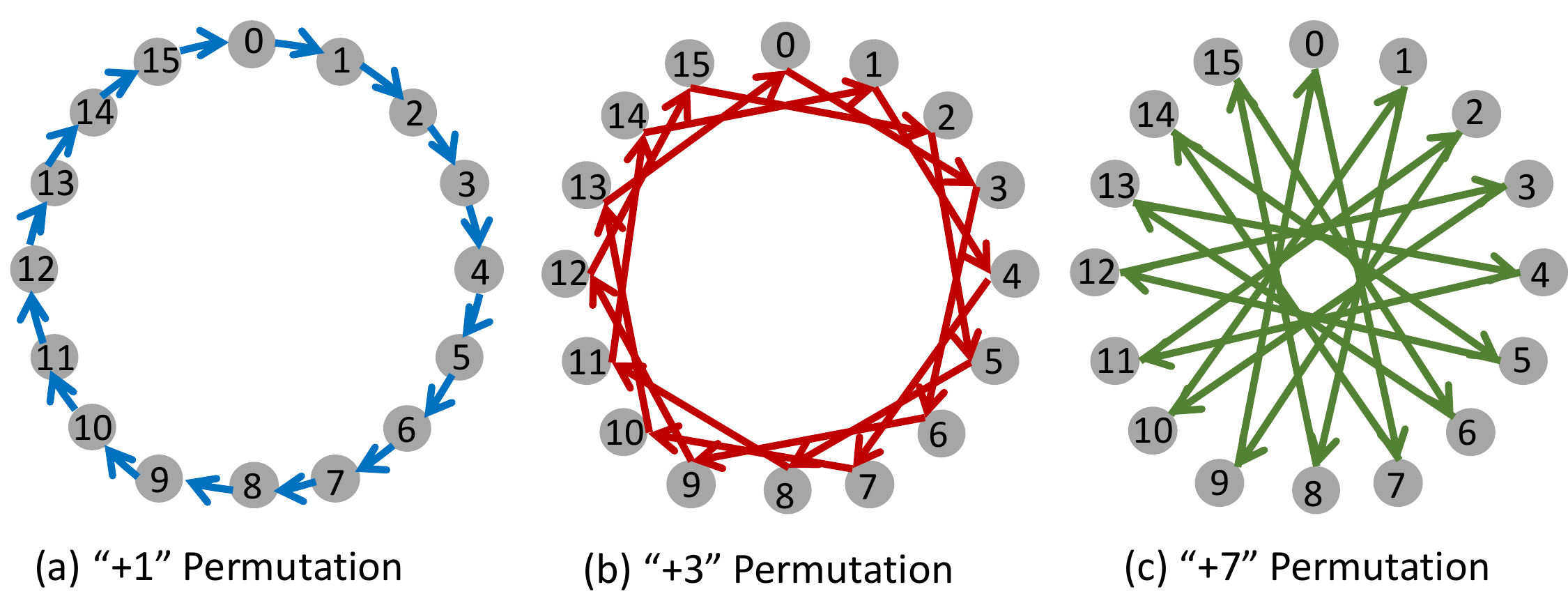}
  \vspace{-0.5cm}
 \caption{Ring-\reduce permutations.} 
 \label{fig:ring_permutations}
 \end{figure}
 \begin{figure}[t]
 \includegraphics[width=\columnwidth]{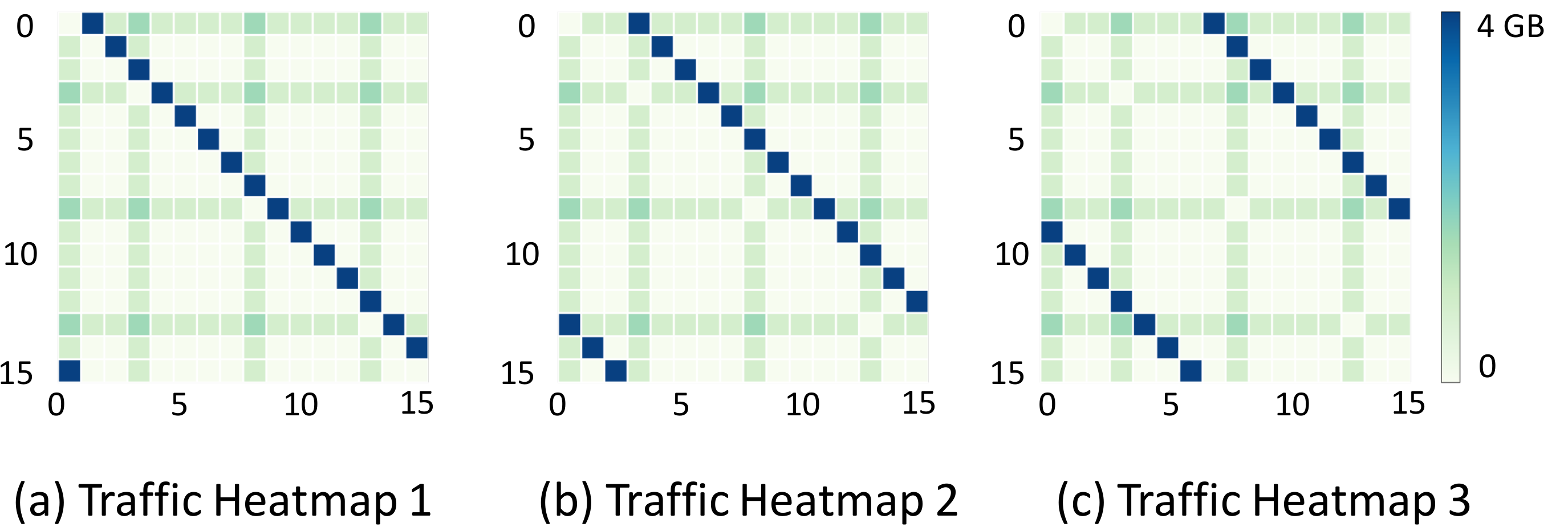}
 \vspace{-0.5cm}
 \caption{DLRM traffic heatmaps.}
 \label{fig:dlrm_traffic_pattern}
 \end{figure}
 
 \para{Mutability of \reduce transfers.} 
 We define traffic mutability as the ability to change the traffic pattern without altering parallelization strategy or device placement while maintaining correctness, and demonstrate that \reduce transfers are \textit{mutable} whereas \MP transfers are not. Intuitively, this is because \MP traffic is composed of network flows among nodes that contain \textit{different} parts of a DNN model thus creating immutable data dependencies, while \reduce transfers contain network flows among nodes that handle the \textit{same} part of the model, providing flexibility in the order of nodes participating in \reduce.  For instance, consider a DLRM distributed across 16 servers each with three NICs. The common \reduce pattern is shown as a  ring with consecutive node IDs, as shown in Figure~\ref{fig:ring_permutations}a. However, \textit{this is not the only possible permutation}. Each heatmap in \ref{fig:dlrm_traffic_pattern}a, \ref{fig:dlrm_traffic_pattern}b, and \ref{fig:dlrm_traffic_pattern}c corresponds to a different ring-\reduce permutation, shown in Figures~\ref{fig:ring_permutations}a, \ref{fig:ring_permutations}b, and \ref{fig:ring_permutations}c.  We denote each of these permutations as $+p$, where server $S_i$ connects to server $S_{(i+p)\%n}$, and $n$ is the number of servers, as shown in Figure~\ref{fig:ring_permutations}. Although all three heatmaps correspond to the \textit{exact same parallelization strategy and device placement}, the blue diagonal lines appear at different parts of the heatmaps, depending on the order of servers in the ring-\reduce permutation. But \MP transfers (green vertical and horizontal lines in each heatmap) are dictated by the parallelization strategy and device placement; thus, they remain at the same spot in all three heatmaps. 
 
 \para{Leveraging \reduce traffic mutability.} Traffic mutability implies that if a group of servers is connected in a certain order, simply permuting the label of the servers gives another ordering that will finish the \reduce operation with the same latency while potentially providing a smaller hop-count for \MP transfers. Instead of selecting just one \reduce order, \name finds multiple permutations for each \reduce group and overlaps their corresponding sub-topologies. In doing so, \name efficiently serves the \reduce traffic while decreasing the hop-count for \MP transfers.

 \begin{algorithm}[t]
 \footnotesize
 \begin{algorithmic}[1]
     \Procedure{\fancyperms}{$n$, $k$}
         \IOComment{\textbf{Input} $n$: Total number of nodes}
         \IOComment{\textbf{Input} $k$: \reduce group size}
         \IOComment{\textbf{Output} $P_k$: Set of permutations for \reduce group of size $k$}
         \textcolor{cyan}{\OldStatex \ \ \ \ \ \ \(\triangleright\) \textit{Initially, $P_k$ is empty}}
         \State $P_k$ \psass $\{\}$
         \textcolor{cyan}{\OldStatex \ \ \ \ \ \ \ \ \ \(\triangleright\) \textit{This loop runs $\phi(p)$ times, where}}
         \textcolor{cyan}{\OldStatex \ \ \ \ \ \ \ \ \ \(\triangleright\) \textit{$\phi$ is the Euler Totient function, $\phi(p)=|\{k<p:gcd(k,p)=1\}|$}}
         \textcolor{cyan}{\OldStatex \ \ \ \ \ \ \ \ \ \(\triangleright\) \textit{one can also restrict $p$ to be prime only}}
         \For {$p \leq k,\ gcd(p, k) \eqeq 1$} \label{alg:totient_for}
             \State $one\_perm$ \psass []
             \For {$i \text{ in } 0 \text{ to } n/k$}
                 \State $one\_perm$ \addeq $[i + j \times p\ \mathbf{for}\ j\ \text{ in } 0 \text{ to } k]$
             \EndFor
             \State $P_k$ \addeq $one\_perm$
         \EndFor
         \State \Return $P_k$
     \EndProcedure
 \end{algorithmic}
 \caption{\texttt{\fancyperms} pseudocode \label{alg:fancyperms}}
 \end{algorithm}

 \begin{algorithm}[t]
 \footnotesize
 \begin{algorithmic}[1]
     \Procedure{SelectPermutations}{$n$, $d_k$, $P_k$}
         \IOComment{\textbf{Input} $n$: Total number of nodes}
         \IOComment{\textbf{Input} $d_k$: Degree allocated for group this \reduce group of size $k$}
         \IOComment{\textbf{Input} $P_k$: Candidate permutations for this \reduce group of size $k$}
         \IOComment{\textbf{Output} $G_k$: Parameter synchronization topology, given as a set of permutations}
         \textcolor{cyan}{\OldStatex \ \ \ \ \ \ \(\triangleright\) \textit{Initially, $G_k$ is empty}}
         \State $G_k$ \psass $\{\}$
         \textcolor{cyan}{\OldStatex \ \ \ \ \ \ \(\triangleright\) \textit{$q$ now is the minimum candidate in $P_k$}}
         \State $q$ \psass $P_k[0]$
         \textcolor{cyan}{\OldStatex \ \ \ \ \ \ \(\triangleright\) \textit{\texttt{GetConn}($q$) gives the connection  described }}
         \textcolor{cyan}{\OldStatex \ \ \ \ \ \ \(\triangleright\) \textit{by the permutation corresponding to $q$}}
         \State $G_k$ \psass $G_k\cup$\texttt{GetConn}($q$)
         \textcolor{cyan}{\OldStatex \ \ \ \ \ \ \(\triangleright\) \textit{Ratio of the geometric sequence to fit}}
         \State $x$ \psass $\sqrt[d_k]{n}$ \label{alg:selperm_geo}
         \For {$i \in \{1, \cdots, d_k-1\}$} 
             \textcolor{cyan}{\OldStatex \ \ \ \ \ \ \ \ \ \ \ \ \ \ \ \ \(\triangleright\) \textit{Select the next candidate based on the ratio}}
             \State $q'$ \psass $x\times q$ 
             \textcolor{cyan}{\OldStatex \ \ \ \ \ \ \ \ \ \ \ \ \ \ \ \ \(\triangleright\) \textit{Project $q'$ onto $P_k\setminus G_k$ with minimal distance (L1-norm) }} \label{alg:selperm_l1}
             \State $q'$ \psass $\text{argmin}_{r\in P_k \setminus G_k}|r-q'|$ 
             \textcolor{cyan}{\OldStatex \ \ \ \ \ \ \ \ \ \ \ \ \ \ \ \ \(\triangleright\) \textit{Add this candidate to final topology}}
             \State $G_k$ \psass $G_k\cup$\texttt{GetConn}($q'$)
             \State $q$ \psass $q'$
         \EndFor
         \State \Return $G_k$
     \EndProcedure
 \end{algorithmic}
 \caption{\texttt{SelectPermutations} pseudocode \label{alg:SelectPermutations}}
 \end{algorithm}
 
 \para{\fancyperms algorithm.} While overlapping multiple permutations sounds straightforward, navigating through the set of all possible \reduce orderings is non-trivial, since the number of possible permutations is $O(n!)$. To reduce the search space of all possible permutations, we design the \fancyperms algorithm to find the ring generation rule for all \textit{regular rings}, based on group theory. Regular rings are those where the distance between indices of consecutive servers is equal; i.e., server $S_i$ is connected to server $S_{(i+p)\%n}$ for some $p$. Algorithm~\ref{alg:fancyperms} presents the pseudocode of \texttt{\fancyperms}. Inspired by Euler's totient function~\cite{totientfn}, we find all integer numbers $p < n$, where $p$ is co-prime with $n$ (i.e. $gcd(p ,n) = 1$, line~\ref{alg:totient_for}, Algorithm~\ref{alg:fancyperms}), represent a valid ring-\reduce permutation (\S\ref{app:algo_details}). For instance, for $n=12$ servers, the ring generation rule for $p=1, 5, 7, 11$ will lead into four distinct ring-\reduce permutations between the servers. Note that each $p$ describes a unique regular permutation.  To handle large scale clusters, we restrict $p$ to be a prime number, thereby reducing the search space size to only $O(\frac{n}{\ln~n})$, as per the Prime Number Theorem~\cite{num_theory}.

 \begin{figure}[t]
 \centering
 \subfloat[\name topology]{
 \includegraphics[width=0.4\columnwidth]{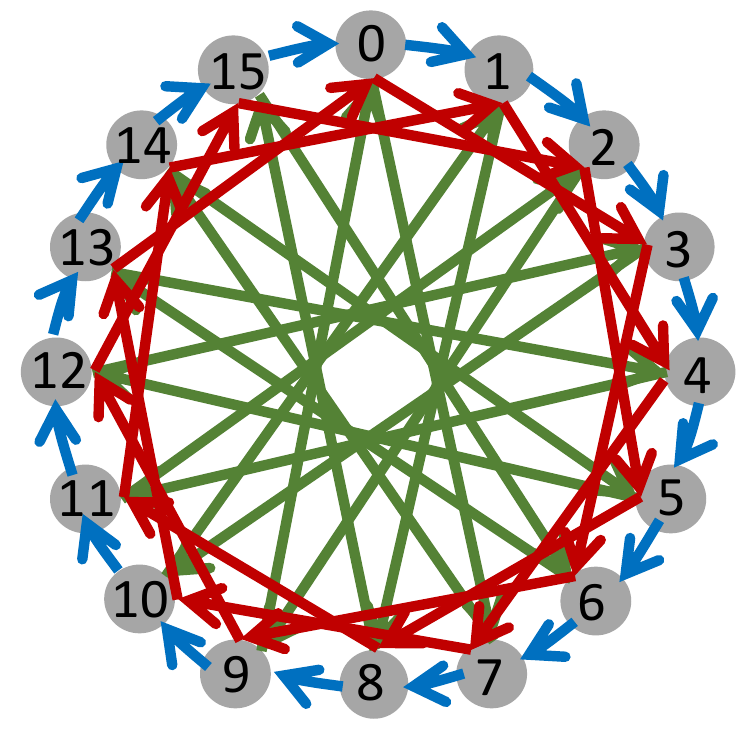}
 \label{fig:topoopt_topology}
 }
 \subfloat[\name traffic pattern]{
  \includegraphics[width=0.45\columnwidth]{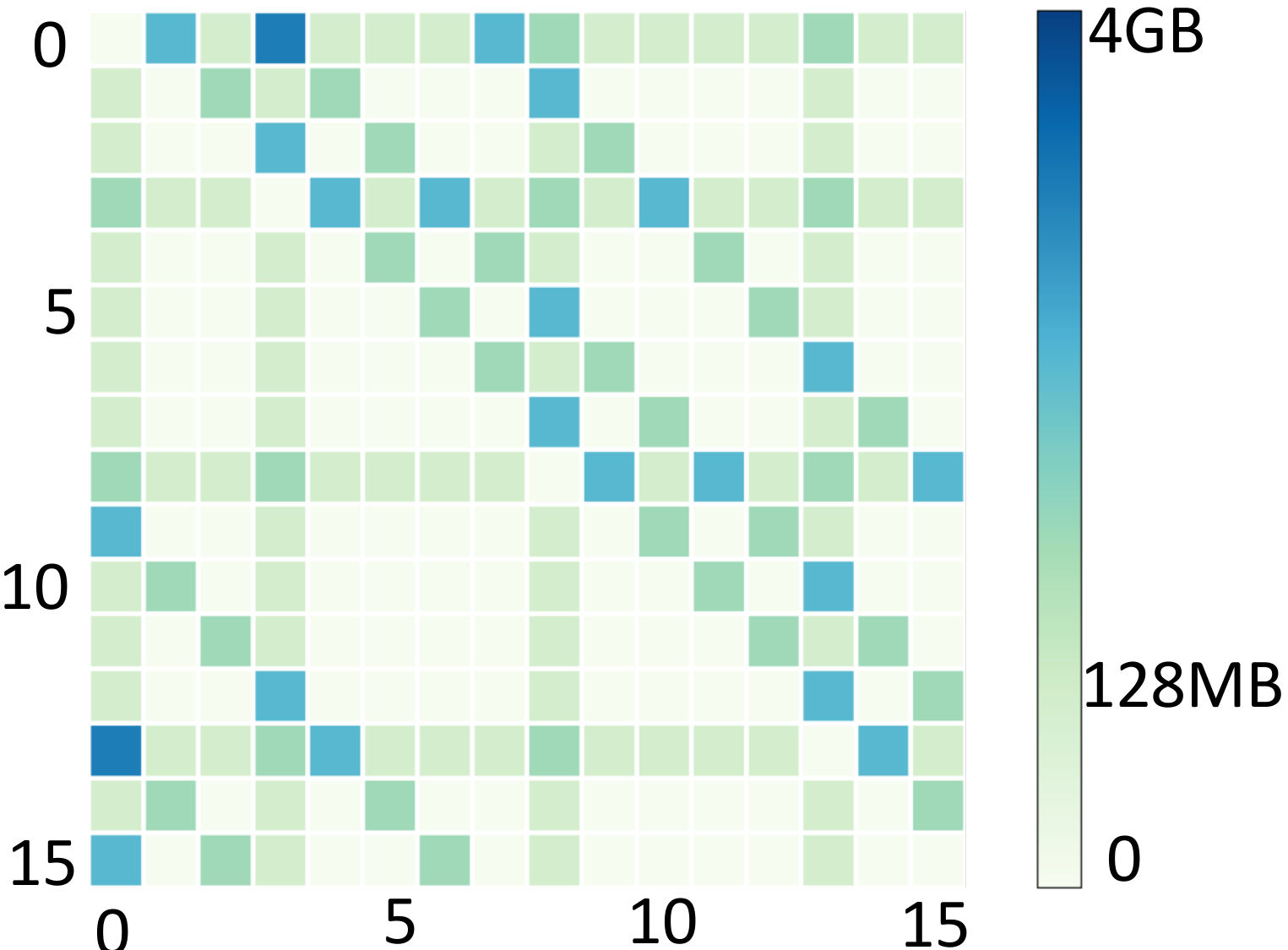}
  \label{fig:topoopt_heatmap}
  }
 \caption{\name's topology and traffic matrix.}
 \vspace{\captionvspace}
 \label{fig:fancy_allreduce_example}
 \end{figure}
 
 \para{SelectPermutations algorithm.} For a group of $n$ servers participating in \reduce, \texttt{\fancyperms} finds a set of regular permutations $P_k=\cup_{p:gcd(p ,n) = 1}\{p\}$ across them. \algo then selects $d_k$ permutations using a module called \texttt{SelectPermutations}, where $d_k$ is the number of degree allocated to the group of nodes running \reduce (line~\ref{algo:g_all_reduce_degree}, Algorithm~\ref{alg:topo_finder}). Algorithm~\ref{alg:SelectPermutations} presents the pseudocode of \texttt{SelectPermutations}. Several metrics can be used in the \texttt{SelectPermutations} module. In our implementation, \texttt{SelectPermutations} aims to reduce the cluster diameter to benefit the \MP transfers. To this end, \texttt{SelectPermutations} chooses $\{p_1,\cdots,p_{d_k}\}\subset P_k$, such that $\{p_1,\cdots,p_{d_k}\}$ is close (in L1-norm) to a geometric sequence (line~\ref{alg:selperm_l1}, Algorithm~\ref{alg:SelectPermutations}). 
 
 \newtheorem{theorem}{Theorem}
 \begin{theorem}
 \label{thm:cluster_bnd}
 \name's \texttt{SelectPermutations} algorithm bounds the diameter of the \reduce sub-topology to $O(d_A \cdot n^{1/d_A})$, under certain assumptions.
 \end{theorem}
 
 We list the assumptions and proof of Theorem~\ref{thm:cluster_bnd} in Appendix~\ref{app:hopcount_bound}. Intuitively, each server in the topology is able to reach a set of servers with a geometrically distributed hop-count distance (line~\ref{alg:selperm_geo}, Algorithm~\ref{alg:SelectPermutations}), creating a topology similar to Chord~\cite{stoica2001chord}. 
 
 \para{Example.} Consider the DLRM model in Figure~\ref{fig:dlrm_traffic_pattern}. Instead of choosing one of the \reduce permutations in Figure~\ref{fig:ring_permutations}, \name combines the three ring-\reduce permutations to load-balance the \reduce transfers while providing a short hop-count for \MP transfers. Figure~\ref{fig:fancy_allreduce_example} illustrates \name's topology and traffic matrix and shows a more balanced traffic matrix than Figure~\ref{fig:dlrm_traffic_pattern}.
 

\section{Large Scale Simulations}
\label{sec:sims}

This section evaluates the performance of a large-scale \name interconnect. 
First, we explain our simulation software and methodology (\S\ref{sec:sim_methodology}). Then, we provide a cost analysis of \name to inform our simulations when comparing different interconnects (\S\ref{sec:cost_model}). Next, we demonstrate the  performance of \name when a cluster is dedicated to a single distributed DNN training job (\S\ref{sec:sims_single_job}). We perform a sensitivity analysis to quantify the impact of all-to-all traffic (\S\ref{eval:a2a}) and host-based forwarding (\S\ref{eval:indirect_fw}). We extend this setting to a case where a training cluster is shared among multiple jobs (\S\ref{sec:sims_multi_job}). Finally, we evaluate the impact of reconfiguration latency (\S\ref{sec:reconfig_time}) on \name's performance.

\subsection{Methodology \& Setup}
\label{sec:sim_methodology}

We implement two simulators to evaluate \name.

\para{\fnet simulator.} We augment FlexFlow's simulator~\cite{flexflow_github} to be network-aware and call it \fnet. Given a DNN model and a batch size, FlexFlow's simulator explores  different parallelization strategies and device placements to minimize iteration training time. The output of this simulator is a \textit{task graph} describing the set of computation and communication tasks on each GPU and their dependencies. The current implementation of FlexFlow ignores the network topology by assuming servers are connected in a \textit{full-mesh} interconnect.  Our \fnet simulator extends the FlexFlow simulator and enables it to consider multiple networks, including \fattrees, \name, and expander networks. Moreover, \fnet implements our alternating optimization framework (\S\ref{sec:topoopt_algorithms}) to find an optimized network topology and routing rules for \name.

\para{\franksim simulator.} FlexFlow's simulator only provides course-grind estimation of training iteration time, because it does not simulate individual packets traversing through a network. Extending \fnet to become a packet-level simulator is computationally infeasible, because FlexFlow generally requires thousands of MCMC iterations to converge. To faithfully simulate per-packet behavior of network switches, buffers, and multiple jobs sharing the same fabric, we build a second event-based packet simulator, called \franksim, on top of htsim~\cite{htsim}. \franksim takes the output of \fnet (i.e., the optimized parallelization strategy, device placement of each operator, network topology, and routing rules) and simulates several training iterations. The link propagation delay is set to 1~$\mu$s throughout this section.

\para{Simulated network architectures.} We simulate distributed training clusters with $n$ servers equipped with four NVIDIA A100 GPUs~\cite{a100}. We vary $n$ in different experiments and simulate the following network architectures: 

\begin{itemize}[align=left, leftmargin=0pt, labelindent=0pt, listparindent=\parindent, labelwidth=0pt, itemindent=!]

\item \textbf{\namebf.} A \name interconnect where each server is equipped with $d$ NICs, each with bandwidth $B$ connected via a flat layer of optical devices.  At the beginning of each job, a shard of the network is selected, and the topology of the shard is reconfigured based on the output of our alternating optimization framework (\S\ref{sec:topoopt_algorithms}) and remains unchanged throughout the entire training job. Both OCS and patch panels are suitable for this architecture.

\item \textbf{OCS-reconfig.} To study the impact of changing the network topology within training iterations, we simulate a reconfigurable \name interconnect. We rely on commercially available Optical Circuit Switches (OCSs) for this design and assume the reconfiguration latency is 10~ms. Given that FlexFlow's parallelization strategy search is not aware of dynamically reconfigurable networks, following  prior work~\cite{sip-ml}, we measure the traffic demand every 50~ms and adjust the circuits based on a heuristic algorithm to satisfy the current traffic demand as much as possible. We also enable host-based forwarding such that the communication is not blocked even when a direct link is not available (Appendix~\ref{app:reconfig_heuristic}).

\item \textbf{\SBE.} An ideal electrical switch that scales to any number of servers, where each server is connected to the switch via a link with $d \times B$ bandwidth. For any pair of $d$ and $B$, no network can communicate faster than this ideal case. In practice, the \SBE can be approximated with a full-bisection bandwidth \fattree where the bandwidth of each link is $d \times B$.

\item \textbf{\LBE.} To compare the performance of \name to that of a similar-cost \fattree architecture, we simulate a full bisection bandwidth \fattree where each server has one NIC and the bandwidth of each link is $d \times B'$, where $B'$ is lower than $B$ and is selected such that \fattree's cost is similar to \name (\S\ref{sec:cost_model}).

\item \textbf{\OBE.} This is a 2:1 oversubscribed \fattree interconnect, where the bandwidth of each link is $d \times B$ but half of the links in the ToR uplink layer are omitted.

\item \textbf{SiP-ML~\cite{sip-ml}.} SiP-ML is a futuristic DNN training cluster with Tbps of bandwidth per GPU. While having a Tbps network is beneficial, our goal is to compare the algorithmic contributions of \name and SiP-ML. Hence, to make a fair comparison, we allocate $d$ wavelengths, each with bandwidth $B$, to each SiP-ML GPU and follow its SiP-Ring algorithm to find a topology with a reconfiguration latency of 25~$\mu$s. Appendix~\ref{app:sipml} elaborates on our modifications to SiP-ML.

\item \textbf{Expander~\cite{expander, jellyfish}.} Finally, we simulate a fabric where each server has $d$ NICs with bandwidth $B$ interconnected via an Expander topology.

\end{itemize}

\begin{figure}[t] 
\centering
\includegraphics[width=\columnwidth]{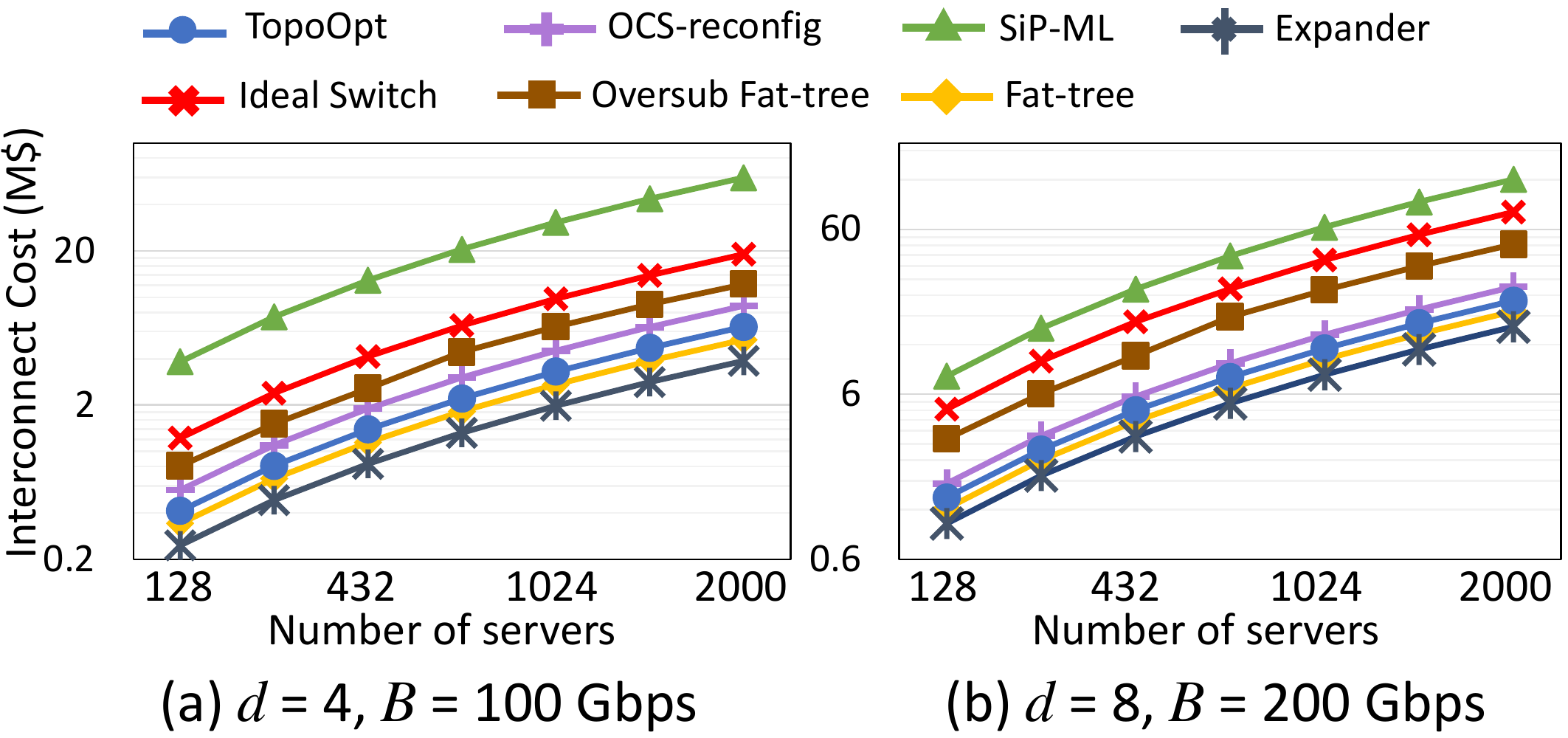}
\caption{Interconnect cost comparison.} 
\label{fig:cost}
\end{figure}

\begin{figure*}[t]
\centering
\includegraphics[width=\textwidth]{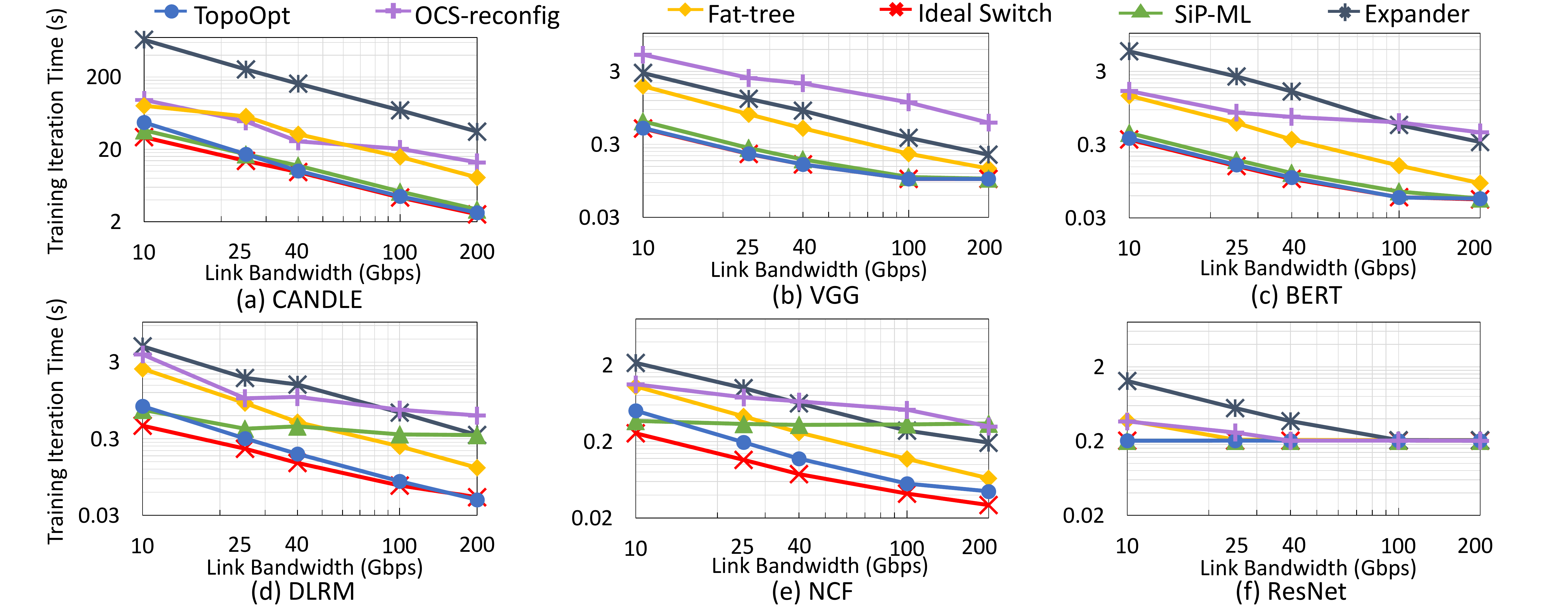}
\vspace{\captionvspace}
\caption{Dedicated cluster of 128 servers ($d$ = 4).}
\vspace{\captionvspace}
\label{fig:128_d4_single_job}
\end{figure*}

\para{DNN Workloads.} We simulate six real-world DNN models: DLRM~\cite{dlrm}, CANDLE~\cite{candle_uno}, BERT~\cite{bert}, NCF~\cite{ncf}, ResNet50~\cite{resnet} , and VGG~\cite{vgg}. List~\ref{list:model_parameters} (Appendix~\ref{sec:dnn_details}) provides details about model configurations and batch sizes used in this paper. 

\para{Parallelization strategy.} We use \fnet's   topology-aware parallelization strategy search for \SBE, \LBE, \OBE, SiP-ML, and Expander networks. For \name, we use \fnet's alternating optimization framework to find the best parallelization strategy jointly with topology, where the final parallelization strategy is either hybrid or pure data-parallel. We use ring-\reduce and distributed parameter server~\cite{osdi_parameter_server} as default \reduce communication collectives between servers and within servers, respectively. Each data point averages 5--10 simulation runs. 

\subsection{Cost Analysis}
\label{sec:cost_model}

We begin our evaluations by comparing the cost of various network architectures. Details about the cost of each component used in each architecture are given in Appendix~\ref{sec:cost_table}.

Figure~\ref{fig:cost} compares the interconnect cost across various network architectures as the number of servers is increased. We estimate the cost of \SBE with a full-bisection \fattree of the same bandwidth. We make the following observations. First, using OCSs for \name is more expensive (1.33$\times$, on average) than patch panels. Note that OCSs can be used in both \name and OCS-reconfig interconnects. Second, the cost of \name overlaps with that of the \LBE. This is intentional, because having a cost-equivalent architecture enables us to compare the performance of \name to a cluster at the same price point. Third, the ratio of \SBE's cost to \name's cost is 3.2$\times$ on average. Finally, the most and least expensive fabrics are SiP-ML and Expander, respectively, and as this section shows, they both perform worse than \name for certain workloads.

We acknowledge that estimating the cost of networking hardware is challenging because prices are subject to significant discounts with bulk orders. Assuming all components in this analysis are subject to similar bulk order discounts, the relative comparison across architectures remains valid. As a point of comparison, we compute the cost of a cluster with 4,394 servers ($k=26$ \fattree) by following the discounted cost trends in Sirius~\cite{sirius} and with 50\% discounts for patch panels. For a cluster at this scale, the cost of full-bisection bandwidth \fattree (which approximates our \SBE baseline) relative to the cost of \name changes from  3.0$\times$ to 3.6$\times$, indicating our estimates are reasonable. Moreover, a \name cluster incurs lower energy cost than \fattrees, as optical switches are passive.

\subsection{Performance Comparison on Dedicated Clusters}
\label{sec:sims_single_job}

This section compares the training iteration time of \name with that of other network architectures when the cluster is dedicated to serving one DNN training job.

Figure~\ref{fig:128_d4_single_job}a compares the training iteration times of various architectures for CANDLE distributed on a dedicated cluster of 128 servers with a server degree of four ($d = 4$). We vary the link bandwidth ($B$) on the x-axis. The figure shows that \SBE, \name, and SiP-ML architectures achieve similar performance because the best parallelization strategy for CANDLE at this scale is mostly data parallel, with few \MP transfers. The OCS-reconfig architecture performs poorly because it uses the instantaneous demand as the baseline to estimate the future traffic to schedule circuits. This estimation becomes inaccurate during training, in particular when the current \reduce traffic is about to finish but the next round of \reduce has not started. The Expander architecture has the worst performance, as its topology is not optimized for DNN workloads. Averaging across all link bandwidths, compared to \LBE interconnect, \name improves the training iteration time of CANDLE by 2.8$\times$; i.e., the ratio of CANDLE's iteration time on \LBE to \name is 2.8. \name's servers have more raw bandwidth, resulting in faster completion time.\footnote{It is possible to improve the performance of the Expander fabric by augmenting Blink's approach~\cite{wang2020blink} to a cluster-level solution.} 

Figures~\ref{fig:128_d4_single_job}b and \ref{fig:128_d4_single_job}c show the training iteration times for VGG and BERT. The trends are similar to CANDLE, as these models have similar degree requirements. Compared to \LBE, on average, \name improves the iteration time of VGG and BERT by 2.8$\times$ and 3$\times$, respectively.
 
The cases of DLRM and NCF are more interesting, as they have more \MP transfers than the other DNNs. As shown in Figures~\ref{fig:128_d4_single_job}d and~\ref{fig:128_d4_single_job}e, \name's performance  starts to deviate from \SBE, especially for NCF, because it uses host-based forwarding for the many-to-many \MP transfers (\S\ref{eval:a2a} and \S\ref{eval:indirect_fw}). For DLRM (and NCF), \name is 2.8$\times$ (and 2.1$\times$) faster than \LBE, while \SBE further improves the training iteration time by 1.3$\times$ (and 1.7$\times$) compared to \name.
SiP-ML performs poorly, and even when we increase the link bandwidth, its training iteration time stays flat. This happens because \MP transfers in DLRM and NCF require several circuit reconfigurations to meet the traffic demand.

Finally, Figure~\ref{fig:128_d4_single_job}f shows most architectures achieve similar training iteration times for ResNet50 since it is not a  communication-heavy model. The Expander architecture performs poorly when the link bandwidth is lower than 100~Gbps, as the topology does not match the \reduce traffic pattern.   

We repeat this simulation with $d=8$ and observe a similar performance trend (Appendix~\ref{sec:impact_of_degree}).

\subsection{Impact of All-to-all Traffic}
\label{eval:a2a}

This section evaluates the impact of all-to-all traffic patterns on \name's performance. In particular, \name's host-based forwarding approach incurs bandwidth tax~\cite{opera} exacerbated by \textit{all-to-all} and \textit{many-to-many} communication patterns. This tax is defined as the ratio of the traffic volume in the network (including forwarded traffic) to the volume of logical communication demand. Hence, the bandwidth tax for a full bisection bandwidth \fattree topology is always one, because hosts do not act as relays for each other.

Consider a DNN model with $R$~bytes of \reduce traffic and $A$~bytes of all-to-all traffic, distributed on a full bisection bandwidth topology with total network bandwidth $n\cdot B_{F}$ (i.e., number of servers multiplied by the bisection bandwidth). The training iteration time of this DNN is: $T_{F}=\frac{R}{n\cdot B_{F}} + \frac{A}{n\cdot B_{F}} + C_{bs}$, where $C_{bs}$ is the computation time of the model with batch size $bs$.\footnote{For clarify of presentation, this formulation assumes no overlap between communication and computation stages and no competing traffic.}

Now suppose the same DNN is distributed on a \name topology with total network bandwidth $n\cdot B_T$. In this case, assuming the entire \reduce traffic is carried on \fancyperms with direct links, the training iteration time becomes  $T_{T}=\frac{R}{n\cdot B_{T}} + \frac{\alpha \cdot A}{n\cdot B_{T}} + C_{bs}$ (Eq. 1), where $\alpha$ represents the slow-down factor that all-to-all transfers create in the network, due to host-based forwarding. The value of $\alpha$ depends on the amount of bandwidth tax and routing strategy (\S\ref{eval:indirect_fw}).

Increasing the amount of all-to-all traffic ($A$) increases the iteration time for both $T_{F}$ and $T_{T}$. But when $n\cdot B_{F}$ and $n\cdot B_T$ are equal, \name's performance degrades faster because of the $\alpha$ factor in the numerator. To quantify this behavior concretely, we distribute a DLRM training task with 128 embedding tables on a cluster with 128 servers. We choose large embedding tables and distribute each table on each server, creating worst-case all-to-all traffic.

Figure~\ref{fig:sim_a2a_impact} compares the training iteration times of \name, \SBE, and \LBE as the batch size is increased. The top x-axis lists the ratio of all-to-all to \reduce traffic for each batch size value given on the bottom x-axis. As shown in Figure~\ref{fig:sim_a2a_impact}a, when the batch size is 128 and $d = 4$, \name's performance matches that of \SBE, while \LBE is a factor of 2.7 slower. This result agrees with the performance gains in Figure~\ref{fig:128_d4_single_job}d, as the batch sizes are the same. 

Increasing the batch size increases $A$, and this, in turn, increases the training iteration times in all three architectures. As predicted by Eq.~(1), \name's iteration time increases faster. Specifically, when the batch size is 2048 and all-to-all traffic is $80\%$ of \reduce traffic, \name performs poorly, and the iteration time is a factor of 1.1 higher than that of the \LBE architecture. Increasing the server degree $d$ mitigates the problem, as shown in Figure~\ref{fig:sim_a2a_impact}b. Note that increasing the batch size does not always result in faster training time~\cite{sip-ml, shallue2018measuring, strong_scaling}. Moreover, publicly available data suggest 2048 is the largest batch size for training DLRM~\cite{mudigere2021highperformance}. 
The number of columns in the embedding tables and the number of servers are smaller in their workload: (92, 16) vs. (128, 128), respectively. Hence, the DLRM workload we evaluate contains more all-to-all traffic than the state-of-the-art model used in industry.

\begin{figure}[t]
\centering
\includegraphics[width=0.9\columnwidth]{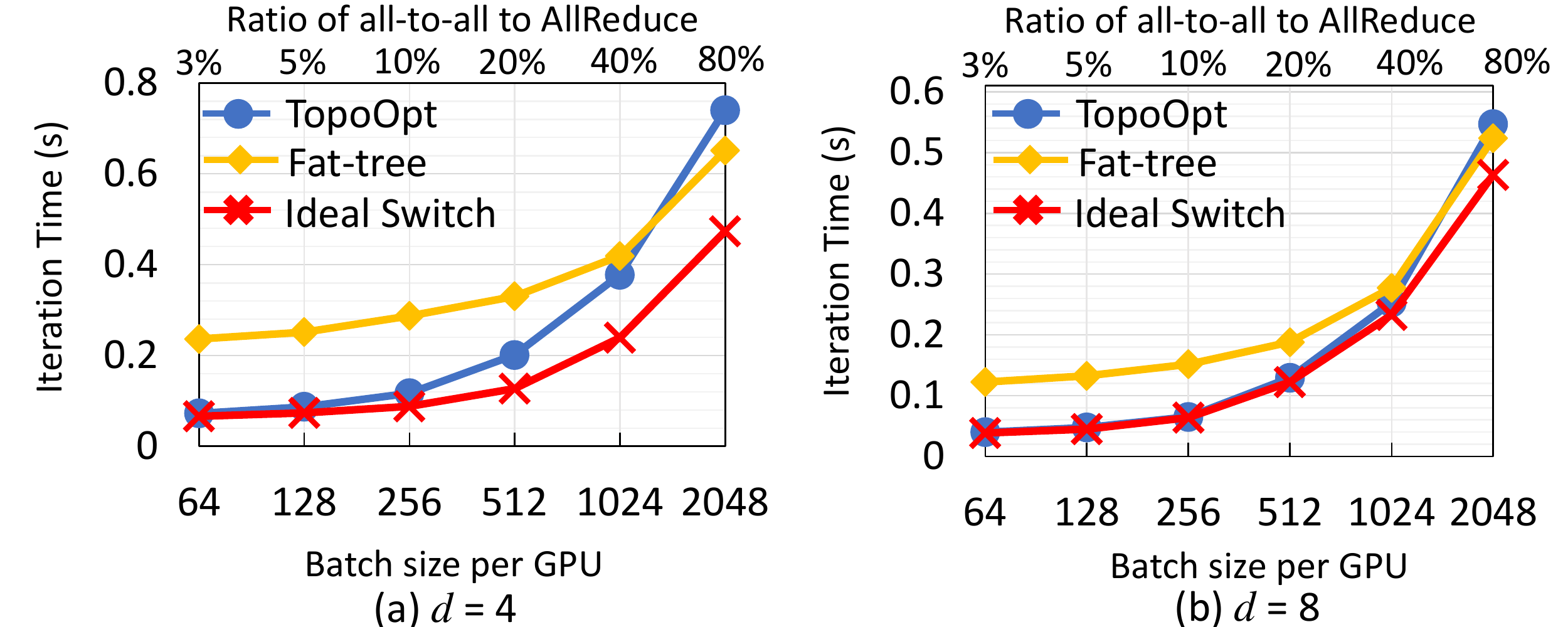}
\vspace{\captionvspace}
\caption{Impact of all-to-all traffic on a dedicated cluster of 128 servers ($B=$~100 Gbps).}
\vspace{\captionvspace}
\label{fig:sim_a2a_impact}
\end{figure}

\subsection{Impact of Host-based Forwarding}
\label{eval:indirect_fw}

Two factors impact the performance of host-based forwarding in \name: bandwidth tax and routing strategy.  

\para{Bandwidth tax.} Figure~\ref{fig:sim_a2a_bwtax} shows the amount of bandwidth tax experienced by the DLRM job in the previous section. Each bar represents a different batch size. 
At batch size 64 with $d=4$, \name experiences a bandwidth tax of 1.11, indicating that host-based forwarding creates $11\%$ extra traffic in the network. Increasing the degree to $d=8$ further improves this number to 1.05. In the worst-case scenario with batch size 2048, \name pays a bandwidth tax of 3.03 when $d=4$, causing it to perform worse than \LBE, as shown in Figure~\ref{fig:sim_a2a_impact}a. Determining the value of tolerable bandwidth tax is challenging for a \name cluster, as it depends on the compute time and the amount of compute-communication overlap, and this varies for different DNN models.

\para{Impact of path length.} Intuitively, the amount of bandwidth tax grows with the path length~\cite{opera}. Figure~\ref{fig:sim_hop_cdf} shows the CDF of path length across all server pairs. When $d=4$, the average path length is 5.7, resulting in at least 5.7$\times$ overhead of host-based forwarding relative to \SBE for all-to-all traffic. Based on Eq.~(1), and since the total network bandwidth in \name is higher than \LBE ($n\cdot B_{T} > n\cdot B_{F}$), the overhead of host-based forwarding becomes at least 1.4$\times$ for the \LBE architecture. Increasing the server degree to 8 reduces the average path length to 3, thereby reducing  the overhead bound. Appendix~\ref{sec:impact_of_degree} evaluates the impact of increasing node degree on performance for other models.

\para{Routing strategy.} Building a topology with a small path length is necessary but not sufficient to reduce the impact of host-based forwarding. To handle forwarded traffic with minimum performance impact, the routing strategy also needs to be efficient. The best routing strategy \textit{minimizes the maximum link utilization} for a given network topology, similar to WAN traffic engineering solutions~\cite{smore}. However, finding the optimal routing strategy requires solving a set of linear equations with a centralized controller~\cite{msoftswan, b4}. To quantify the load imbalance in \name, Figure~\ref{fig:load_balancing} illustrates the CDF of the amount of traffic carried by each physical link for an all-to-all traffic matrix. When the batch size is 128 (Figure~\ref{fig:load_balancing}a), the link with the least traffic carries $39\%$ and $59\%$ less traffic than the link with the most traffic, for $d=4$ and $d=8$, respectively. This imbalance in load suggests further opportunities to improve the performance of \name. Achieving optimal routing makes $\alpha$  (Eq.~(1)) equal to the average path length. Without a centralized controller, however, the link utilization becomes non-uniform, and the average path length only serves as a lower bound. We leave optimizing the routing strategy in \name to future work.

\begin{figure}[t]
\centering
\begin{minipage}{0.22\textwidth}
\includegraphics[width=\textwidth]{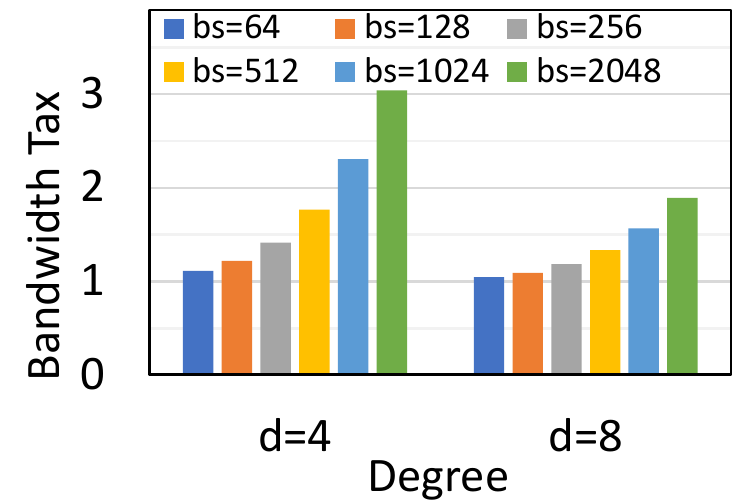}
\vspace{\captionvspace}
\caption{Bandwidth tax.}
\vspace{\captionvspace}
\label{fig:sim_a2a_bwtax}
\vspace*{-0.2cm}
\end{minipage} 
\hspace{0.02cm}
\begin{minipage}{0.22\textwidth}
\includegraphics[width=\textwidth]{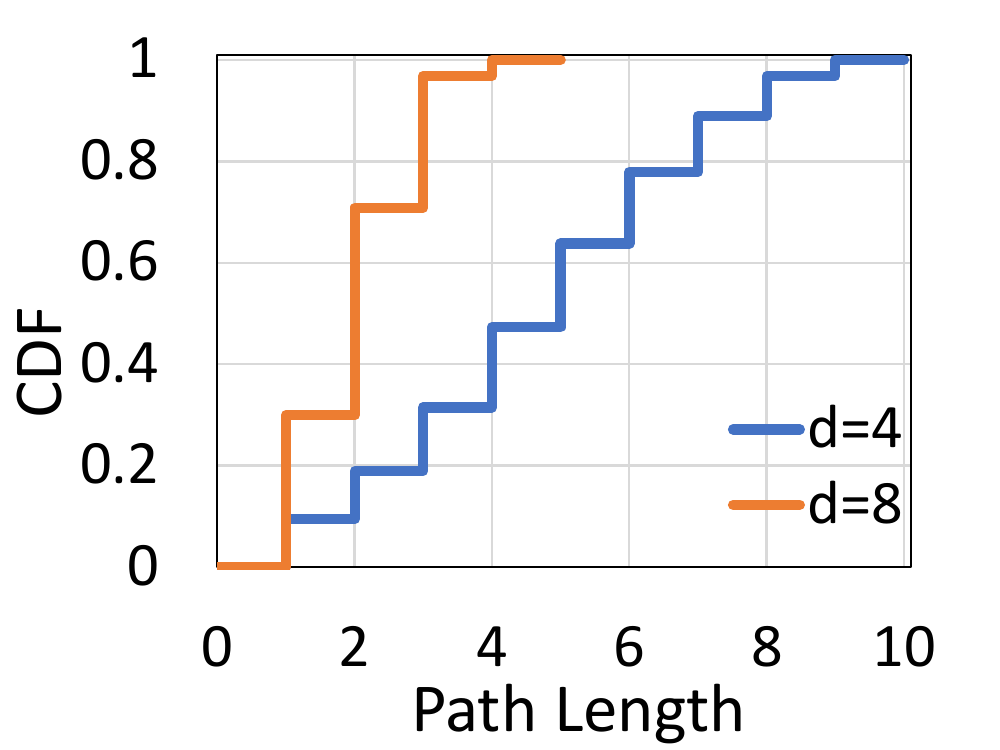}
\vspace{\captionvspace}
\caption{Path length CDF.} 
\vspace{\captionvspace}
\label{fig:sim_hop_cdf}

\end{minipage}
\end{figure}

\begin{figure}[t]
\centering
\includegraphics[width=0.9\columnwidth]{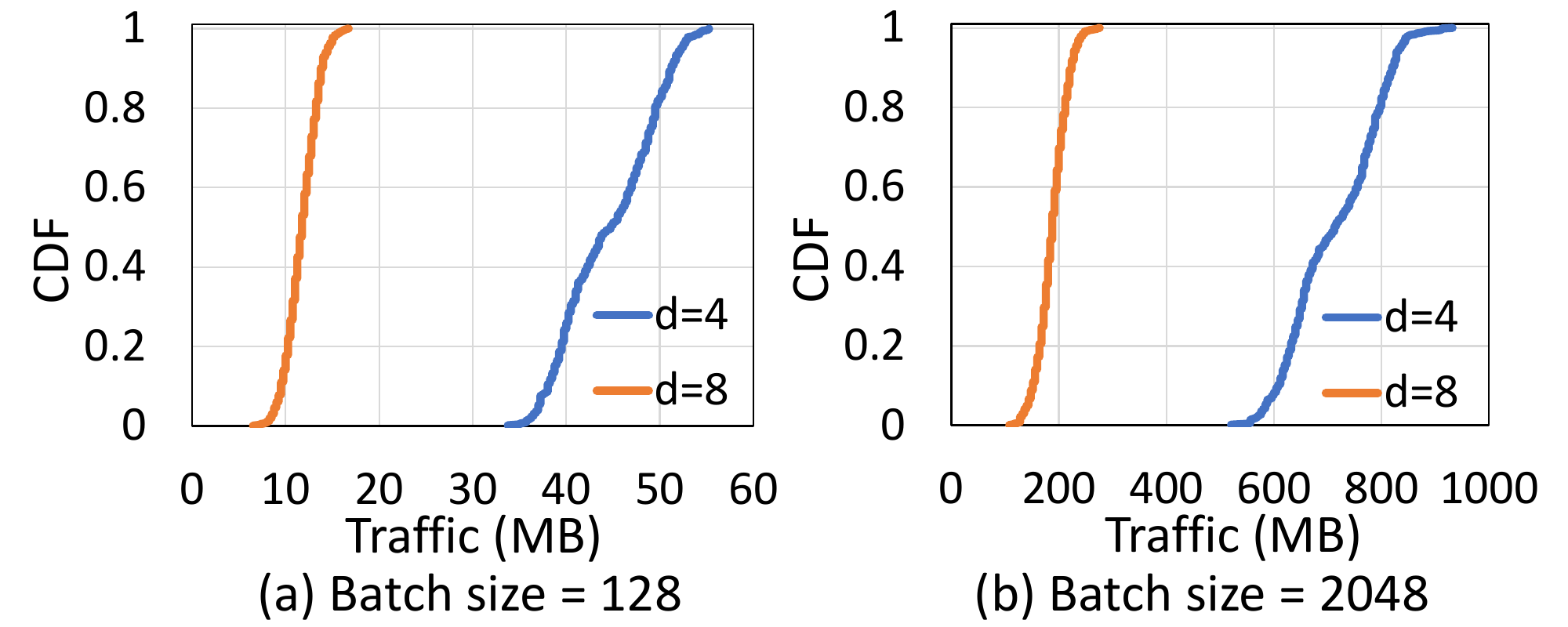}
\vspace{\captionvspace}
\caption{Traffic distribution.}
\label{fig:load_balancing}
\end{figure}

\subsection{Performance on Shared Clusters}
\label{sec:sims_multi_job}

We now compare the performance of different network architectures when the cluster is shared across multiple DNN jobs. Following prior work~\cite{themis, pollux}, we run a series of simulations where 40\% of the jobs are DLRM, 30\% are BERT, 20\% are CANDLE, and 10\% are VGG16. We change the number of active jobs to represent the load on the cluster. Assuming each job requests 16 servers (64 GPUs), we execute 5, 10, 15, 20, and 27 jobs on the cluster to represent $20\%$, $40\%$, $60\%$, $80\%$ and $100\%$ load, respectively.

Figure~\ref{fig:432_d8_multi_job} compares the average and 99\%-tile iteration time at different loads for a cluster with 432 servers where $d$ = 8 and $B$ = 100~Gbps. SiP-ML does not support multiple jobs; hence, we omit it in this experiment. We omit OCS-reconfig and Expander networks, as they both show poor performance in this setting. Instead, we add the \OBE interconnect to demonstrate the impact of congestion on \fattree topologies. Figure~\ref{fig:432_d8_multi_job}a shows that \name improves the average iteration time by 1.7$\times$ and 1.15$\times$, compared to the \LBE and \OBE architectures,  respectively. We observe a similar trend for the tail iteration completion times, depicted in Figure~\ref{fig:432_d8_multi_job}b. At the extreme case when all servers are loaded, \name's tail training iteration time is $3.4\times$ faster compared to \LBE architecture. Averaging across all load values on the x-axis, \name improves the tail training iteration time by 3$\times$ and 1.4$\times$ compared to \LBE and \OBE architectures.

\begin{figure}[t]
\centering
\includegraphics[width=0.9\columnwidth]{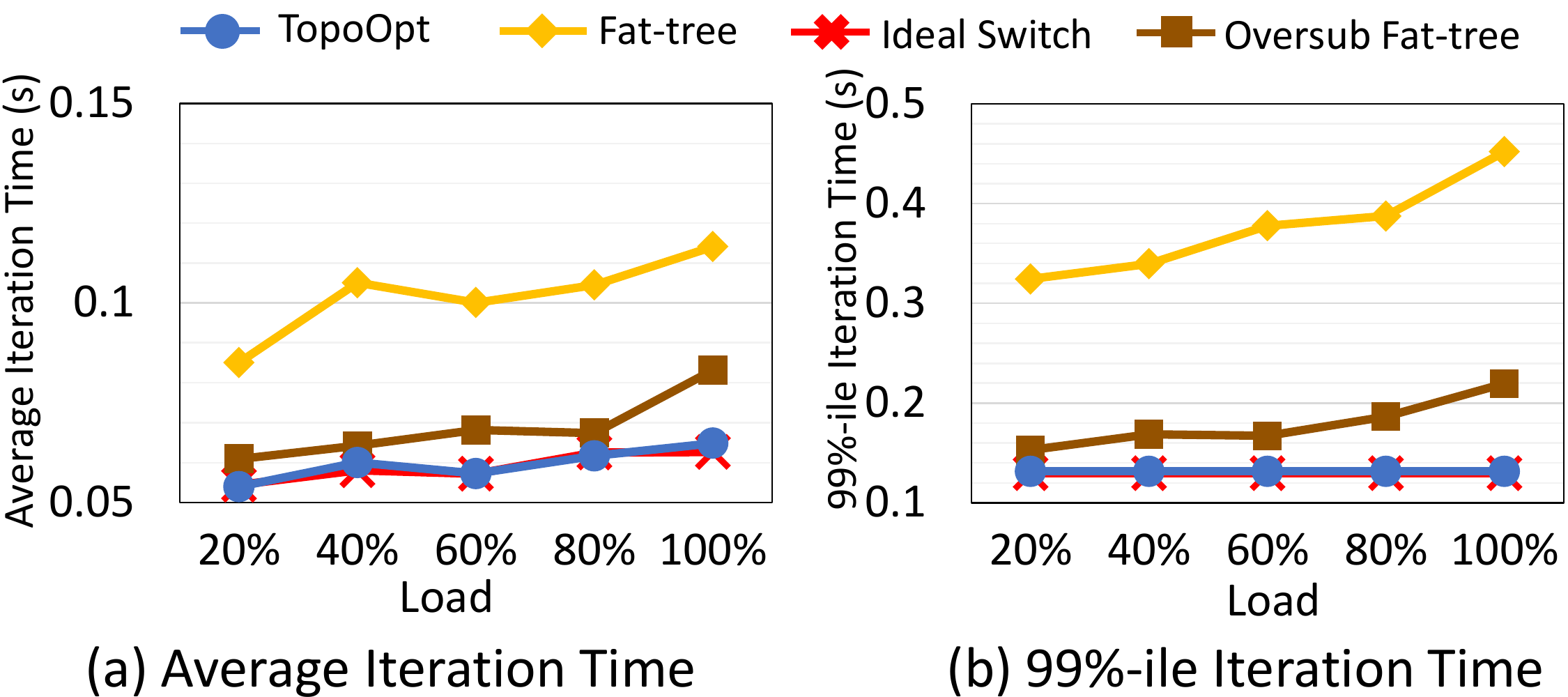}
\caption{Shared cluster of 432 servers ($d$ = 8, $B$ = 100~Gbps).}
\label{fig:432_d8_multi_job}
\end{figure}

\subsection{Impact of Reconfiguration Latency}
\label{sec:reconfig_time}

Figure~\ref{fig:reconfig_impact} shows the training iteration time of DLRM and BERT in the same setting as Figure~\ref{fig:128_d4_single_job}, while sweeping the reconfiguration latency of OCSs in OCS-reconfig from $1~\mu s$ to $10~ms$. The horizontal blue line corresponds to \name's iteration time; it remains constant as it does not reconfigure the network topology. We find host-based forwarding is challenging when the network is reconfigurable, as the circuit schedules need to account for forwarding the traffic while the topology reconfigures. Therefore, we evaluate the performance of OCS-reconfig with and without host-based forwarding. The purple line corresponds to OCS-reconfig with host-based forwarding (same as OCS-reconfig evaluated in Figure~\ref{fig:128_d4_single_job}), denoted by OCS-reconfig-FW. For the orange line, we disable host-based forwarding (similar to SiP-ML) and call it OCS-reconfig-noFW. 

We find enabling host-based forwarding when the topologies reconfigures within a training iteration is not always beneficial. For DLRM (Figure~\ref{fig:reconfig_impact}a), OCS-reconfig-FW achieves better performance than OCS-reconfig-noFW, as DLRM has all-to-all \MP transfers which benefit from host-based forwarding. However, for BERT (Figure~\ref{fig:reconfig_impact}b), enabling forwarding increases the chance of inaccurate demand estimation and imposes extra bandwidth tax, therefore increasing the iteration time of OCS-reconfig-FW by a factor of 1.4 compared to OCS-reconfig-noFW. 

Reducing the reconfiguration latency all the way to 1~$\mu$s enables OCS-reconfig-noFW to match the performance of \name. However, OCS-reconfig-FW still suffers from inaccurate demand estimations. Although fast reconfigurable switches are not yet commercially available, they are going to be essential in elastic scenarios where the cluster is shared across multiple jobs and servers join and leave different jobs unexpectedly, or when large, high-degree communication dominates the workload. We believe futuristic fast reconfigurable switches, such as Sirius~\cite{sirius}, are well-suited for this setting. Finding a parallelization algorithm that is aware of reconfigurability within training iterations is a challenging and exciting future research problem.

\begin{figure}[t]
\centering
\includegraphics[width=\columnwidth]{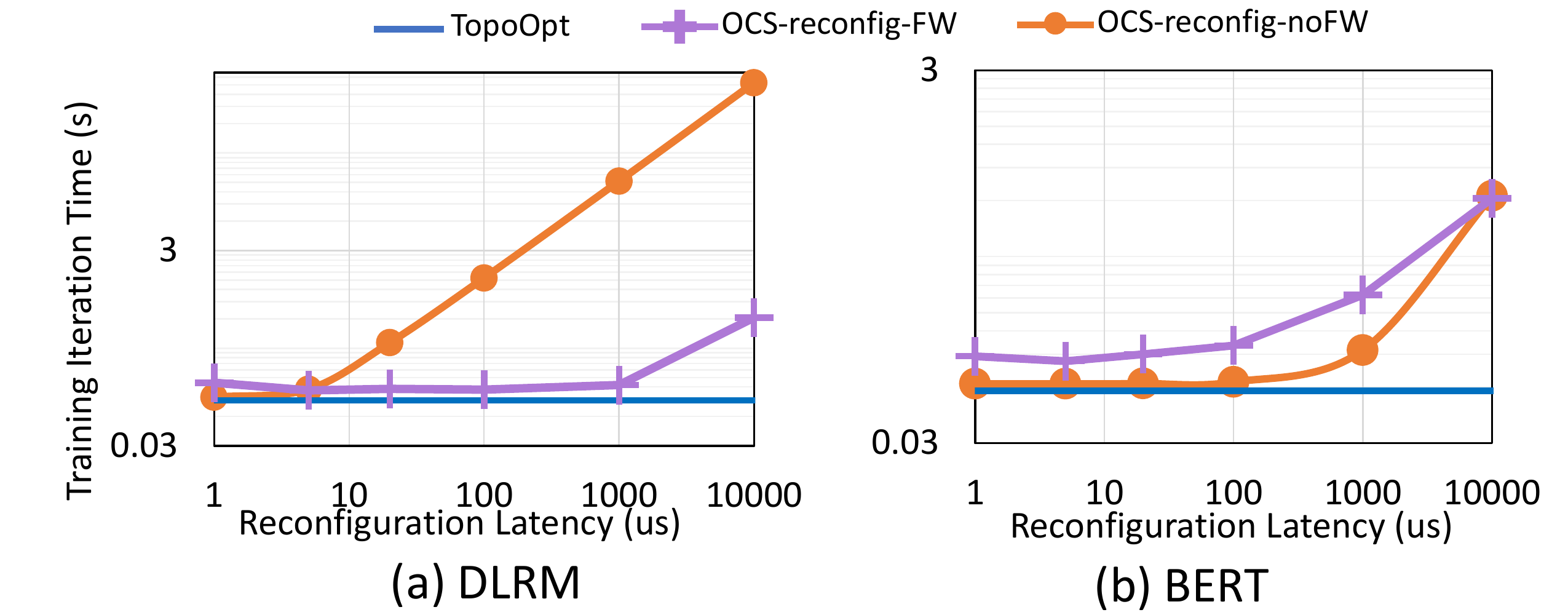}
\vspace{\captionvspace}
\caption{Impact of reconfiguration latency ($d$=8, $B$=100~Gbps).}
\vspace{\captionvspace}
\label{fig:reconfig_impact}
\end{figure}


\section{Prototype}
\label{sec:prototype}

\para{Testbed setup.} We build a prototype to demonstrate the feasibility of \name. Our prototype includes 12 ASUS ESC4000A-E10 servers and a G4 NMT patch panel~\cite{telescent}. Each server is equipped with one A100 Nvidia GPU~\cite{a100} (40~GB of HBM2 memory), one 100~Gbps HP NIC~\cite{hpe-nic}, and one 100~Gbps Mellanox ConnectX5 NIC. Our HP NICs are capable of supporting 4$\times$25~Gbps interfaces using a PSM4 transceiver with four breakout fibers~\cite{psm4_transceiver}, enabling us to build a \name system with degree $d = 4$ and $B = 25$~Gbps. We use RoCEv2 for communication, and enable DCB~\cite{dcb} and PFC on these interfaces to support a lossless fabric for RDMA. We build a completely functional \name prototype with our patch panel (Figure~\ref{fig:testbed_photo}). We compare \name's performance with two baselines: ($i$) Switch~100Gbps, where the servers are connected via 100~Gbps links to a switch, and ($ii$) Switch~25Gbps, where the servers are connected via 25~Gbps links to a switch. The Switch~100Gbps baseline corresponds to the \SBE case in our simulations.

\begin{figure*}[t]
\centering
\begin{minipage}{0.18\textwidth}
\includegraphics[width=\textwidth]{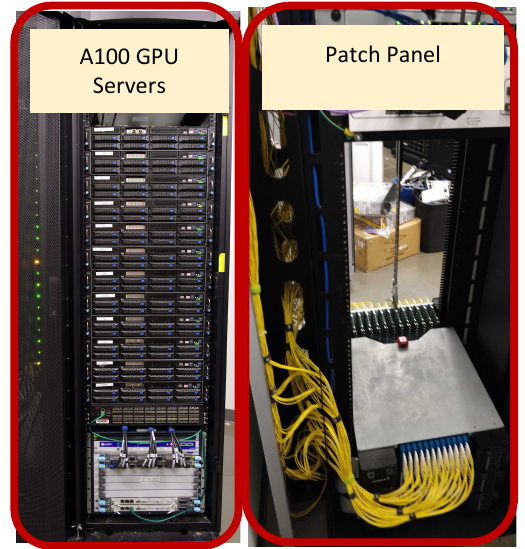}
\caption{Testbed photo.}
\label{fig:testbed_photo}
\end{minipage}
\hspace{0.2cm}
\begin{minipage}{0.25\textwidth}
\includegraphics[width=\textwidth]{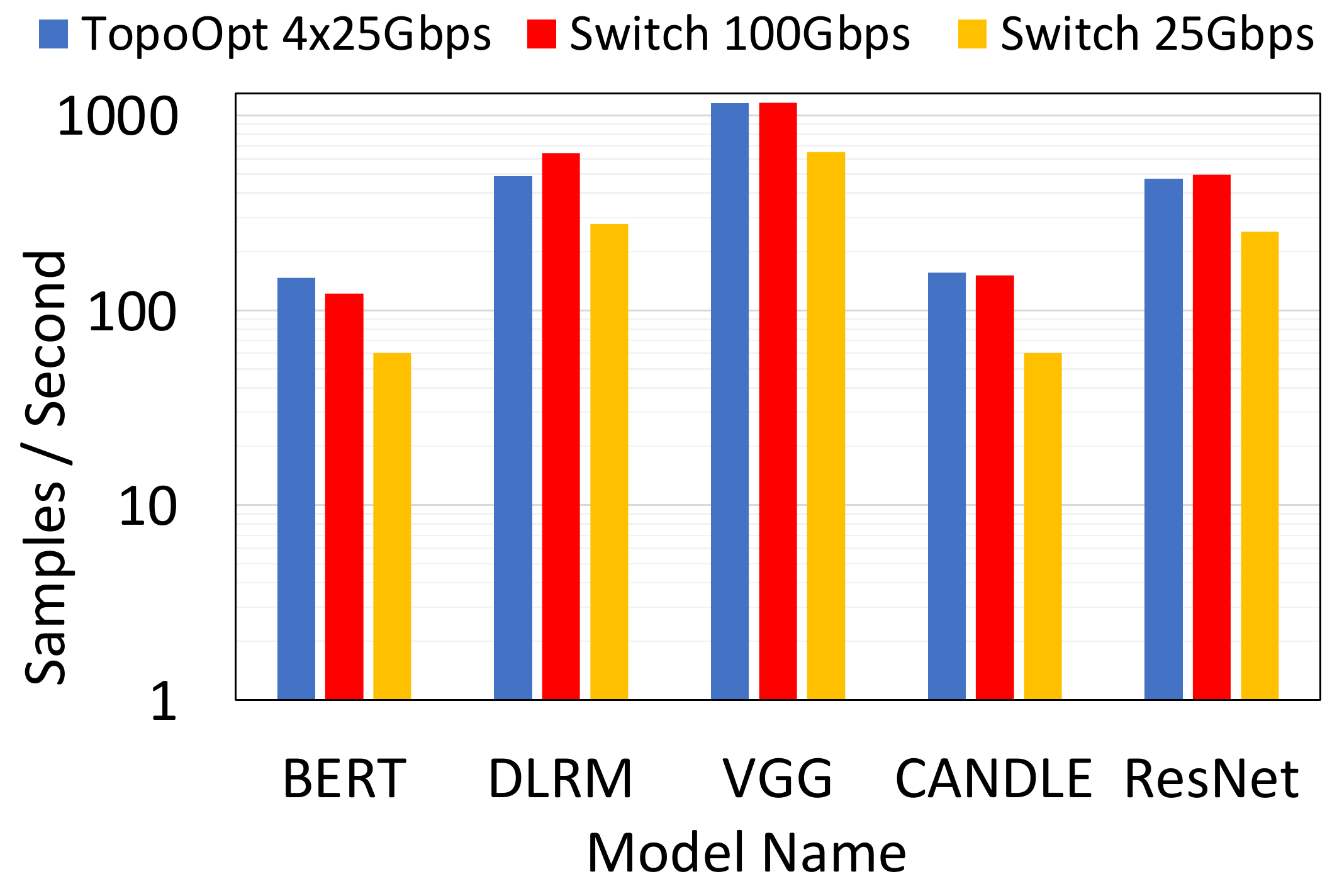}
\caption{Training throughput (samples/second).} 
\label{fig:testbed_results_tput}
\end{minipage} 
\hspace{0.2cm}
\begin{minipage}{0.25\textwidth}
\includegraphics[width=\textwidth]{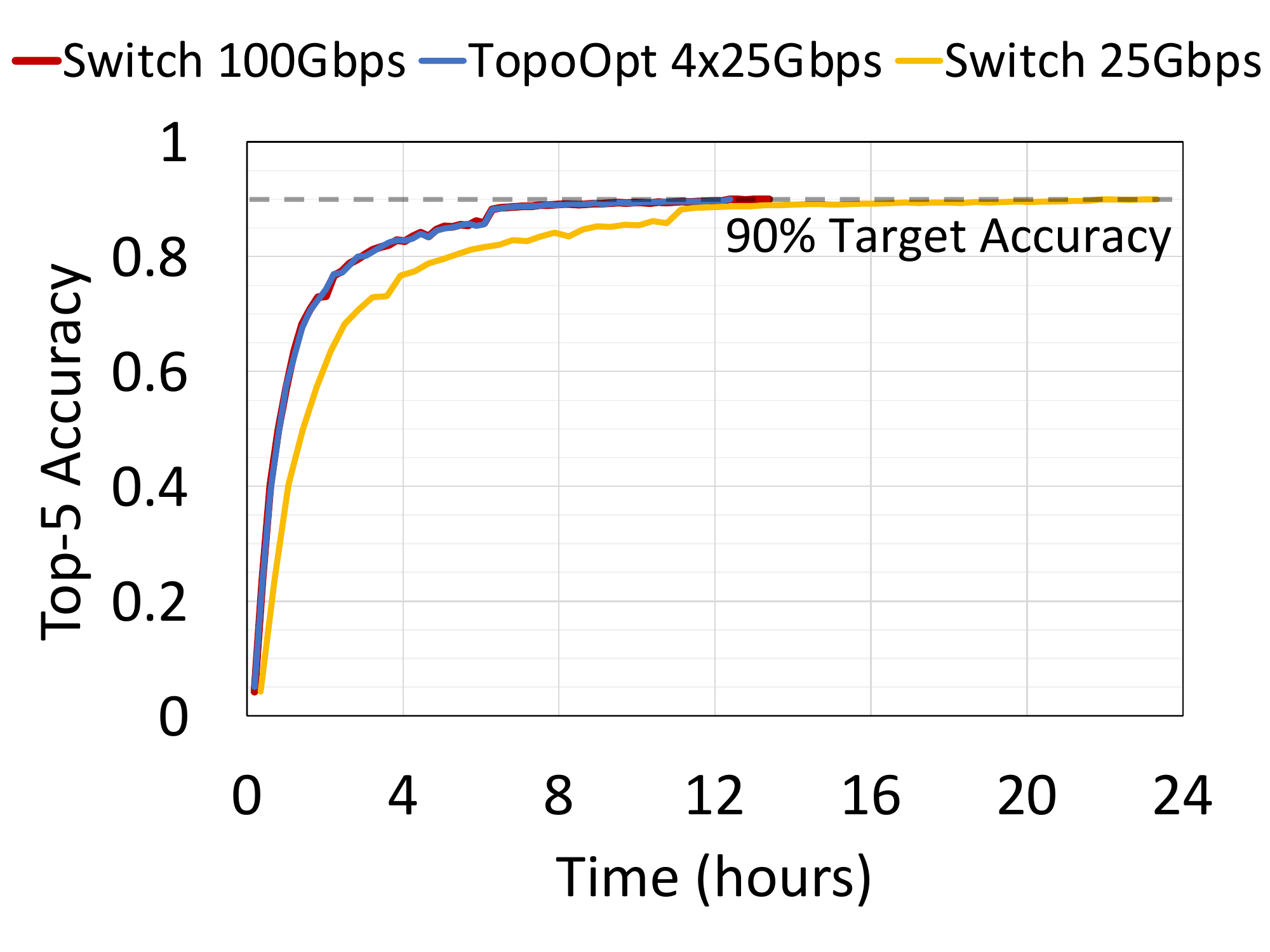}
\caption{Time-to-accuracy of VGG19 with ImageNet.}
\label{fig:time_to_acc_resnet}
\end{minipage}
\hspace{0.2cm}
\begin{minipage}{0.25\textwidth}
\includegraphics[width=\textwidth]{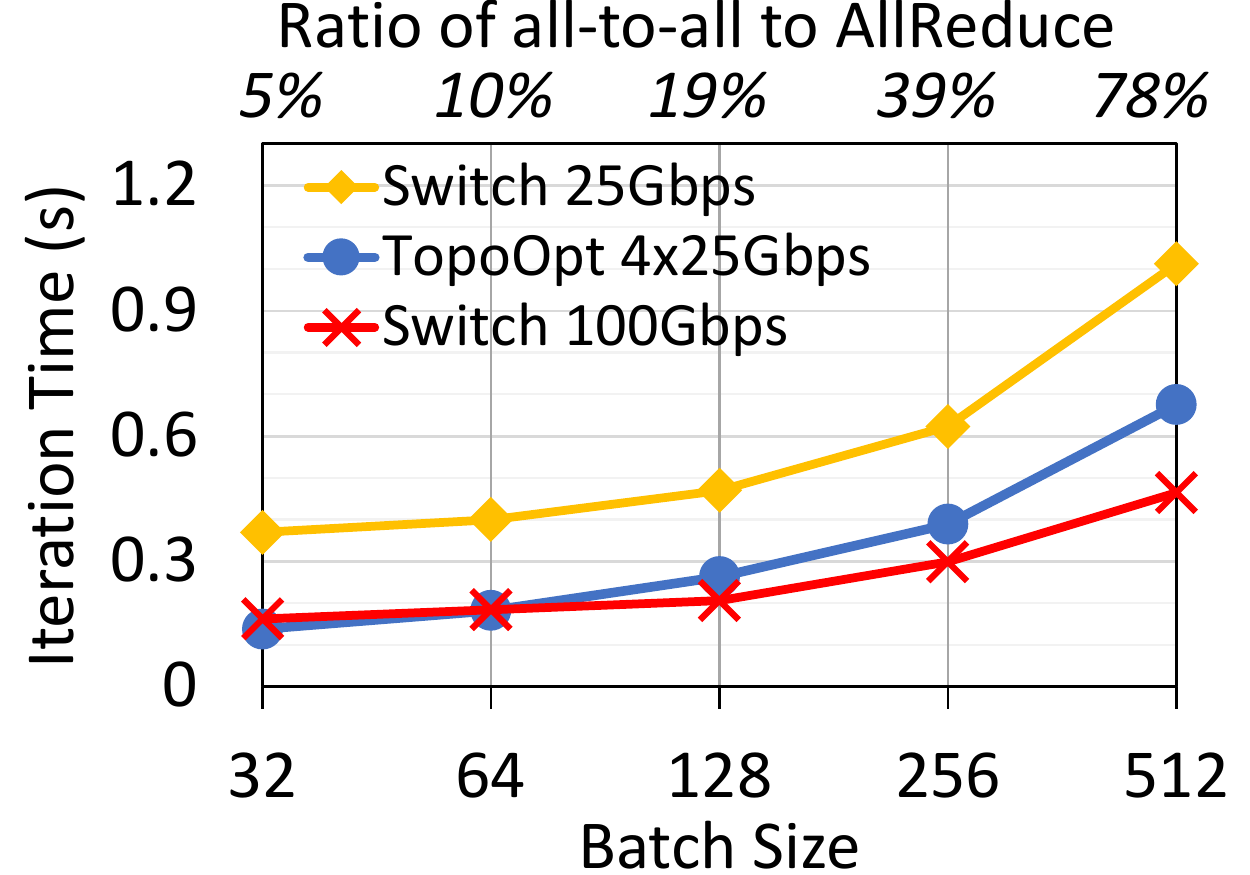}
\caption{Impact of all-to-all traffic in our testbed.}
\label{fig:testbed_a2a}
\end{minipage}
\end{figure*}

\para{Distributed training framework.} We use FlexFlow's training engine~\cite{flex_flow_ai},  based on Legion's parallel programming system~\cite{legion}, to train four DNN models: ResNet50~\cite{resnet}, BERT~\cite{bert}, VGG16~\cite{vgg}, and CANDLE~\cite{candle_uno}. For DLRM, we use Facebook's implementation from~\cite{dlrm}. Since our prototype is an order of magnitude smaller in scale than our simulation setup, we use smaller model and batch sizes. 

\para{Modifications to NCCL.} By default, the NCCL communication library~\cite{nccl} assumes all network interfaces are routable from other interfaces. This assumption is not ideal for \name because we have a specific routing strategy to optimize training time.  We modify NCCL to understand \name's topology and respect its routing preferences. Moreover, we integrate our \fancyperms\ \reduce permutations into NCCL and enable it to load-balance parameter synchronization across multiple ring-\reduce permutations.

\para{RDMA forwarding.} Implementing \name with today's RDMA NICs requires solving an engineering challenge, because the RDMA protocol assumes a switch-based network. Packet processing and memory access in RDMA protocol are offloaded to the NIC, and a RoCEv2 packet whose destination IP address is different from that of the host is assumed to be corrupted. Therefore, the NIC silently drops forwarded packets. To address this issue, we collaborated with engineers at Marvell who developed the firmware and driver of our HP NICs. Our solution uses a feature called network partitioning (NPAR)  which enables the NIC to separate host-based forwarding traffic from  direct traffic, and uses the Linux kernel to route them (details in Appendix~\ref{app:rdma_forwarding}). Our conversations with Marvell indicate that updating the firmware and the driver enables the NIC to route forwarded RoCEv2 packets, thereby bypassing the kernel entirely.

\para{Training performance.} Figure~\ref{fig:testbed_results_tput} demonstrates that \name's training throughput (samples/second) is similar to our Switch 100~Gbps baseline for all models. The performance of Switch~25Gbps baseline is lower because its available bandwidth is lower than \name. Figure~\ref{fig:time_to_acc_resnet} shows the time-to-accuracy plot of training VGG19 on the ImageNet~\cite{deng2009imagenet} dataset. As the figure indicates, \name reaches the target accuracy of 90\% $2.0\times$ faster than the Switch~25Gbps baseline. \name achieves similar performance to the Switch~100Gbps baseline, as the blue and red lines overlap in Figure~\ref{fig:time_to_acc_resnet}.

\para{Impact of all-to-all traffic.}
Similar to Section~\ref{eval:a2a}, we  evaluate the impact of all-to-all \MP traffic on our RDMA-forwarding enabled testbed by measuring the average iteration time across 320 iterations of a DLRM job distributed in our testbed. We vary the amount of all-to-all traffic by changing the batch size. To create worst-case traffic, we increase the embedding dimensions by 128$\times$ relative to the state-of-the-art~\cite{dlrm} (model details are in List~\ref{list:model_parameters}, Appendix~\ref{sec:dnn_details}). Figure~\ref{fig:testbed_a2a} shows the training iteration time for various batch sizes. The results are consistent with Figure~\ref{fig:sim_a2a_impact}, but since the bandwidth tax in our 12-server testbed is much smaller than a 128-server cluster in simulations, \name performs better relative to the switch-based architectures for a given all-to-all to \reduce traffic ratio. For instance, for batch size 512, the ratio of all-to-all traffic to \reduce is $78\%$, and the training iteration time with \name is 1.6$\times$ better than the Switch~25Gbps baseline. 
\section{Discussion}
\label{sec:discussion}

\para{Target workload.} The most suitable workload for a \name cluster is a set of large DNN training jobs with hybrid data and model parallelism (or simply data parallelism).  We assume the set of servers assigned to each job remains the same throughout the lifetime of the job, and the GPUs are not shared across multiple jobs.

\para{Storage and control plane traffic}. \net's training clusters consist of custom-designed servers, each with eight GPUs, eight dedicated NICs for training traffic (GPU NICs), and four additional NICs for storage and other traffic (CPU NICs)~\cite{mudigere2021highperformance}. Other companies, such as NVIDIA, have similar architectures~\cite{superpod}. \name only considers GPU NICs as server degree and partitions the network dedicated for training traffic. The CPU NICs are connected through a separate fabric to carry storage and other control plane traffic.

\para{Supporting dynamic scheduling and elasticity.} Others have demonstrated the benefits of dynamically choosing the training servers for elastic training jobs~\cite{pollux, themis}. Our target use case in \net is to leverage \name for the vast number of long-lasting training jobs that do not change dynamically. In cases where elasticity is required, instead of using patch panels, we use OCSs (or other fast reconfigurable optical switches) to change the servers participating in a job quickly. Note that dynamically changing the set of servers participating in a job while keeping both the topology and the parallelization strategy optimal requires augmenting the optimization space with an additional dimension, making the problem even more challenging. We leave this to future work.

\para{Handling failures.} Unlike SiP-ML's single ring topology~\cite{sip-ml}, a single link failure does not disconnect the graph in \name. When a fiber fails, \name can temporarily use a link dedicated to \MP traffic to recover an \reduce ring. In case of permanent failures, \name  reconfigures to swap ports and recover the failed connection.

\para{Supporting multi-tenancy.} To support multi-tenancy~\cite{MLSYS2021_9f61408e, salus}, \name can leverage NVIDIA's MIG~\cite{mig} to treat one physical server as multiple logical servers in its topology.

\para{\fancyperms in \fattrees.} Although our \fancyperms technique is well-suited for  reconfigurable optical interconnects, it may be of independent interest for \fattree interconnects as well since load-balancing the \reduce traffic across multiple permutations can help with network congestion.

\para{\namebf's limitations.} \name's approach assumes the traffic pattern does not change between iterations. However, this assumption may not hold for Graphic Neural Network (GNN) models~\cite{gnn} or Mixture-of-Expert (MoE) models~\cite{moe}. In addition, we plan to extend \name by bringing its demand-awareness design within training iterations. This is an open research question, and as shown in Section~\ref{sec:reconfig_time}, we will need fast-reconfigurable optical switches, as well as a more sophisticated scheduling algorithm. Another limitation of \name is that a single link failure within a \reduce ring causes the full ring to become inefficient for \reduce traffic. A fast optical switch addresses this problem by quickly reconfiguring the topology.

\section{Related Work}

\para{Optimizing DNN training.} To address the increasing computation and network bandwidth requirements of large training jobs, a plethora of frameworks have been proposed~\cite{flex_flow, pipedream, byteps_1, byteps_2, wang2020blink, efficient_comp_comm_neurips, placeto, gpipe, horovod, blueconnect, goyal, uber, firecaffe, topo_aware_sparse, kylix, deepspeedmoe, zeroinfinity}. These frameworks distribute the dataset and/or DNN model across accelerators while considering the available network bandwidth, but unlike \name, they do not consider optimizing the \textit{physical layer topology}. Specifically, Blink~\cite{wang2020blink} builds collectives for distributed ML, but it needs a physical topology to generate spanning trees. Zhao et al.~\cite{cloptzhao22} study the optimal topology for collective communication operations, but this does not apply for general \MP traffic. In addition, several methods have been proposed to quantize and compress the gradients to reduce the amount of communication data across servers~\cite{qsgd, adaptive_compression, gradient_sparsification}. While these approaches are effective, they are designed for data parallel strategies and do not consider the large amount of data transfers caused by model parallel training. Wang et al.~\cite{network_topology_on_performance} compare the performance of \fattrees and BCube topologies for distributed training workloads and highlight several inefficiencies in \fattrees. SiP-ML~\cite{sip-ml} demonstrates the benefits of 8~Tbps silicon photonics-based networks for distributed training. However, unlike \name, these proposed approaches do not \textit{co-optimize} topology and parallelization strategy. 

\para{DNN parallelization strategies.} Data and model parallelism are widely used by today's DNN frameworks (e.g., TensorFlow~\cite{tensorflow}, PyTorch~\cite{pytorch}, MXNet~\cite{mxnet}) to parallelize training across multiple devices. Recent work has also proposed \emph{automated frameworks} (e.g., FlexFlow~\cite{flex_flow}, ColocRL~\cite{colocRL}, MERLIN~\cite{merlin}) that find efficient parallelization strategies by searching over a comprehensive space of potential strategies. These frameworks rely on and are optimized for the conventional \fattree interconnects. \name proposes a new approach to building DNN training systems by jointly optimizing network topology and parallelization strategy.

\para{DNN training infrastructures and schedulers.} Several training infrastructures have been proposed recently, including NVIDIA DGX SuperPOD~\cite{superpod}, TPU cluster~\cite{tpu}, and supercomputers~\cite{summit}. All these systems assume non-reconfigurable network topologies, such as \fattree, Torus, and other traffic-oblivious interconnects. \name is the first DNN system to use commodity reconfigurable interconnects to accelerate DNN jobs.Gandiva~\cite{gandiva}, Themis~\cite{themis}, Tiresias~\cite{tiresias}, BytePS~\cite{byteps_1, byteps_2}, and Pollux~\cite{pollux} seek to improve the utilization of GPU clusters through scheduling algorithms. These approaches are complementary to ours, and many of their techniques can be applied to a \name cluster.

\para{Optical Interconnects.} Several papers have demonstrated the benefits of optically reconfigurable  interconnects for datacenters~\cite{sirius, rotornet, opera, helios, projector, megaswitch, mordia, reactor, solstice, sirius_ecoc, quartz,optlatreconfig21}. These designs lead to sub-optimal topologies for distributed DNN traffic. Similarly, \textit{traffic oblivious} interconnects, such as RotorNet~\cite{rotornet, opera}, are a great fit for datacenter workloads, but they are not suitable for DNN training jobs characterized by repetitive traffic demands. Hybrid electrical/optical datacenter proposals~\cite{helios, cthrough} can be used to route \reduce traffic through the optical fabric and \MP flows through a standard electrical \fattree network. But hybrid clusters are not cost effective and suffer from many problems, including TCP ramp-up inefficiencies~\cite{etalon}, segregated routing issues~\cite{algorithmic_compleixty}, and uncertainty in terms of how to divide the cluster between electrical and optical fabrics~\cite{projector, firefly}.

\vspace*{-2mm}

\section{Conclusion}
We present \name, a novel system based on optical devices that jointly optimizes DNN parallelization strategy and topology to accelerate training jobs. We design an alternating optimization algorithm to explore the large space of  \textit{Computation $\times$ Communication $\times$ Topology} strategies for a DNN workload, and demonstrate \name obtains up to 3.4$\times$ faster training iteration time than \LBE.

\vspace*{-2mm}

\section{Acknowledgments}
We thank our shepherd Sangeetha Abdu Jyothi and anonymous reviewers for their valuable feedback. We also acknowledge Meta for supporting this research. In particular, we thank Omar Baldonado, Gaya Pradeep Sindhu, and Jahangir Hasan. In addition, we thank Alan Gibbemeyer, Bob Shine, Karl Kuhn and Ramiro Voicu from Telescent for their support on the Telescent NTM-G4. We also thank Arial Elior, Karl Erickson, and Nishant Lodha from Marvell for their help on RDMA forwarding. The MIT-affiliated authors are supported by ARPA-E ENLITENED PINE DE-AR0000843, DARPA FastNICs 4202290027, NSF grants CNS-2008624, SHF-2107244, ASCENT-2023468, CAREER-2144766, PPoSS-2217099, CNS-2211382, Meta faculty award, Google faculty award, and Sloan fellowship FG-2022-18504.

\label{bodypage}
\bibliographystyle{plain}
\balance
\bibliography{reference}

\appendix
\newpage

\nobalance
\section{Tree-\reduce and Other \reduce Permutations}
\label{app:tree_reduce}

Section~\ref{sec:measurement} established that we can manipulate the traffic of a ring-\reduce collective by permuting the labeling of servers in the \reduce group. Here, we illustrate how to use the same technique on another \reduce algorithm, called tree-\reduce. 

In the tree-\reduce algorithm, the servers are connected logically to form a tree topology. The \reduce operation starts by running a reduce operation to the root node with recursive halving, followed by a broadcast to the rest of the cluster with recursive doubling~\cite{thakur2005optimization}. 

A common instantiation of tree-\reduce is the \text{double binary tree} (DBT) algorithm~\cite{sanders2009two}. In this algorithm, the first step is to create a balanced binary tree for the nodes. The properties of balanced binary trees guarantee that one half of the nodes will be leaf-nodes, and the other half will be in-tree; thus, a second binary tree is constructed by flipping the labeling of the leaf and in-tree nodes. This way, each node (except the root in both trees) has the same communication requirements for the \reduce operation, as described in the last paragraph, and bandwidth-optimally is achieved. Figure~\ref{fig:tree_permutations}a shows an example where in the first binary tree, the in-tree nodes are even, and the leaf nodes are odd, while the second tree flips the labeling.

The DBT itself is essentially an example of permuting the node labeling to achieve an \reduce operation with balanced communication load. We also note that we can permute the labeling \textit{for the entire set of nodes} for a pair of DBTs to create a new pair of trees that can perform the \reduce operation at the same speed. Figures~\ref{fig:tree_permutations}b and~\ref{fig:tree_permutations}c illustrate two other possible double binary trees, and their corresponding traffic demand matrix for the DLRM and CANDLE example shown in Figures~\ref{fig:dlrm_traffic_pattern_tree} and~\ref{fig:candle_traffic_pattern_tree} (\S\ref{sec:measurement}). Arbitrary permutations can be used, and to limit the cases, we could simply consider the cyclic permutations in the modular space as described in \fancyperms. 

In general, all \reduce operations can be described as a directed graph $G=(V, E)$ where $V$ is the set of nodes in the cluster, and $E$ denotes data dependencies. The \textit{permutable} property says every graph $G'=(V, E')$ that is isomorphic to $G$ can perform the \reduce operation equally well, where the homomorphism between $G$ and $G'$ is described by the symmetric group on $V$ (generally denoted by $Sym(V)$ in group theory).

\begin{figure}[t]
\centering
\footnotesize
\begin{minipage}{\columnwidth}
\includegraphics[width=0.98\columnwidth]{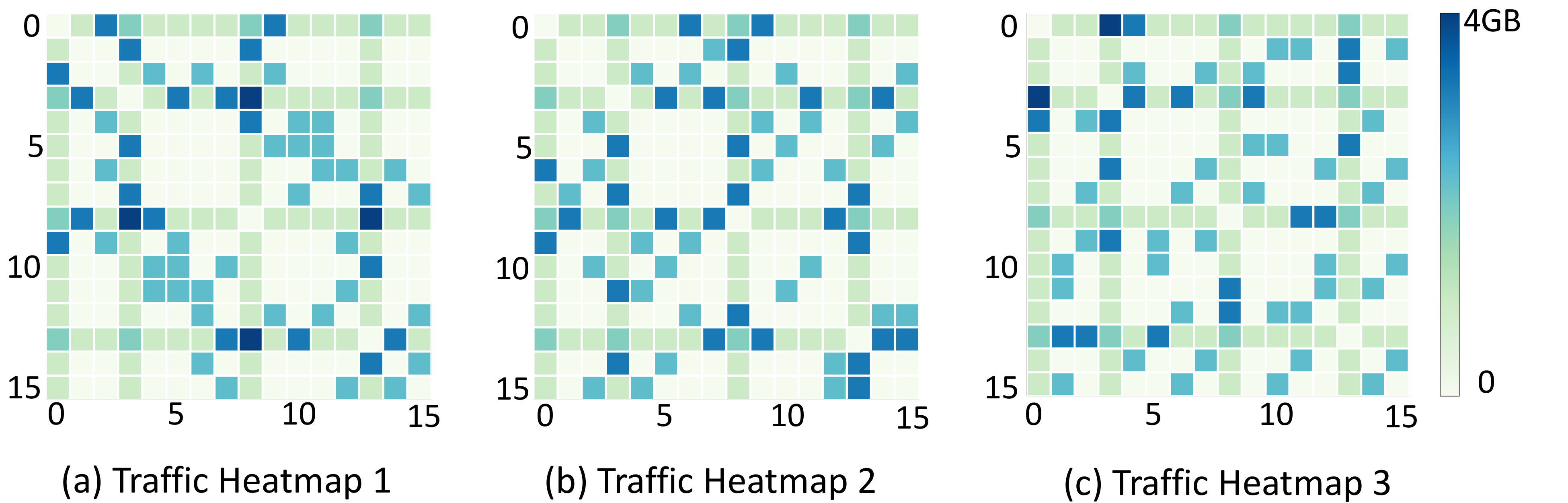}
\caption{DLRM traffic heatmaps with double binary tree \reduce.}
\vspace*{8mm}
\label{fig:dlrm_traffic_pattern_tree}
\end{minipage}
\vspace*{10mm}
\begin{minipage}{\columnwidth}
\centering
\includegraphics[width=0.98\columnwidth, trim= 0 0 0 35, clip]{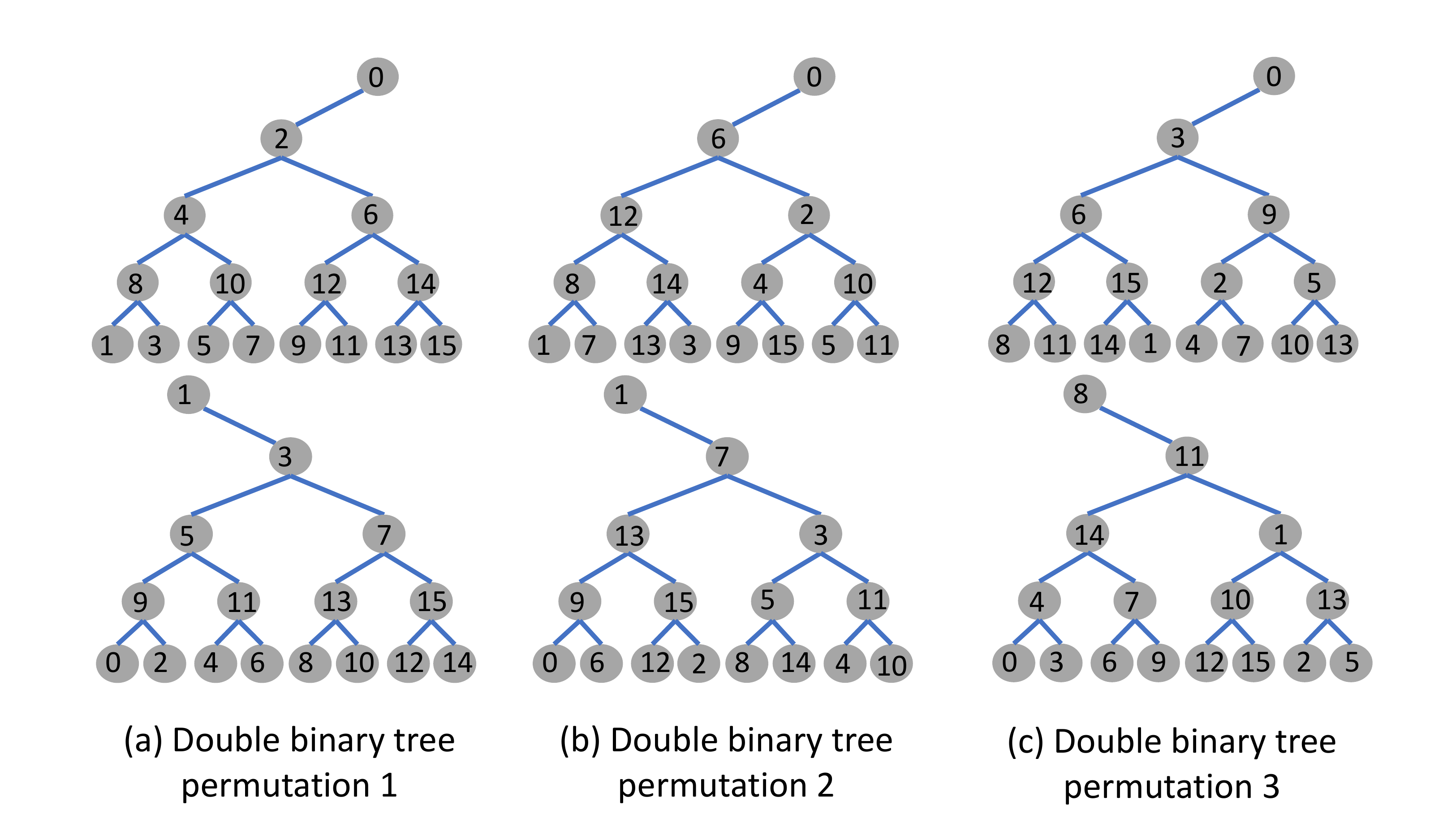}
\caption{Double binary tree (DBT) permutations.} 
\label{fig:tree_permutations}
\end{minipage}
\begin{minipage}{\columnwidth}
\centering
\includegraphics[width=0.98\columnwidth]{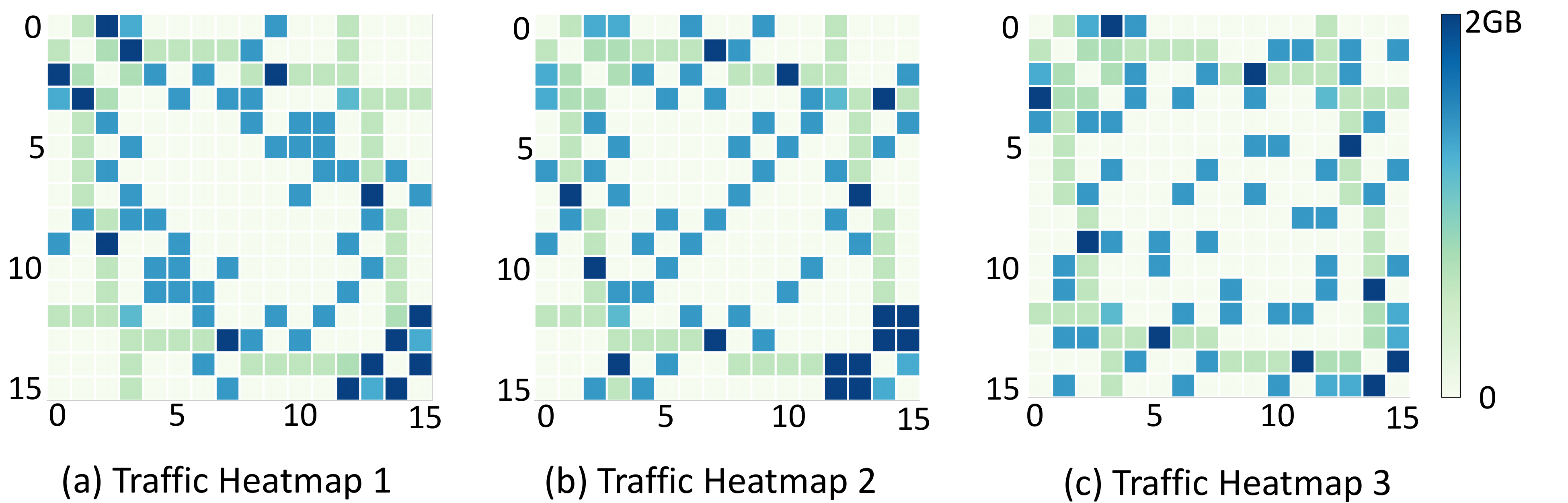}
\caption{CANDLE traffic heatmaps with double binary tree \reduce.}
\label{fig:candle_traffic_pattern_tree}
\end{minipage}
\end{figure}

\section{Commercially Available Patch Panels and Optical Circuit Switches}
\label{sec:patch_panel_details}

\para{Optical patch panels.} A patch panel is a device to facilitate connecting different parts of a system. For instance, electrical patch panels are used in recording studios and concert sound systems to connect microphones and electronic instruments on demand~\cite{patch_panel_wiki}. 
Fiber optic patch panels are commonly used for cable management, and have been proposed in recent datacenter topology designs~\cite{FatClique}. Reconfigurable optical patch panels are a new class of software-controlled patch panels and are already commercialized at scale~\cite{telescent-commercialized}. For instance, Telescent offers 1008 duplex ports with insertion loss less than 0.5~dB and cost $\approx$\$100K (\$100/port)~\cite{telescent-commercialized, kewitsch2009patchpanel}. Reconfiguration is performed using a robotic arm that grabs a fiber on the transmit side and connects it to a fiber on the receive side~\cite{kewitsch2009patchpanel}. 
However, the reconfiguration latency of optical patch panels is several minutes~\cite{telescent}.
Note that reliability is of utmost concern for operation in unmanned locations; for example, Telescent NTM patch panels have been certified to NEBS Level 3 and have over 1 billion port hours in operation~\cite{telescent_reliability}.

\para{3D MEMS-based Optical Circuit Switches (OCSs).} An OCS uses tiny mirrors to change the direction of light, thereby reconfiguring optical links. The largest optical circuit switch on the market has 384 duplex ports with $\approx$10~ms reconfiguration latency and is available for \$200K (\$520/port)~\cite{polatis}. However, the optical loss of these switches is 1.5--2.7~dB~\cite{polatis-datasheet}. Compared to patch panels, OCSs have the following disadvantages: ($i$) each port is five times more expensive; ($ii$) their insertion loss is higher; and ($iii$) their port-count is three times lower. The main advantage of OCSs is that their reconfiguration latency is \textit{four orders of magnitude} faster than patch panels. 

\section{Handling Sharding and Dynamic Job Arrivals in Shared Clusters}
\label{app:sharding}

\begin{figure}[t]
\includegraphics[width=0.9\columnwidth]{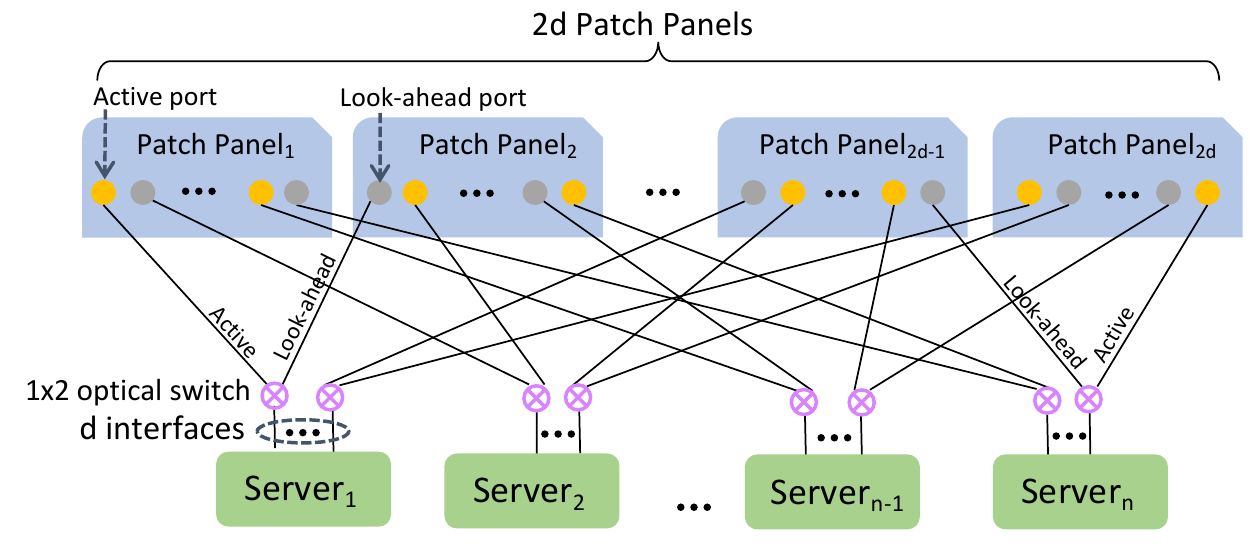}
\caption{Active \& Look-ahead ports for high reconfiguration latency.}
\label{fig:flat_active_look_ahead}
\end{figure}

\begin{figure}[t]
\centering
\includegraphics[width=0.9\columnwidth,trim= 0 0 0 65, clip]{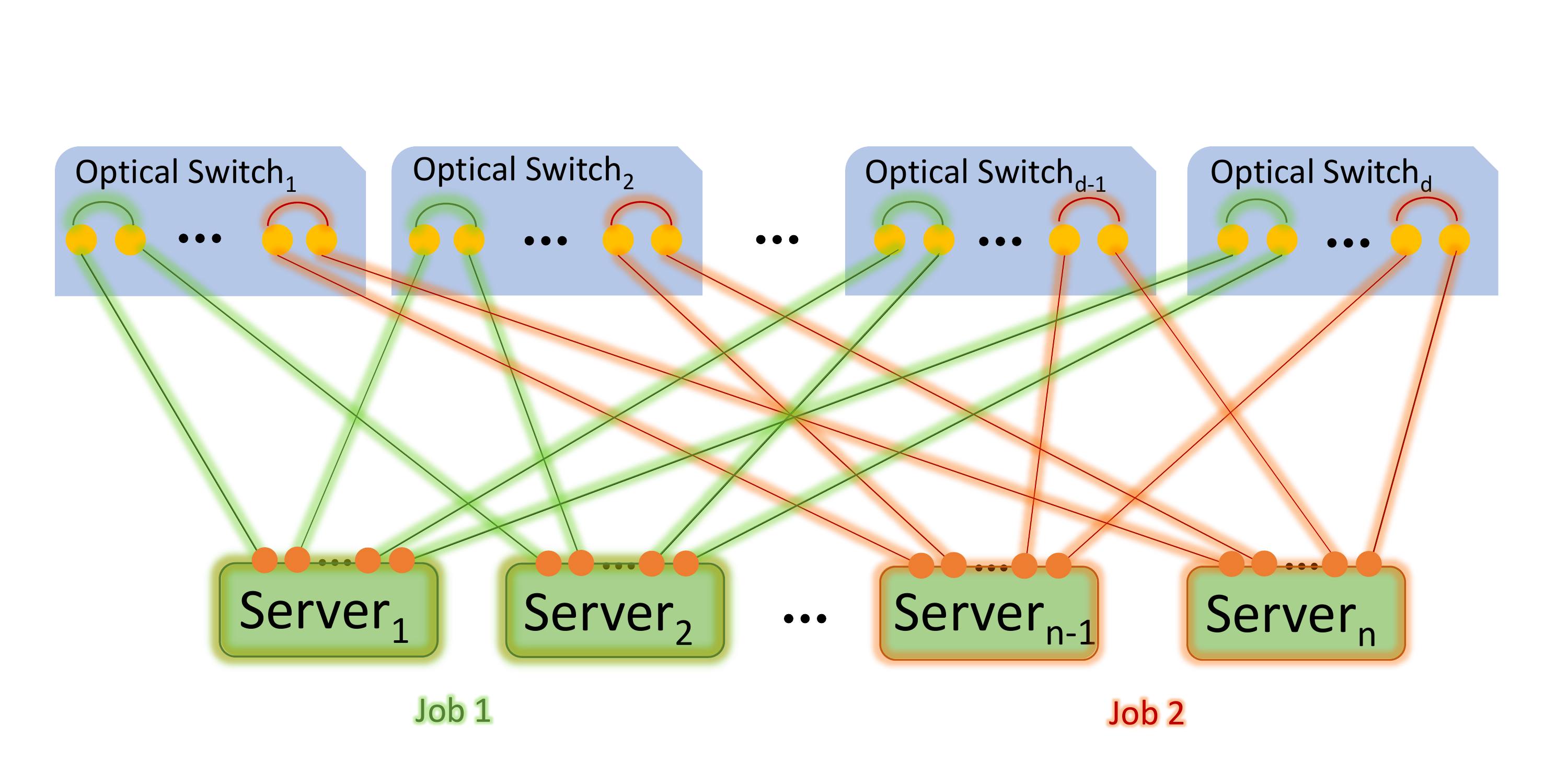}
\caption{Sharding \name cluster for two jobs.}
\label{fig:sharding}
\end{figure}

Section~\ref{sec:design} explained how \name can support multiple job sharing the cluster through sharding; here we provide a detailed explanation of how sharding works. Figure~\ref{fig:sharding} shows how a \name cluster is sharded to train two jobs together. In this scenario, the optical switches are configured in a way such that the green part (Server 1, 2 and their corresponding links) is completely disjointed from the red part (Server $n-1$, server $n$). The complete isolation ensures each job gets its dedicated resources, and benefits the performance (especially the tail latency) as shown in Section~\ref{sec:sims_multi_job}.

To start a job with $k$ servers, we need to reconfigure the interconnection between these $k$ servers before the job starts. This can be done quickly when OCSs are used, but when patch panels are used, there could be several minutes of delay before the job can start. To address this challenge, we use a look-ahead approach to pre-provision the next topology while current jobs are running. More specifically, we use a simple 1$\times$2 mechanical optical switch~\cite{one_by_two_switch} at each server's interface to choose between \textit{Active} and \textit{Look-ahead} ports. These 1$\times$2 switches are inexpensive (\$25) and have 0.73~dB optical loss measured in our prototype. 
Unlike optical splitters~\cite{fiber_spliter}, that incur 3~dB loss, these switches choose where to send light between their two output ports. 
We then connect the two ends of each 1$\times$2 switch to different patch panels, as shown in Figure~\ref{fig:flat_active_look_ahead}. 
As a result, a \name cluster with $n$ servers, each with $d$ interfaces, has $2d$ patch panels where each interface is split into two parts: Active and Look-ahead. At any point in time, only one end of each 1$\times$2 switch is participating in the active topology; the other end is pre-provisioning the topology for the next job.  Since the topology and parallelization strategy are calculated off-line, we already know the sequence of job arrivals and the number of servers required by each job. This design allows each server to participate in two independent topologies. Hence, when a set of servers uses one topology for a training job, \name pre-provisions the next topology, optimized for the next task by reconfiguring Look-ahead ports. 
Once all the servers for the new job are ready, \name immediately \textit{flips} to the new topology by reconfiguring the corresponding 1$\times$2 switches.%
\section{Model Configurations and Transfer Sizes}
\label{sec:dnn_details}

List~\ref{list:model_parameters} summarizes the parameters we used in our simulation and testbed. Model parameters and batch sizes are selected based on common values used in \net for simulations. For the prototype, we reduce parameter values and batch sizes to fit the models in our 12-node cluster.

\begin{List}[ht]
\begin{myframe}
\noindent \small\textbf{VGG:}\\
\footnotesize
\indent Batch/GPU: 64 (\S\ref{sec:sims_single_job}, \S\ref{sec:sims_multi_job}), 32 (\S\ref{sec:prototype})

\noindent \small\textbf{ResNet50:}\\
\footnotesize
\indent Batch/GPU: 128(\S\ref{sec:sims_single_job}), 20 (\S\ref{sec:prototype})

\noindent \small\textbf{BERT:}\\
\footnotesize
\indent Batch/GPU: 16 (\S\ref{sec:sims_single_job}, \S\ref{sec:sims_multi_job}), 2 (\S\ref{sec:prototype}) \\
\indent \#Trans. blks: 12 (\S\ref{sec:sims_single_job}), 6 (\S\ref{sec:sims_multi_job}, \S\ref{sec:prototype}) \\
\indent Hidden layer: 1024 (\S\ref{sec:sims_single_job}), 768 (\S\ref{sec:sims_multi_job}), 1024(\S\ref{sec:prototype}) \\
\indent Seq. length: 64 (\S\ref{sec:sims_single_job}), 256 (\S\ref{sec:sims_multi_job}), 1024(\S\ref{sec:prototype}) \\
\indent \#Attn. heads: 16 (\S\ref{sec:sims_single_job}), 6 (\S\ref{sec:sims_multi_job}), 16(\S\ref{sec:prototype}) \\
\indent Embed. size: 512 (\S\ref{sec:sims_single_job}, \S\ref{sec:sims_multi_job}, \S\ref{sec:prototype})

\noindent \small\textbf{DLRM:}\\
\footnotesize
\indent Batch/GPU: 128  (\S\ref{sec:sims_single_job}),[32$,\cdots,$2048] (\S\ref{eval:a2a}), 256 (\S\ref{sec:sims_multi_job}), [64$,\cdots,$512] (\S\ref{sec:prototype}) \\
\indent \#Dense layer: 8 (\S\ref{sec:sims_single_job}, \S\ref{sec:sims_multi_job}), 4 (\S\ref{sec:prototype}) \\
\indent Dense layer size: 2048 (\S\ref{sec:sims_single_job}), 1024 (\S\ref{sec:sims_multi_job}, \S\ref{sec:prototype})\\ 
\indent \#Dense feat. layer: 16  (\S\ref{sec:sims_single_job}, \S\ref{sec:sims_multi_job}), 8 (\S\ref{sec:prototype}) \\
\indent Feat. layer size: 4096 (\S\ref{sec:sims_single_job}), 2048 (\S\ref{sec:sims_multi_job}, \S\ref{sec:prototype}) \\
\indent Embed.: $128\times10^7$ (\S\ref{sec:sims_single_job}), $256\times 10^7$ (\S\ref{sec:sims_multi_job}), 32768$\times 10^5$ (\S\ref{sec:prototype}) \\
\indent \#Embed. tables: 64 (\S\ref{sec:sims_single_job}), 16 (\S\ref{sec:sims_multi_job}), 128 (\S\ref{eval:a2a}) , 12 (\S\ref{sec:prototype}) 

\noindent \small\textbf{CANDLE:}\\
\footnotesize
\indent Batch/GPU: 256 (\S\ref{sec:sims_single_job}, \S\ref{sec:sims_multi_job}), 10 (\S\ref{sec:prototype}) \\
\indent \#Dense layer: 8 (\S\ref{sec:sims_single_job}, \S\ref{sec:sims_multi_job}), 4 (\S\ref{sec:prototype}) \\
\indent Dense layer size: 16384 (\S\ref{sec:sims_single_job}), 4096 (\S\ref{sec:sims_multi_job}, \S\ref{sec:prototype}) \\
\indent \#Dense feat. layer: 16  (\S\ref{sec:sims_single_job}, \S\ref{sec:sims_multi_job}), 8 (\S\ref{sec:prototype}) \\
\indent Feat. layer size: 16384  (\S\ref{sec:sims_single_job}), 4096 (\S\ref{sec:sims_multi_job}, \S\ref{sec:prototype})

\noindent \small\textbf{NCF:}\\
\footnotesize
\indent Batch/GPU: 128 (\S\ref{sec:sims_single_job})\\
\indent \#Dense layer: 8 (\S\ref{sec:sims_single_job}) \\
\indent Dense layer size: 4096 (\S\ref{sec:sims_single_job})\\
\indent \#User embedding table (MF, MLP): 32, 32 (\S\ref{sec:sims_single_job}) \\
\indent \#User per table: $10^6$  (\S\ref{sec:sims_single_job}) \\
\indent \#Item embedding table (MF, MLP): 32, 32 (\S\ref{sec:sims_single_job}) \\
\indent \#Item per table: $10^6$  (\S\ref{sec:sims_single_job}) \\
\indent MF embedding dimension: 64  (\S\ref{sec:sims_single_job}) \\
\indent MLP embedding dimension: 128  (\S\ref{sec:sims_single_job})

\end{myframe}
\captionof{List}{DNN models used in our simulations and testbed. \label{list:model_parameters}}
\end{List}

In most workloads observed in \net, the size of \reduce transfers is larger than the size of \MP transfers for each iteration, because in most cases, it would not be worthwhile if \MP transfers were as large as \reduce transfers. Consider the DLRM example in Section~\ref{sec:mutability_reducetopo} with 20~GB embedding tables with double-precision floating parameters. If we were to distribute this embedding table using data parallelism, each server would need to send and receive 37.5~GB of data for the \reduce operation. On a 100~Gbps fabric, this would take 3 seconds by itself, but if we put it on one server, it would only need to transfer 32~MB/server (assume we have a per-server batch size of 8192; then, \MP traffic is calculated as 16 servers $\times$ 8192 samples/server $\times$ 512 activation per sample $\times$ 8 bytes per activation / 16 \text{servers} = 32~MB). We note that adding \textit{pipeline parallelism} can increase the amount of \MP traffic as it overlaps forward and backward passes. Efficient ways to pipeline batches remains an active research area~\cite{pipedream, gpipe} especially when hybrid parallelism is employed. Pure model parallelism creates another type of sparse traffic pattern where only accelerators with inter-layer dependencies need to communicate. Our \algo algorithm can support such communication patterns.

Conceptually, however, when the network bandwidth goes to infinity, other overheads in the system (e.g. CUDA kernel launch) will dominate the latency. In such cases, it might be beneficial to choose model parallelism instead of data parallelism, to reduce the amount of system overheads. In particular, prior work shows 10~Tbps Silicon Photonics links enable more aggressive model parallelism where the size of \MP traffic is significant~\cite{sip-ml}. \name's approach to distribute the degree between the \MP and \reduce sub-topologies enables us to accommodate this case as well.

\section{Algorithm Details}

\subsection{\algobf}
\label{app:algo_details}

\para{Using group theory to find \reduce permutations.} For a ring-\reduce group with $n$ servers labeled $S_0, ..., S_{n-1}$, a straightforward permutation  is $(S_0 \rightarrow S_1 \rightarrow S_2 \cdots \rightarrow S_{n-1} \rightarrow S_0)$. We denote this permutation by a ring generation rule as: $S_{i} \rightarrow S_{(i+1)~mod~n}$. Since the servers form a ring, the index of the starting server does not matter. For instance, these two rings are equivalent: $(S_0 \rightarrow S_1 \rightarrow S_2 \rightarrow S_3 \rightarrow S_0)$ and $(S_1 \rightarrow S_2 \rightarrow S_3 \rightarrow S_0 \rightarrow S_1)$.\footnote{Given that ring-\reduce is the dominant  \reduce collective, we describe our algorithms based on ring-\reduce. Appendix~\ref{app:algo_details} explains how to extend our algorithm to other \reduce communication collectives.}

We first provide the mathematical foundation of the ring permutation rule.

\begin{theorem}[Ring Generation]
\label{thm:ringgen}
For a cluster of $n$ nodes $V=\{S_0, S_1, \cdots, S_{N-1}\}$, all integer numbers $p < n$, where $p$ is co-prime with $n$ (i.e. $gcd(p ,n) = 1$) represent a unique ring-\reduce permutation rule.
\end{theorem}

\begin{proof}
Consider the integer modulo $n$ group with addition $\mathbb{Z}_n^+=\{0, 1, \cdots, (n-1)\}$. $\mathbb{Z}_n^+$ is a cyclic group. By the fundamental theorem of cyclic groups, $p$ is a generator of $\mathbb{Z}_n^+$ if and only if $gcd(p, n) = 1$. Hence we can cover the entire $\mathbb{Z}_n^+$ by repeatedly adding $p$ to itself. 

Now consider the graph $G_{\mathbb{Z}_n^+, p}=(V_{\mathbb{Z}_n^+}, E_p)$ where the set of vertices $V_{\mathbb{Z}_n^+}=\mathbb{Z}_n^+$ and $E_p=\{(a\times p, (a+1)\times p) \in V_{\mathbb{Z}_n^+}^2,\ a\in \mathbb{Z}_n^+\}$. The set $E_p$ forms a cycle on $G_{\mathbb{Z}_n^+, p}$. Now denote our cluster as $G=(V, E)$ where $V$ is defined as above and $E$ represents a set of directed links. Then $G_{\mathbb{Z}_n^+, p}$ is isomorphic to $G$, hence following the rule in $E_p$ we can define a valid ring in $G$. Furthermore, since $\forall p_i \neq p_j$ we can guarantee that $(0, p_i) \in E_{p_i}$ and $(0, p_j) \notin E_{p_i}$, and each $p_i$ is guaranteed to describe a unique ring.
\end{proof}

One way to extend our approach to other \reduce algorithms is to generalize \fancyperms (Algorithm~\ref{alg:fancyperms}) so that the $E_p$ described in theorem~\ref{thm:ringgen} simply represents a \textit{permutation} which we apply to the original node labeling, while keeping the edge relation, to create an isomorphic graph that describes the new \reduce topology.

\subsection{Bounding maximum hop count with \fancyperms}
\label{app:hopcount_bound}
In this section, we argue that fitting a geometric sequence for choosing permutation provides an approximately $O(d\sqrt[d]{n})$ bound for the maximum diameter of a cluster with $n$ nodes and degree $d$. Denote $x\equiv \sqrt[d]{n}$. We simplify the question to the following: given a contiguous set of numbers $\mathcal{N}=\{1, \dots, n\}$ and a set of numbers from the geometric sequence $S=\{x^{0}, x^{1}, \dots, x^{d-1}\}$, choose $h$ numbers (allow repetition) $s_1, \cdot, s_k$ from $S$ so that $m=\sum_{i=1}^h s_i$ for some $m\in \mathcal{N}$. Let $h=\kappa(m)$, find $\min_{m\in \mathcal{N}}\kappa(m)$.

Again for simplicity, assume $x\in \mathbb{Z}$. Then for a given $m\in \mathcal{N}$, we get the recursive relation $\kappa(m)=1+\kappa(m-x^i)$ where $i=\mathrm{argmax}_{i \leq d, x^i\leq m}$. $m=n-1$ gives the maximum $\kappa(n-1)=dx$. 

The problem above is simpler than the one in \name. In \name, $x$ is rarely an integer, and $S$ is a projection of the geometric sequence  $S=\{x^{0}, x^{1}, \dots, x^{d-1}\}$ onto the candidates (co-prime numbers with the size of a subset of node participating in \reduce). The intuition still holds.

Note that when $\sqrt[d]{n}<2$, it is advantageous to choose $x=2$ and spend less degree to create a geometric sequence with a ratio of at least 2. In this case, the $d$ factor becomes the actually used degree $d=\log_2 n$, and the bound holds at $O(\log_2~n)$.

\subsection{Coin Change Routing}

Consider servers $S_i$ and $S_j$ that need to exchange \reduce transfers but do not have a direct edge between them. We use a modified version of the classical coin change problem~\cite{coin_change} to find an efficient routing path (line~\ref{algo:coin_change}). In classical coin change, the goal is to find the minimum number of \textit{coins} that would sum to a certain \textit{total value}. Our ring generation rules enable us to treat the routing problem similarly. In particular, the $p$ values of \reduce permutations that have been selected in the \reduce sub-topology are the coin values, and the difference between server $i$ and $j$ indices ($(j-i)\%n$) is the target total value that we want to achieve. The only difference is that our problem runs with $modulo\ n$ arithmetic, as the server IDs wrap around in the ring structure. Algorithm~\ref{alg:coin_change} lists the pseudocode of \texttt{CoinChangeMod}.

\begin{algorithm}[t]
\footnotesize
\begin{algorithmic}[1]
    \Procedure{CoinChangeMod}{$n$, $G$}
        \IOComment{\textbf{Input} $n$: Total number of nodes}
        \IOComment{\textbf{Input} $G$: Network Topology}
        \IOComment{\textbf{Output} $R$: Routings}
        \textcolor{cyan}{\OldStatex \ \ \ \ \ \ \(\triangleright\) \textit{$R$ is the routing result}}
        \State $R$ \psass $\{\}$
        \textcolor{cyan}{\OldStatex \ \ \ \ \ \ \(\triangleright\) \textit{Acquire the set of ``coins" from the topology, }}
        \textcolor{cyan}{\OldStatex \ \ \ \ \ \ \(\triangleright\) \textit{which are the choices of Algorithm~\ref{alg:SelectPermutations}}}
        \State $C$ \psass \texttt{GetCoins($G$)}
        \For {$i \in [1, n-1]$}
            \textcolor{cyan}{\OldStatex \ \ \ \ \ \ \ \ \ \ \(\triangleright\) \textit{$curr\_dist$ denotes the ``distance" of a value }}
            \textcolor{cyan}{\OldStatex \ \ \ \ \ \ \ \ \ \ \(\triangleright\) \textit{(node distance) counted by number of ``coins"}}
            \State $curr\_dist[i]$ \psass $\infty$
            \textcolor{cyan}{\OldStatex \ \ \ \ \ \ \ \ \ \ \(\triangleright\) \textit{$curr\_bt$ record a back-trace of ``coins" to }}
            \textcolor{cyan}{\OldStatex \ \ \ \ \ \ \ \ \ \ \(\triangleright\) \textit{get to a value (node distance)}}
            \State $curr\_bt[i]$ \psass $\infty$
        \EndFor
        \For {$c\in C$}
            \State $curr\_dist[c]$ \psass $0$
            \State $curr\_bt[c]$ \psass $c$
        \EndFor
        \While {$curr\_dist$ has at least one $\infty$ in it}
            \For {$i \in [1, n-1]$}
                \State $new\_dist[i]$ \psass $curr\_dist[i]$
                \State $new\_bt[i]$ \psass $curr\_bt[i]$
                \For {$c\in C$}
                    \If {$curr\_dist[(i-c)\% n] < new\_dist[i]$}
                        \State $new\_dist[i]$ \psass $cur\_dist[(i-c)\%n] + 1$
                        \State $new\_bt[i]$ \psass $c$
                    \EndIf
                \EndFor
            \EndFor
            \State $curr\_dist$ \psass $new\_dist$
            \State $curr\_bt$ \psass $new\_bt$
        \EndWhile
        \textcolor{cyan}{\OldStatex \ \ \ \ \ \ \ \ \(\triangleright\) \textit{Construct the routing for each node distance from the back-trace}}
        \State $R$ \psass \texttt{GetRouteSeq($curr\_bt$)}
        \State \Return $R$
    \EndProcedure
\end{algorithmic}
\caption{\texttt{CoinChangeMod} pseudocode \label{alg:coin_change}}
\end{algorithm}
\subsection{OCS-reconfig Heuristic}
\label{app:reconfig_heuristic}

Algorithm~\ref{alg:uheuristic} describes the heuristic we use for OCS-reconfig. As mentioned in Section~\ref{sec:topoopt_algorithms}, our goals are ($i$) to have enough bandwidth for large transfer demands, ($ii$) while also minimizing the latency of indirect routing for nodes that do not have a direct link between them.

To achieve this goal in a reconfigurable interconnect, we propose a utility function that finds a balance between the two goals by maximizing the number of parallel links between high demand nodes but with a \textit{diminishing return}. More formally, assume a network topology is represented by graph $G=(V, E)$ and each node has degree $d$. We define $L(i,j)$ to be the number of parallel links between node-pair $(i, j)$. 
Let $T(i,j)$ be the amount of unsatisfied traffic demand. We define a topology $G$'s utility function as follows:

\begin{equation}
\label{eq:utility}
\begin{aligned}
Utility(G) = \sum\limits_{\{i,j\} \in E} T(i,j) \times Discount(L(i, j))
\end{aligned}
\end{equation}

The $Discount$ function can be defined in different ways; in Algorithm~\ref{alg:uheuristic} and Algorithm~\ref{alg:topo_finder}'s \MP construction, we use 
\begin{equation}
\label{eq:exp_discount}
\begin{aligned}
Discount(l) = \sum\limits_{x=1}^{l}2^{-x}
\end{aligned}
\end{equation}
to reduce the utility of additional links exponentially. We can also explore other discount scaling, such as linear or factorial functions.

When the fabric is reconfigurable (as in OCS-reconfig), we collect the unsatisfied traffic demand every 50~ms and run Algorithm~\ref{alg:uheuristic} to decide the new network topology. After the new topology is computed, we pause all the flows for 10~ms representing the reconfiguration delay of the OCS, apply the new topology, and then resume the flows with one or more corresponding physical links across the flow source and destination. The two-edge replacement algorithm from OWAN~\cite{owan} in line~\ref{algo:twoer} ensures the topology is connected, when we enable host-based forwarding.

\begin{algorithm}[t]
\footnotesize
\begin{algorithmic}[1]
    \Procedure{OCS-reconfig}{$V$, $T$, $d$, $L$}
        \IOComment{\textbf{Input} $V$: Nodes in the network}
        \IOComment{\textbf{Input} $T$: Unsatisfied traffic demand matrix}
        \IOComment{\textbf{Input} $d$: Node degree limit}
        \IOComment{\textbf{Input} $L$: Number of links between ordered node-pair, initially zero}
        \IOComment{\textbf{Output} $E$: Allocated links, initially empty}
        \textcolor{cyan}{\OldStatex \ \ \ \ \ \ \(\triangleright\) \textit{Initially, $E$ is empty}}
        \State $E$ \psass $\{\}$
        \textcolor{cyan}{\OldStatex \ \ \ \ \ \ \(\triangleright\) \textit{Initially, each node has $d$ available tx and rx interfaces}}
        \For {$v \in V$}
            \State $available_{tx}[v]$ \psass $d$
            \State $available_{rx}[v]$ \psass $d$
        \EndFor
        \textcolor{cyan}{\OldStatex \ \ \ \ \ \ \(\triangleright\) \textit{Create new links according to the demand list}}
        \While{$\exists i, j < |V|: i \neq j, available_{tx}[v_i] > 0,available_{rx}[v_j] > 0$}
            \textcolor{cyan}{\OldStatex \ \ \ \ \ \ \ \ \ \ \ \(\triangleright\) \textit{allocate a direct connection for the highest demand pair}}
            \State $(v_1, v_2)$ \psass node-pair with highest demand in $T$
            \State $e$ \psass \textbf{NewLink($v_1$, $v_2$)} 
            \State $E$ \psass $E \cup \{e\}$
            \textcolor{cyan}{\OldStatex \ \ \ \ \ \ \ \ \ \ \ \(\triangleright\) \textit{Increment the number of parallel links from $v_1$ to $v_2$}}
            \State $L(v_1, v_2)$ \addeq 1 
            \textcolor{cyan}{\OldStatex \ \ \ \ \ \ \ \ \ \ \ \(\triangleright\) \textit{Scale the demand down by the number of links}}
            \State $T(v_1, v_2)$ \muleq $1/2$
            \textcolor{cyan}{\OldStatex \ \ \ \ \ \ \ \ \ \ \ \(\triangleright\) \textit{Update available interfaces}}
            \For{$v \in (v_1, v_2)$} 
                \State $available_{tx}[v_1]$ \subeq 1
                \State $available_{rx}[v_2]$ \subeq 1
                \textcolor{cyan}{\OldStatex \qquad \ \ \ \ \ \ \ \ \ \ \ \ \ \ \ \ \(\triangleright\) \textit{Stop considering nodes with zero available interfaces}}
                \If{$available_{tx}[v_1]$ \eqeq 0}
                    \For{$u \in V$}
                        \State Remove $(v_1, u)$'s entry from $T$
                    \EndFor 
                \EndIf 
                \If{$available_{rx}[v_2]$ \eqeq 0}
                    \For{$u \in V$}
                        \State Remove $(u, v_2)$'s entry from $T$
                    \EndFor 
                \EndIf 
            \EndFor
        \EndWhile
        \textcolor{cyan}{\OldStatex \ \ \ \ \ \ \(\triangleright\) \textit{Ensure the network graph is connected}}
        \State{\texttt{2-EdgeReplacement}($E$, $T$)} \label{algo:twoer}
        \textcolor{cyan}{\OldStatex \ \ \ \ \ \ \(\triangleright\) \textit{Updte route for host-based forwarding}}
        \State{\texttt{UpdateRoute}($E$)} 
        \State \Return $E$ 
  \EndProcedure
\end{algorithmic}
\caption{OCS-reconfig pseudocode \label{alg:uheuristic}}
\end{algorithm}

\section{Modifications to SiP-ML}
\label{app:sipml}

Since SiP-ML's SiP-Ring proposal is based on a physical ring topology, its reconfiguration algorithm has several constraints on wavelength allocation for adjacent nodes. Given that \name's physical topology is not a ring, directly applying SiP-Ring's optimization  using the original C++ code causes SiP-ML to perform extremely poorly in our setup. To give SiP-ML a leg up, we observe that its formulation tries to optimize a utility function very similar to Equation~\ref{eq:utility} without the $Discount$ part (i.e. $Discount = 1$), but with an integer liner program (ILP). While an ILP gives the optimal solution, its runtime makes it prohibitive for the amount of simulation parameters we explore. Therefore, we substitute the ILP with Algorithm~\ref{alg:uheuristic} with $Discount = 1$, a heuristic that tries to achieve a similar goal. 

Note that the SiP-ML paper has another design called SiP-OCS, which is similar architecturally to \name. In the paper, SiP-OCS is proposed as a one-shot reconfiguration approach due to the long reconfiguration latency of 3D-MEMS based OCSs. 

\section{Cost of Network Components}
\label{sec:cost_table}

\begin{table}[t]
\scriptsize
\centering
\renewcommand{\arraystretch}{0.92} 
\linespread{1.05}\selectfont\centering
    \begin{tabular}{|p{0.95cm}|p{0.7cm}|p{0.72cm}|p{0.8cm}|p{0.73cm}|p{0.7cm}|p{0.7cm}|}
    \hline
        Link bandwidth & Tran-sceiver (\$) & NIC (\$) & Electrical switch port (\$) & Patch panel port (\$) & OCS port (\$) & 1$\times$2 switch (\$)\\ \hline
        10~Gbps     & 20~\cite{ts10}  & 185~\cite{nic1_25} & 94~\cite{sw10} & 100~\cite{telescent} & 520~\cite{polatis} & 25~\cite{one_by_two_switch} \\\hline
        25~Gbps     & 39~\cite{ts25new}  & 185~\cite{nic1_25}  & 144~\cite{sw25} & 100~\cite{telescent}  & 520~\cite{polatis} & 25~\cite{one_by_two_switch} \\\hline
        40~Gbps     & 39~\cite{ts40new}  & 354~\cite{nic1_40} & 144~\cite{sw40}& 100~\cite{telescent}  & 520~\cite{polatis} & 25~\cite{one_by_two_switch} \\\hline
        100~Gbps    & 99~\cite{ts100new} & 678~\cite{nic1_100}  & 187~\cite{sw100} & 100~\cite{telescent}  & 520~\cite{polatis}& 25~\cite{one_by_two_switch} \\\hline
        200~Gbps\footnote{200~G transceivers and switch ports are estimated as 2$\times$ 100G cost.}  & 198~\cite{ts100new} & 815~\cite{nic2_100} & 374~\cite{sw100} & 100~\cite{telescent}  & 520~\cite{polatis} & 25~\cite{one_by_two_switch} \\\hline
    \end{tabular}
    \caption{Cost of network components.}
    \label{tab:cost_model}
\end{table}

Table~\ref{tab:cost_model} lists the cost of network components we use in  Section~\ref{sec:cost_model}, namely NICs, transceivers, fibers, electrical switches, patch panels, and optical switches. The cost of transceivers, NICs, and electrical switch ports is based on the lowest available prices in official retailer websites~\cite{fs, colfax}. Note that for 200~Gbps, we use more 100~Gbps ports and fibers, because they were less expensive than high-end 200~Gbps and 400~Gbps components, or their price was not available. To estimate the cost of electrical switch ports, we consider Edgecore bare metal switches with L3 switching and maximum number of ports to amortize the per port cost. The cost of NICs is taken from the Mellanox ConnectX series, and we consider two 2-port NICs as one 4-port NIC. We obtain the cost of the patch panel, OCS, and 1$\times$2 optical switch directly from their suppliers, Telescent~\cite{telescent} and Polatis~\cite{polatis} (with 40\% discount). The cost of transceivers matches that reported in Sirius~\cite{sirius}.

To compute the network cost of \LBE and \SBE, we consider number of nodes in a full bisection bandwidth \fattree. For example, a standard $k=8$ \fattree has 80 switches with 64 ports, or 640 switch ports in total, in addition to 1 NIC per host and one transceiver per NIC and switch port. A \name system of 128 nodes with degree $d$ uses $128\times d$ NICs and transceivers, but $128\times 2\times d$ patch panel ports because of the look-ahead design. Note that the cost of optical components stays constant as link bandwidth increases, an inherent advantage of optics. Following prior work, we estimate the cost of fiber optics cables as 30~cents per meter~\cite{projector} and select each fiber's length from a uniform distribution between 0 and 1000~meters~\cite{rail}. We calculate the cost of \name based on $2d$ patch panels and 1$\times$2 switches at each link to support its look-ahead design (\S\ref{app:sharding}). OCS-reconfig's cost is based on $d$ OCSs connected to all servers in a flat topology. 

\section{Impact of Server Degree on \name's Performance}
\label{sec:impact_of_degree}

\begin{figure*}[t]
    \centering
\includegraphics[width=\textwidth]{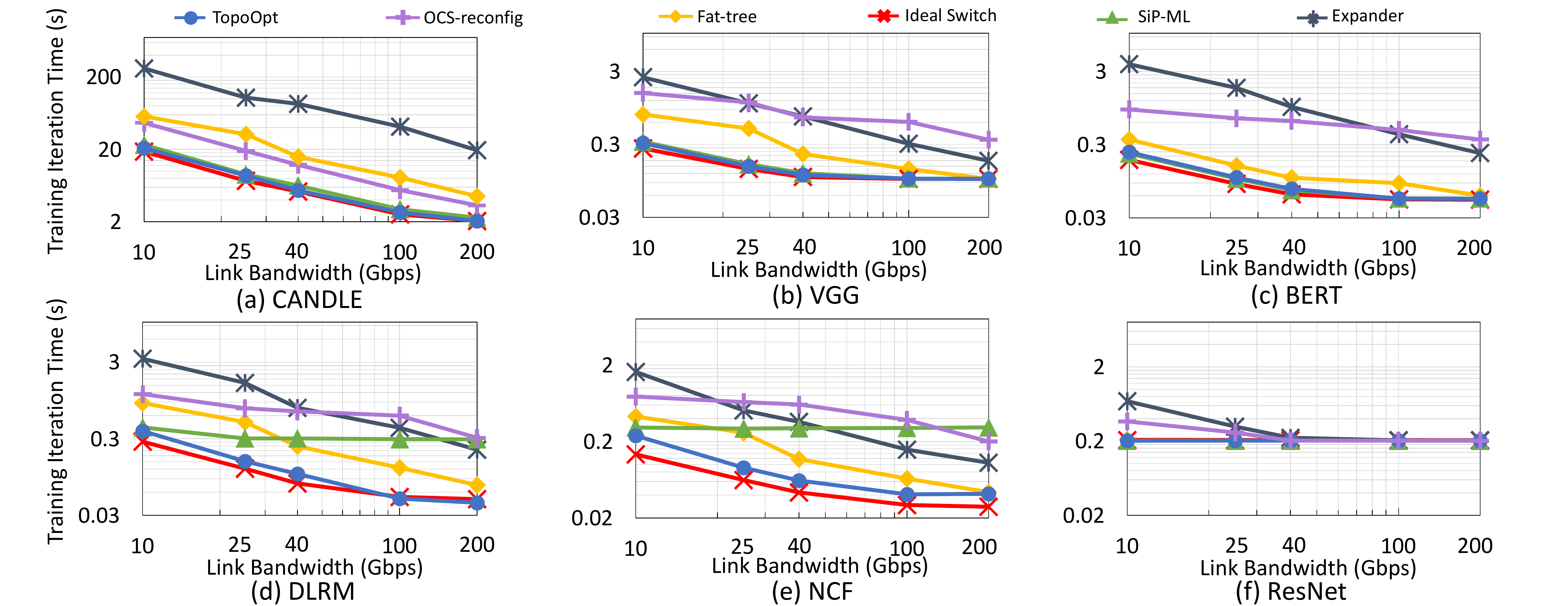}
\caption{Dedicated cluster of 128 servers ($d$ = 8).}
\label{fig:128_d8_single_job}
\end{figure*}

Figure~\ref{fig:128_d8_single_job} shows the same setting as Figure~\ref{fig:128_d4_single_job} except that each server has a degree of eight ($d=8$). The results show a similar trend: even though per server bandwidth has increased, the behavior of different network architectures remains consistent.

\begin{figure}[t]
\centering
\includegraphics[width=0.99\columnwidth]{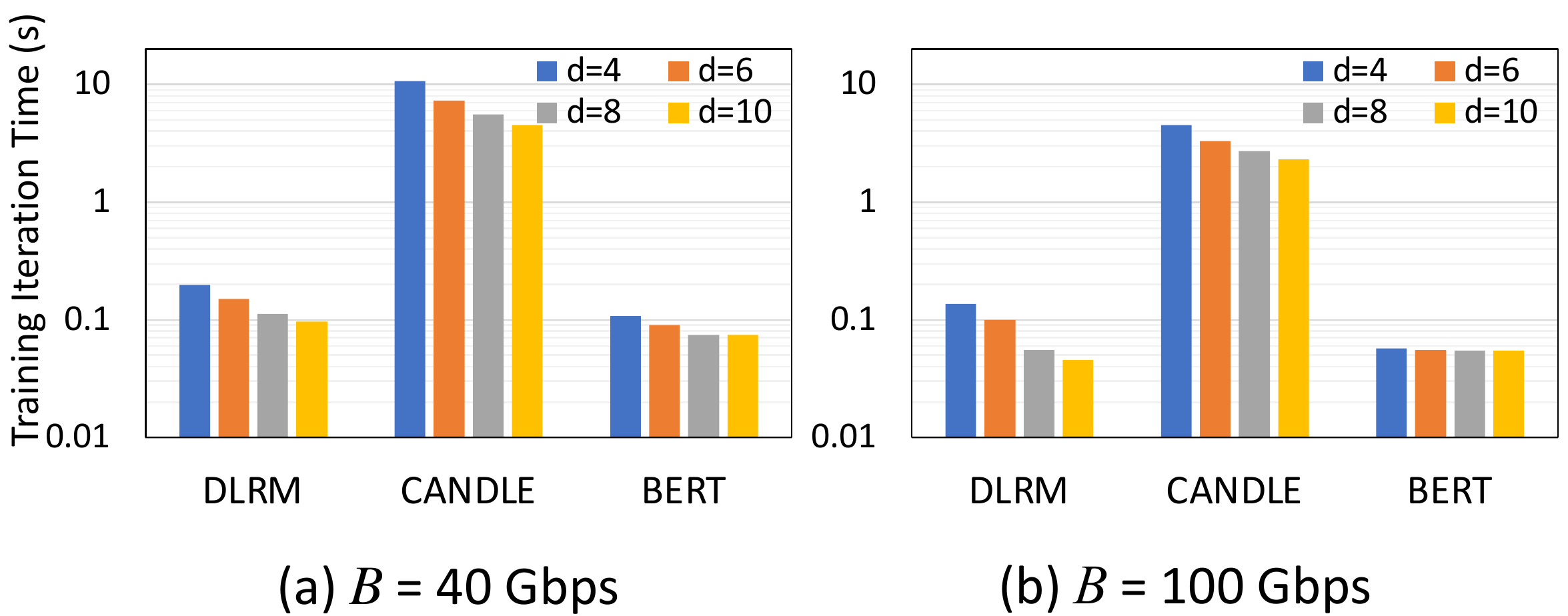}
\caption{Impact of server degree ($d$) on performance.}
\label{fig:degree_impact}
\end{figure}

Next we do a sensitivity analysis of impact of server degree $d$ on \name's performance. 
Specifically, we vary the degree of each server in \name for two link bandwidths: 40~Gbps and 100~Gbps. Figure~\ref{fig:degree_impact} shows the trend for different DNN models. Both DLRM and CANDLE are network-heavy;  therefore, they benefit more from the additional bandwidth obtained by increasing $d$. CANDLE's improvement is almost linear as degree goes up, as the strategy is closer to data parallel and the amount of bandwidth available to \reduce operation increases linearly as well. In the case of DLRM, we  observe a super-liner scaling when $B$ = 100~Gbp because DLRM has one-to-many and many-to-one \MP transfers which require a low hop count in the topology. As we increase $d$, \algo is able to find network topologies with much lower diameter, consequently benefiting the performance by both increasing bandwidth and reducing hop-count for \MP transfers. Finally, BERT is mostly compute bound at higher bandwidth; hence, increasing the server degree and bandwidth per node has marginal impact on its iteration time.

\section{Enabling Host-based Forwarding in RDMA}
\label{app:rdma_forwarding}

\balance

\begin{figure}[t]
\centering
\includegraphics[width=0.99\columnwidth]{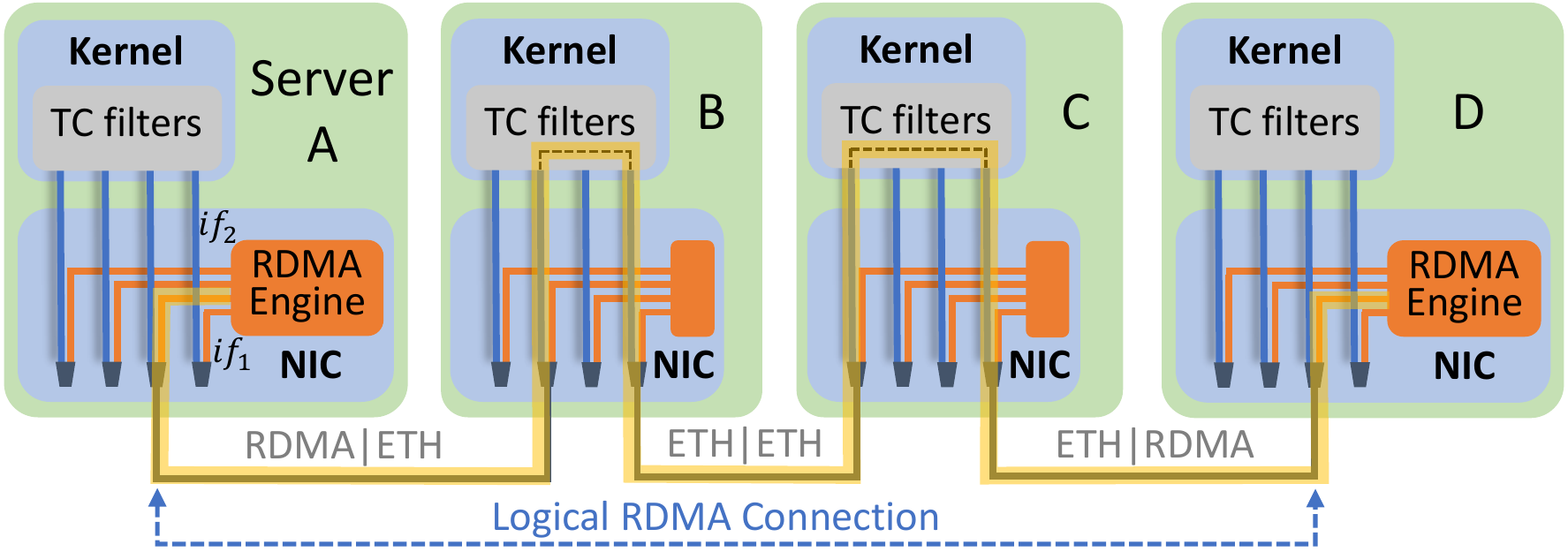}
\caption{Host-based RDMA forwarding to create a logical RDMA connection between end hosts.} 
\label{fig:forwarding}
\end{figure}

To support a multihop \name interconnect using host-based forwarding, we enable RDMA RoCEv2 forwarding on all our HP NICs. 
RoCEv2 is an implementation of RDMA on top of UDP/IP protocol, by utilizing a particular UDP port (4791) and encapsulating an InfiniBand (IB) data packet. Hence, each RoCEv2 packet can be routed with its source and destination IP addresses. However, host-based forwarding is challenging in RDMA protocol, as the packet processing and memory access are offloaded to the NIC, and the host does not have access to individual packets. More precisely, if a packet's IP destination IP address does not match the NIC's IP address, the RDMA engine silently drops the  packet. 

To address this issue, we collaborated with engineers from Marvell, the provider of the firmware and driver for our HP NICs. The solution that came out of our collaboration does not require proprietary software or firmware, and is applicable to commodity NICs with the same ASIC. We will release our scripts publicly. At a high-level, we use a feature called \textit{NPAR}, or \textit{network partitioning}. It allows us to split each 25~Gbps physical interface into two \textit{logical interfaces} in the hardware level: $if_1$ and $if_2$, as shown in the right-most port of server A in Figure~\ref{fig:forwarding}. $if_1$ is a normal RDMA interface, where the RDMA engine of the NIC bypasses the kernel, and it has an IP address. This enables the upper layer software to consider $if_1$ as a normal RDMA interface. However, $if_2$ does not have an IP address and RDMA is disabled. $if_2$ has a different MAC address from $if_1$, and we use this address to split the traffic across $if_1$ and $if_2$. The traffic that needs to be forwarded uses the MAC address of $if_2$ and hence is delivered to the host networking stack instead of NIC's RDMA engine.

Furthermore, we establish a set of \texttt{iproute}, \texttt{arp}, and \texttt{tc flower} rules in Linux to enable the proper forwarding of packets. If two servers are directly connected, such as the third port of server A and the second port of Server B in Figure~\ref{fig:forwarding}, we only need to indicate the outgoing interface on each of these servers. RDMA engines will handle the communication. However, for the connection between server A and D, we set the \texttt{iproute} and \texttt{arp} tables on server A and server D to dictate which port the traffic should go out, as well as the proper MAC address of the next server in the forwarding chain. In this case, the packets are delivered to the kernel. Then, on servers B and C, we set the \texttt{tc flower} rules to forward the packets to the next server with the proper MAC address. In these \texttt{tc flower} rules, we look-up the final destination IP and assert the routing that was computed by our algorithm. 

\para{Walk-through of an example of a packet going from server A to server D.} In Figure~\ref{fig:forwarding}, the RDMA engine of server A assumes server D is connected on the third port. It uses the kernel's routing tables for the destination MAC address, which is set to the MAC address of $if_2$ of the second port on server B. Therefore, a packet which starts as an RDMA packet of server A is treated as an Ethernet packet when it arrives at server B, and goes to server B's kernel. In the kernel, based on the packet's final destination IP of server D, server B redirects the packet to the fourth port, with destination MAC address set to $if_2$ of server C. In this connection, the packet is treated as a normal Ethernet packet. Finally, on server C, the kernel rewrites the destination MAC address to that of $if_1$ on the third port of server D, and redirects it to that port. In this connection, the outgoing Ethernet packet is considered an RDMA packet because of the destination MAC address. For the reverse connection from server D to A, the same process happens in reverse, to support a bidirectional connection. 

With these forwarding rules, we construct logical RDMA connections between all pairs of servers. Upper layer communication libraries such as NCCL requires all-to-all connectivity, and they will utilize these connections. We also modify NCCL to be topology-aware, as certain pairs of servers are only connected through specific ports.

Compared to native point-to-point RDMA, this approach takes a performance penalty. Our experiments indicate the overhead is negligible when the amount of forwarded traffic is small. Our NICs currently support TCP forwarding offload. With firmware and driver modifications or future versions of the NICs, they will also support RDMA forwarding offload. This will further reduce the overhead of our approach.
\label{lastpage} 

\end{document}